\newif\ifabstract
\abstracttrue
\abstractfalse 
\newif\iffull
\ifabstract \fullfalse \else \fulltrue \fi

\ifabstract
\documentclass[twoside,leqno,twocolumn]{article}  
\usepackage{ltexpprt} 
\fi
\iffull
\documentclass{article}  
\usepackage{fullpage}
\fi

\usepackage{amsfonts}
\usepackage{amssymb}
\usepackage{amstext}
\usepackage{amsmath}
\usepackage{xspace}
\usepackage{graphicx} \pdfminorversion=5
\usepackage{url}
\usepackage{subfigure}

\newcommand{\connect}{\leadsto}

        {\hspace*{\fill}$\Box$\par\vspace{4mm}}

\newcommand{\opt}{\mbox{\sf OPT}}

\newcommand{\set}[1]{\left\{ #1 \right\}}
\newcommand{\sse}{\subseteq}

\newcommand{\tset}{{\mathcal T}}

\newcommand{\pset}{{\mathcal{P}}}

\newcommand{\bset}{{\mathcal{B}}}
\newcommand{\aset}{{\mathcal{A}}}
\newcommand{\cset}{{\mathcal{C}}}
\newcommand{\fset}{{\mathcal{F}}}

\newcommand{\rset}{{\mathcal{R}}}
\newcommand{\zset}{{\mathcal{Z}}}

\newcommand{\sset}{{\mathcal{S}}}
\newcommand{\gset}{{\mathcal{G}}}

\newcommand{\nots}{\overline S}

\newcommand{\be}{\begin{enumerate}}
\newcommand{\ee}{\end{enumerate}}
\newcommand{\bd}{\begin{description}}
\newcommand{\ed}{\end{description}}
\newcommand{\bi}{\begin{itemize}}
\newcommand{\ei}{\end{itemize}}

\ifabstract

\newtheorem{claim}{Claim}[section] 
\newtheorem{observation}{Observation}[section]
\newtheorem{remark}{Remark}[section]
\fi

\iffull
\newtheorem{remark}{Remark}
\newtheorem{lemma}{Lemma}
\newtheorem{theorem}{Theorem}
\newtheorem{observation}{Observation}
\newtheorem{corollary}{Corollary}
\newtheorem{claim}{Claim}
\newenvironment{Definition}{{\bf Definition}: }{}
\newtheorem{proposition}{Proposition}
\newenvironment{proof}{\par \smallskip{\bf Proof:}}{\hfill\stopproof}
\def\stopproof{\square}
\def\square{\vbox{\hrule height.2pt\hbox{\vrule width.2pt height5pt \kern5pt
\vrule width.2pt} \hrule height.2pt}}
\fi

\newcommand{\qed}{\hfill\vbox{\hrule height.2pt\hbox{\vrule width.2pt height5pt \kern5pt
\vrule width.2pt} \hrule height.2pt}}

\renewcommand{\phi}{\varphi}

\newcommand{\poly}{\operatorname{poly}}

\newcommand{\MP}{\mbox{\sf Minimum Planarization}\xspace}
\newcommand{\MCN}{\mbox{\sf Minimum Crossing Number}\xspace}

\newcommand{\optmp}[1]{\mathsf{OPT}_{\mathsf{MP}}(#1)}
\newcommand{\optcro}[1]{\mathsf{OPT}_{\mathsf{cr}}(#1)}
\newcommand{\optcrosq}[1]{\mathsf{OPT}^{2}_{\mathsf{cr}}(#1)}

\newcommand{\pcro}{\mathsf{pcr}}

\newcommand{\cro}{\mathsf{cr}}
\newcommand{\irreg}{\mathsf{IRG}}
\newcommand{\dmax}{d_{\mbox{\textup{\footnotesize{max}}}}}

\newcommand{\G}{{\mathbf{G}}}
\renewcommand{\H}{{\mathbf{H}}}
\newcommand{\bphi}{{\boldsymbol{\varphi}}}
\newcommand{\bpsi}{{\boldsymbol{\psi}}}
\newcommand{\ndew}{\mathsf{ndew}}

\iffull
\fi

\begin{document}

\title{\Large On Graph Crossing Number and Edge Planarization\iffull\footnote{Extended abstract to appear in SODA 2011}\fi}
\author{Julia Chuzhoy\thanks{Toyota Technological Institute, Chicago, IL
60637. Email: {\tt cjulia@ttic.edu}. Supported in part by NSF CAREER award CCF-0844872.}
 \and Yury Makarychev\thanks{Toyota Technological Institute, Chicago, IL
60637. Email: {\tt yury@ttic.edu}}\and Anastasios Sidiropoulos\thanks{Toyota Technological Institute, Chicago, IL
60637. Email: {\tt tasos@ttic.edu}}}
\date{}
\maketitle

\begin{abstract}
Given an $n$-vertex graph $G$, a \emph{drawing} of $G$ in the plane is a mapping of its vertices into points of the plane, and its edges into continuous curves, connecting the images of their endpoints. 
A \emph{crossing} in such a drawing is a point where two such curves intersect. In the \MCN problem, the goal is to find a drawing of $G$ with minimum number of crossings. The value of the optimal solution, denoted by $\opt$, is called the graph's \emph{crossing number}. This is a very basic problem in topological graph theory, that has received a significant amount of attention, but is still poorly understood algorithmically. The best currently known efficient algorithm produces drawings with $O(\log^2 n) \cdot (n + \opt)$  crossings on bounded-degree graphs, while only a constant factor hardness of approximation is known. A closely related problem is \MP, in which the goal is to remove a minimum-cardinality subset of edges from $G$, such that the remaining graph is planar.

Our main technical result establishes the following connection between the two problems: if we are given a solution of cost $k$ to the \MP problem on graph $G$, then we can efficiently find a drawing of $G$ with at most $\poly(d)\cdot k\cdot (k+\opt)$ crossings, where $d$ is the maximum degree in $G$. This result implies an $O(n\cdot \poly(d)\cdot \log^{3/2}n)$-approximation for \MCN, as well as improved algorithms for
special cases of the problem, such as, for example, $k$-apex and bounded-genus graphs.
\end{abstract}

\section{Introduction}

A \emph{drawing} of a graph $G$ in the plane is a mapping, in which every vertex is mapped into a point of the plane, and every edge into a continuous curve connecting the images of its endpoints.
We assume that no three curves meet at the same point (except at their endpoints), and that no curve contains an image of any vertex other than its endpoints.
A \emph{crossing} in such a drawing is a point where the drawings of two edges intersect, and the \emph{crossing number} of a graph $G$, denoted by $\optcro{G}$, is the smallest integer $c$, such that $G$ admits a drawing with $c$ crossings.
In the \MCN problem, given an $n$-vertex graph $G$,  the goal is to find a drawing of $G$ in the plane that minimizes the number of crossings.
A closely related problem is \MP, in which the goal is to find a minimum-cardinality subset $E^*$ of edges, such that the graph $G\setminus E^*$ is planar. The optimal solution cost of the \MP problem on graph $G$ is denoted by $\optmp{G}$, and it is easy to see that $\optmp{G}\leq \optcro{G}$.

The problem of computing the crossing number of a graph was first considered by Tur\'{a}n \cite{turan_first}, who posed the question of estimating the crossing number of the complete bipartite graph.
Since then, the problem has been a subject of intensive study.
We refer the interested reader to the expositions by Richter and Salazar \cite{richter_survey}, Pach and T\'{o}th \cite{pach_survey}, and Matou\v{s}ek \cite{matousek_book}, and the extensive bibliography maintained by Vrt'o \cite{vrto_biblio}. 
Despite the enormous interest in the problem, and several breakthroughs over the last four decades, there is still very little understanding of even some of the most basic questions. For example, to the time of this writing, the crossing number of $K_{13}$ remains unknown.

Perhaps even more surprisingly, the \MCN problem remains poorly understood algorithmically.
In their seminal paper, Leighton and Rao \cite{LR}, combining their algorithm for balanced separators with the framework of Bhatt and Leighton \cite{bhatt84}, gave the first non-trivial algorithm for the problem. Their algorithm computes a drawing with at most $O(\log^4 n) \cdot (n + \optcro{G})$ crossings, when the degree of the input graph is bounded.
This algorithm was later improved to $O(\log^3 n) \cdot (n+\optcro{G})$ by Even et al.~\cite{EvenGS02}, and the new approximation algorithm for the Balanced Cut problem by Arora, Rao and Vazirani~\cite{ARV} improves it further to $O(\log^2 n) \cdot (n+\optcro{G})$, thus implying an $O(n \cdot \log^2 n)$-approximation for \MCN on bounded-degree graphs. Their result can also be shown to give an $O(n\cdot\log^2n\cdot \poly(\dmax))$-approximation for general graphs with maximum degree $\dmax$. We remark that in the worst case, the crossing number of a graph can be as large as $\Omega(n^4)$, e.g.~for the complete graph.

On the negative side, computing the crossing number of a graph was shown to be NP-complete by Garey and Johnson \cite{crossing_np_complete}, and it remains NP-complete even on cubic graphs~\cite{Hlineny06a}.
Combining the reduction of \cite{crossing_np_complete} with the inapproximability result for Minimum Linear Arrangement~\cite{Ambuhl07}, we get that there is no PTAS for the \MCN problem unless problems in NP have randomized subexponential time algorithms.
Interestingly, even for the very restricted special case, where there is an edge $e$ in $G$, such that $G\setminus e$ is planar, the \MCN problem still remains NP-hard~\cite{cabello_edge}. However, an $O(\dmax)$-approximation algorithm is known for this special case, where $\dmax$ is the maximum degree in $G$~\cite{HlinenyS06}.
Therefore, while the current techniques cannot exclude the existence of a constant factor approximation for \MCN, the state of the art gives just an $O(n \cdot\poly(\dmax)\cdot \log^2 n)$-approximation algorithm. 

In this paper, we provide new technical tools that we hope will lead to a better understanding of the \MCN problem. We also obtain improved approximation algorithms for special cases where the optimal solution for the \MP problem is small or can be approximated efficiently.

\subsection{Our Results}
Our main technical result establishes the following connection between the \MCN and the \MP problems:

\begin{theorem}\label{thm:main}
Let $G=(V,E)$ be any $n$-vertex graph with maximum degree $\dmax$, and suppose we are given a subset $E^*\sse E$ of edges, $|E^*|=k$, such that $H=G\setminus E^*$ is planar. Then we can efficiently find a drawing of $G$ with at most $O\left(\dmax^3\cdot k\cdot (\optcro{G}+k)\right )$ crossings.
\end{theorem}
\begin{remark}
Note that there always exists a subset $E^*$ of edges of size $\optmp{G}\leq \optcro{G}$, such that $H=G\setminus E^*$ is planar.
However, in Theorem~\ref{thm:main}, we do not assume that $E^*$ is the optimal solution to the \MP problem on $G$, and we allow
$k$ to be greater than $\optcro{G}$.
\end{remark}

A direct consequence of Theorem~\ref{thm:main} is that an $\alpha$-approximation algorithm for \MP would immediately give an algorithm for drawing any graph $G$ with $O(\alpha^2 \cdot \dmax^3\cdot \optcrosq{G})$ crossings. We note that while this connection between \MP and \MCN looks natural, it is possible that in the optimal solution $\phi$ to the \MCN problem on $G$, the induced drawing of the planar subgraph $H=G\setminus E^*$ is not planar, that is, the edges of $H$ may have to cross each other (see Figure~\ref{fig: example} for an example).

Theorem~\ref{thm:main} immediately implies a slightly improved algorithm for \MCN. In particular, while we are not aware of any approximation algorithms for the \MP problem, the following is an easy consequence of the Planar Separator theorem of Lipton and Tarjan~\cite{planar-separator}:
\begin{theorem}\label{thm:sqrt n}
There is an efficient $O(\sqrt {n\log n}\cdot \dmax)$-approximation algorithm for {\sf Minimum Planarization}.
\end{theorem}
The next corollary then follows from combining Theorems~\ref{thm:main} and \ref{thm:sqrt n}, and using the algorithm of~\cite{EvenGS02}.
\begin{corollary}\label{corollary:result for general graphs}
There is an efficient algorithm, that, given any $n$-vertex graph $G$ with maximum degree $\dmax$, finds a drawing of $G$ with at most $O(n\log n\cdot \dmax^5)\optcrosq{G}$ crossings. Moreover,  there is an efficient $O(n\cdot\poly(\dmax)\cdot \log^{3/2} n)$-approximation algorithm for \MCN.
\end{corollary} 

\begin{figure}[h]
\begin{center}
\scalebox{0.25}{\rotatebox{0}{\includegraphics{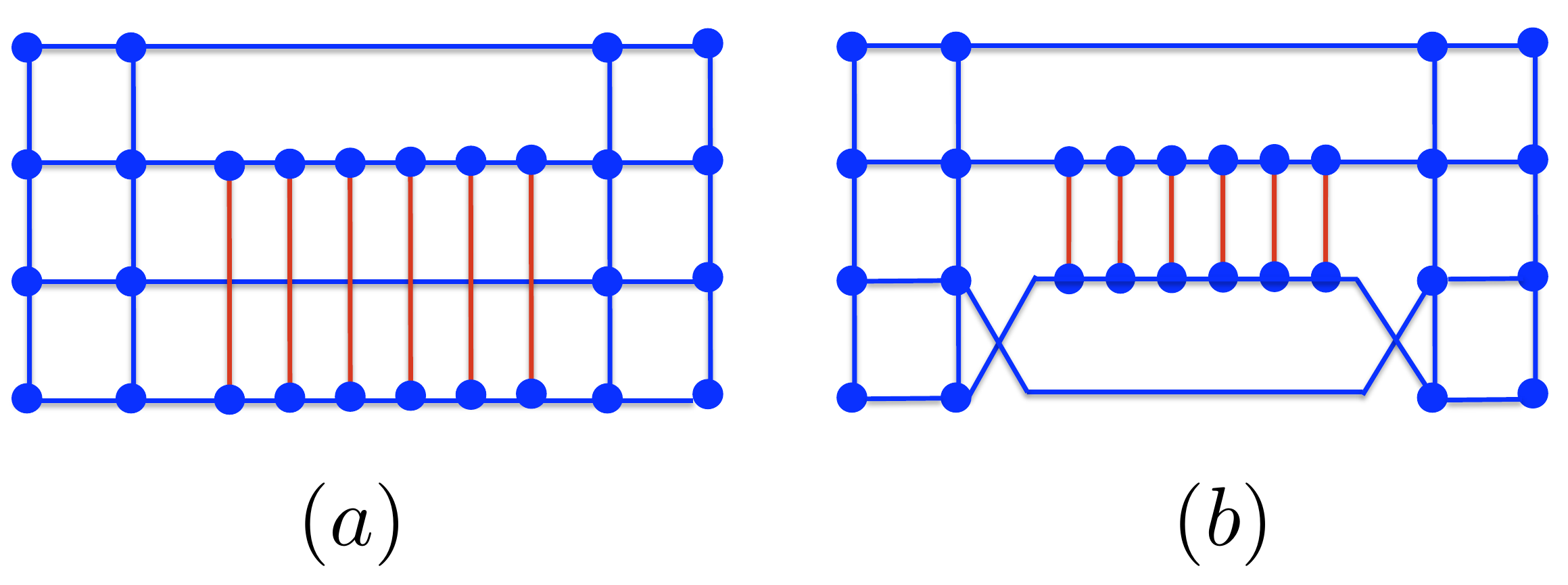}}} \caption{(a) Graph $G$. Red edges belong to $E^*$, blue edges to the planar sub-graph $H=G\setminus E^*$. Any drawing of $G$ in which the edges of $H$ do not cross each other has at least $6$ crossings. (b) An optimal drawing of $G$, with $2$ crossings.} \label{fig: example}
\end{center}
\end{figure}

Theorem~\ref{thm:main} also implies improved algorithms for several special cases of the problem, that are discussed below.

\noindent{\bf Nearly-Planar and Apex Graphs.} We say that a graph $G$ is \emph{$k$-nearly planar}, if it can be decomposed into a planar graph $H$, and a collection of at most $k$ additional edges. For the cases where the decomposition is given, or where $k$ is constant, Theorem~\ref{thm:main} immediately gives an efficient $O(\dmax^3\cdot k^2)$-approximation algorithm for \MCN.
It is worth noting that although this graph family might seem restricted, there has been a significant amount of work on the crossing number of $1$-nearly planar graphs.
Cabello and Mohar \cite{cabello_edge} proved that computing the crossing number remains NP-hard even for this special case, while
Hlin\v{e}n\'{y} and Salazar \cite{HlinenyS06} gave an $O(\dmax)$-approximation.
Riskin \cite{riskin_edge} gave a simple efficient procedure for computing the crossing number when the planar sub-graph $H$ is 3-connected, and
Mohar \cite{Mohar_almost} showed that Riskin's technique cannot be extended to arbitrary 3-connected planar graphs.
Gutwenger et al.~\cite{GutwengerMW05} gave a linear-time algorithm for the case where every crossing is required to be between $e$ and an edge of $G$.

A graph $G$ is a $k$-apex graph iff there are $k$ vertices $v_1,\ldots,v_k$, whose removal makes it planar. Chimani et al.~\cite{crossing_apex} obtained an $O(\dmax^2)$-approximation for \MCN on $1$-apex graphs. Theorem~\ref{thm:main} immediately implies an $O(\dmax^5\cdot k^2)$-approximation for $k$-apex graphs, where either $k$ is constant, or  the $k$ apices are explicitly given.

\noindent{\bf Bounded Genus Graphs.}
Recall that the genus of a graph $G$ is the minimum integer $g$ such that $G$ can drawn on an orientable surface of genus $g$ with no crossings.

B\"{o}r\"{o}zky et al.~\cite{BorozkyPT06} proved that 
the crossing number of a bounded-degree graph of bounded genus is $O(n)$.
Djidjev and Venkatesan~\cite{genus-planarization} show that $\optmp{G}\leq O(\sqrt{g\cdot n\cdot \dmax})$ for any genus-$g$ graph. Moreover, if the embedding of $G$ into a genus-$g$ surface is given, a planarizing set of this size can be found in time $O(n+g)$.
If no such embedding is given, they show how to efficiently compute a planarizing set of size $O(\sqrt{\dmax\cdot  g\cdot n\cdot \log g})$.

Hlin\v{e}n\'{y} and Chimani \cite{crossing_genus}, building on the work of Gitler et al.~\cite{crossing_projective} and Hlin\v{e}n\'{y} and Salazar~\cite{crossing_torus}  gave an algorithm for approximating \MCN on graphs that can be drawn ``densely enough\footnote{More precisely, the density requirement is that the nonseparating dual edge-width of the drawing is $2^{\Omega(g)}$.}'' in an orientable surface of genus $g$,
with an approximation guarantee of $2^{O(g)}\dmax^2$. \iffull
Despite the rather technical conditions on the input, this is the largest family of graphs for which a constant-factor approximation for the crossing number is known.\fi
We prove the following easy consequence of Theorem~\ref{thm:main} and the result of~\cite{crossing_genus}:

\begin{theorem}\label{thm: bounded genus}
Let $G$ be any graph embedded in an orientable surface of genus $g\geq 1$.
Then we can efficiently find a drawing of $G$ into the plane, with at most $2^{O(g)} \cdot \dmax^{O(1)} \cdot \optcrosq{G}$ crossings.
Moreover, for any $g\geq 1$, there is an efficient $\tilde{O}\left(2^{O(g)} \cdot \sqrt{n}\right)$-approximation for \MCN on bounded degree graphs embedded into a genus-$g$ surface.
\end{theorem}

We notice that when $g$ is a constant, a drawing of a genus-$g$ graph on a genus-$g$ surface can be found in linear time \cite{Mohar99, KawarabayashiMR08}.

\subsection{Our Techniques}
We now provide an informal overview of the proof of Theorem~\ref{thm:main}. We will use the words ``drawing'' and ``embedding'' interchangeably.
We say that a drawing $\psi$ of the planar graph $H=G\setminus E^*$ is \emph{planar} iff $\psi$ contains no crossings. Let $\phi$ be the optimal drawing of $G$, and let $\phi_H$ be the induced drawing of $H$.
For simplicity, let us first assume that the graph $H$ is $3$-vertex connected. 
Then we can efficiently find a planar drawing $\psi$ of $H$, which by Whitney's Theorem~\cite{Whitney} is guaranteed to be unique.
Notice however that the two drawings $\phi_H$ and $\psi$ of $H$ are not necessarily identical, and in particular $\phi_H$ may be non-planar. 

We now add the edges $e\in E^*$ to the drawing $\psi$ of $H$. The algorithm for adding the edges is very simple. 
For each edge $e\in E^*$, we  choose the drawing $c_e$ that minimizes the number of crossings between $c_e$ and the images of the edges of $H$ in $\psi$. This task reduces to finding the shortest path in the graph dual to $H$.
We can ensure that the drawings of any pair $e,e'$ of edges in $E^*$ cross at most once, by performing an un-crossing step, which does not increase the number of other crossings. Let $\psi'$ denote this new drawing of the whole graph. The total number of crossings between pairs of edges that both belong to $E^*$ is then bounded by $k^2$, and it only remains to  bound the number of crossings between the edges of $E^*$ and the edges of $H$. In order to complete the analysis,
it is enough, therefore, to show, that for every edge $e\in E^*$, there is a drawing of $e$ in $\psi$,  
that has at most $\poly(\dmax) \optcro{G}$ crossings with the edges of $H$.
Since our algorithm finds the best possible drawing for each edge $e$, the bound on the total number of crossings will follow.

One of our main ideas is the notion of \emph{routing edges along paths}. Consider the optimal drawing $\phi$ of $G$, and let $e=(u,v)$ be some
edge in $E^*$, that is mapped into some curve $\gamma_e$ in $\phi$. We show that we can find a path $P_e$ in the graph $H$, whose endpoints are $u$ and $v$, such that, instead of drawing the edge $e$ along $\gamma_e$, we can draw it along a different curve $\gamma'_e$, that ``follows'' the drawing of the path $P_e$. That is, we draw $\gamma'_e$ very close to the drawing of $P_e$, in parallel to it. Moreover, we show that this re-routing of the edge $e$ along $P_e$ does not increase the number of crossings in which it participates by much. Consider now the drawing of $P_e$ in the planar embedding $\psi$ of the graph $H$. We can again draw the edge $e$ along the embedding of the same path $P_e$ in $\psi$. Let $\gamma''_e$ be the resulting curve. Since the embeddings $\phi_H$ and $\psi$ are different, it is possible that $\gamma''_e$ participates in more crossings than $\gamma'_e$. However, we show that the number of such new crossings can be bounded by the number of vertices and edges in $P_e$, whose local embeddings are different in $\phi$ and $\psi$. We then bound this number, in turn, by $\poly(\dmax)\optcro{G}$.

We now explain the notion of local embeddings in more detail. Given two drawings $\phi_H$ and $\psi$ of the graph $H$, we say that a vertex $v\in V(H)$ is \emph{irregular} iff the ordering of its adjacent edges, as their images enter $v$, is different in the two drawings. In other words, the local drawing around the vertex $v$ is different in $\phi_H$ and $\psi$ (see Figure~\ref{fig: irregular vertices}). We say that an edge $e=(u,v)\in E(H)$ is \emph{irregular} iff both of its endpoints are not irregular, but their orientations are different. That is, the orderings of the edges adjacent to each one of the two endpoints are the same in both $\phi_H$ and $\psi$, but say, for vertex $v$, both orderings are clock-wise, while for vertex $u$, one is clock-wise and the other is counter-clock-wise (see Figure~\ref{fig: irregular edges}). In a way, the number of irregular edges and vertices measures the difference between the two drawings. We show that, on the one hand, if $H$ is $3$-vertex connected, and $\psi$ is a planar embedding of $H$, then the number of irregular vertices and edges is bounded by roughly the number of crossings in $\phi_H$, which is in turn bounded by $\optcro{G}$. On the other hand, we show that for each edge $e\in E^*$, the number of new crossings incurred by the curve $\gamma''_e$ is bounded by the total number of irregular edges and vertices on the path $P_e$, thus obtaining the desired bound.

Assume now that $H$ is not $3$-vertex connected. In this case, it is easy to see that the number of irregular vertices and edges cannot be bounded by the number of crossings in $\phi_H$ anymore. In fact, it is possible that both $\psi$ and $\phi_H$ are planar drawings of $H$, so the number of crossings in $\phi_H$ is $0$, while the number of irregular vertices may be large (see Figure~\ref{fig: example2} for an example). However, if the original graph $G$ was $3$-vertex connected, then for any $2$-vertex cut $(u,v)$ in $H$, there is an edge $e\in E^*$ connecting the resulting two components of $H\setminus\set{u,v}$. We use this fact to find a specific planar drawing $\psi'$ of $H$, that is ``close'' to $\phi_H$, in the sense that, if we define the irregular edges and vertices with respect to the embeddings $\phi_H,\psi'$ of $H$, then we can bound their number  by the number of crossings in $\phi_H$.

Finally, if $G$ is not $3$-vertex  connected, then we first decompose it into $3$-vertex connected components, and then apply the above algorithm to each one of the components separately. In the end, we put all the resulting drawings together, while only losing a small additional factor in the number of crossings.

\begin{figure}[h]
\centering
\subfigure[Vertex $v$ is irregular.]{
	\scalebox{0.3}{\includegraphics{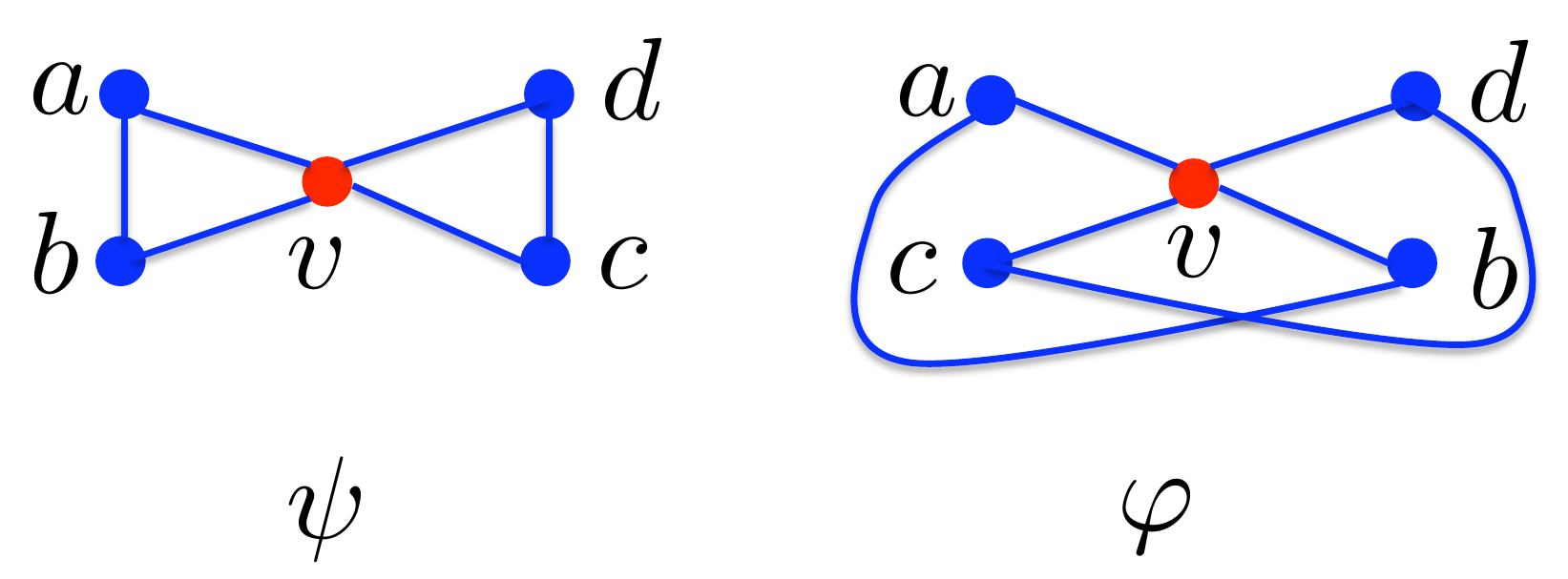}}
	\label{fig: irregular vertices}
}
\hspace{2cm}
\subfigure[Edge $e$ is irregular.]{
	\scalebox{0.3}{\includegraphics{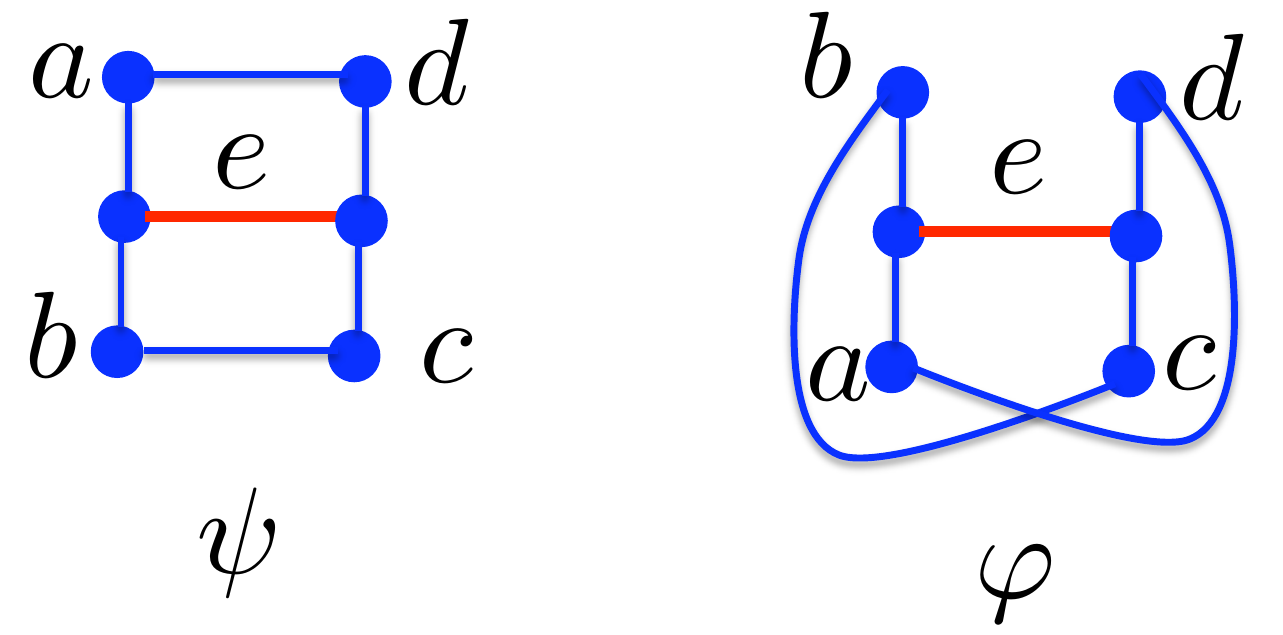}}
	\label{fig: irregular edges}
}
\caption{Irregular vertices and edges.}
\end{figure}

\begin{figure}[h]
\begin{center}
\scalebox{0.25}{\rotatebox{0}{\includegraphics{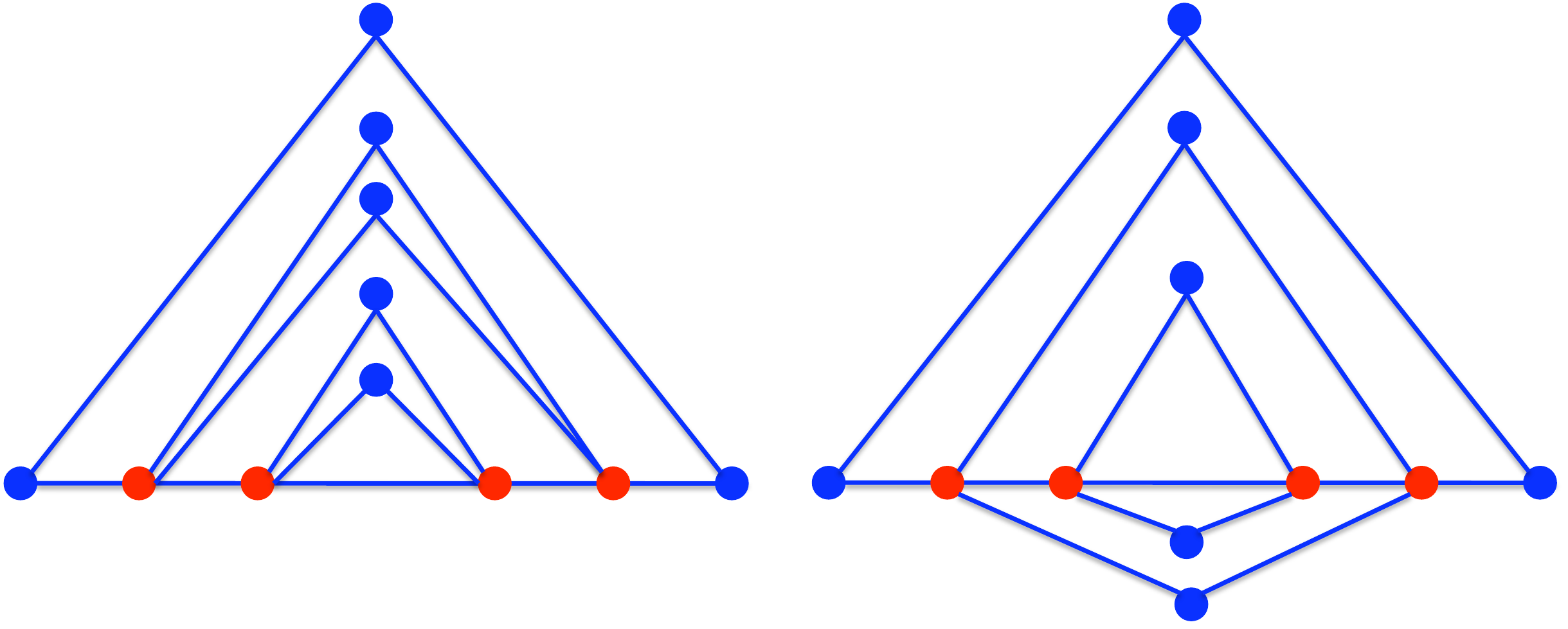}}} \caption{Example of planar drawings $\phi_H$ and $\psi$ of graph $H$. Irregular vertices are shown in red.}\label{fig: example2}
\end{center}
\end{figure}

\subsection{Other Related work}
Although it is impossible to summarize here the vast body of work on \MCN, we give a brief overview of some of the highlights, and related results.

\textbf{Exact algorithms.}
Grohe \cite{Grohe04}, answering a question of Downey and Fellows \cite{fpt}, proved that the crossing number is fixed-parameter tractable.
In particular, for any fixed number of crossings his algorithm computes an optimal drawing in $O(n^2)$ time.
Building upon the breakthrough result of Mohar \cite{Mohar99} for embedding graphs into a surface of bounded genus, Kawarabayashi and Reed \cite{KawarabayashiR07} gave an improved fixed-parameter algorithm with running time $O(n)$.

\textbf{Bounds on the crossing number of special graphs.}
Ajtai et al.~\cite{ajtai82}, and independently Leighton \cite{leighton_book}, settling a conjecture of Erd\"{o}s and Guy \cite{erdos_guy73}, proved that every graph with $m\geq 4n$ edges has crossing number $\Omega(m^3/n^2)$.
B\"{o}r\"{o}zky et al.~\cite{BorozkyPT06} proved that 
the crossing number of a bounded-degree graph of bounded genus is $O(n)$.
This bound has been extended to all families of bounded-degree graphs that exclude a fixed minor by Wood and Telle \cite{WoodT06}.
Spencer and T\'{o}th \cite{spencer_random} gave bounds on the expected value of the crossing number of a random graph.

\noindent{\bf Organization} Most of this paper is dedicated to proving Theorem~\ref{thm:main}. We start in Section~\ref{sec:prelims} with preliminaries, where we introduce some notation and basic tools. We then prove Theorem~\ref{thm:main} in Section~\ref{sec:alg}. We prove Theorem~\ref{thm:sqrt n} and Corollary~\ref{corollary:result for general graphs} in Section~\ref{sec:planarization for general graph}. 
\iffull 
The proof of Theorem~\ref{thm: bounded genus} appears in Section~\ref{sec:genus} of the Appendix.
\fi

\section{Preliminaries}\label{sec:prelims}
In this section we provide some basic definitions and tools used in the proof of Theorem~\ref{thm:main}. In order to avoid confusion, throughout the paper, we denote the input graph by $\G=(V,E)$, with $|V|=n$, and maximum degree $\dmax$. We also denote $\H=\G\setminus E^*$, the planar sub-graph of $\G$ (where $E^*$ is the set of edges from the statement of Theorem~\ref{thm:main}), and by $\bphi$ the optimal drawing of $\G$ with $\optcro{\G}$ crossings. When stating definitions or results for general arbitrary graphs, we will be denoting them by $G$ and $H$, to distinguish them from the specific graphs $\G$ and $\H$. 

We use the words ``drawing'' and ``embedding'' interchangeably. 
Given any graph $G$, a drawing $\phi$ of $G$, and any sub-graph $H$ of $G$, we denote by $\phi_H$ the drawing of $H$ induced by $\phi$, and by $\cro_{\phi}(G)$ the number of crossings in the drawing $\phi$ of $G$. For any pair $E_1,E_2\sse E(G)$ of subsets of edges, we denote by $\cro_{\phi}(E_1,E_2)$ the number of crossings in $\phi$ in which images of edges of $E_1$ and edges of $E_2$ intersect, and by $\cro_{\phi}(E_1)$ the number of crossings in $\phi$ between pairs of edges that both belong to $E_1$. Finally, for any curve $\gamma$, we denote by $\cro_{\phi}(\gamma,E_1)$ the number of crossings between $\gamma$ and the images of the edges of $E_1$, and $\cro_{\phi}(\gamma,H)$ denotes $\cro_{\phi}(\gamma,E(H))$. We will omit the subscript $\phi$ when clear from context.
If $G$ is a planar graph, and $\phi$ is a drawing of $G$ that contains no crossings, then we say that $\phi$ is a \emph{planar} drawing of $G$.

For the sake of brevity, we write $P: u\connect v$ to denote that a path $P$ connects vertices $u$ and $v$. Similarly, if we have a drawing of a graph, we write $\gamma:u\connect v$ to denote that a curve $\gamma$ connects the images of vertices $u$ and $v$ (curve $\gamma$ may not be a part of the current drawing). 
In order to avoid confusion, when a curve $\gamma$ is a part of a drawing $\phi$ of some graph $G$, we write $\gamma\in \phi$. We denote by $\Gamma(\phi)$ the set of all curves that can be added to the drawing $\phi$ of $G$. In other words, these are all curves that do not contain images of vertices of $G$ (except as their endpoints), and do not contain any crossing points of $\phi$.
Finally, for a graph $G=(V,E)$, and subsets $V'\sse V$, $E'\sse E$ of its vertices and edges respectively, we denote by $G\setminus V'$, $G\setminus E'$ the sub-graphs of $G$ induced by $V\setminus V'$, and $E\setminus E'$, respectively.

\begin{Definition} For any graph $G=(V,E)$, a subset $V'\sse V$ of vertices is called a $c$-separator, iff $|V'|=c$, and the graph $G\setminus V'$ is not connected. We say that $G$ is $c$-connected iff it does not contain any $c'$-separators, for any $c'<c$.
 \end{Definition}
 
We will be using the following two well-known results:

\begin{theorem} (Whitney~\cite{Whitney})\label{thm:Whitney} Every 3-connected planar graph has a unique planar embedding.
\end{theorem}

\begin{theorem}(Hopcroft-Tarjan~\cite{planar-drawing})\label{thm:planar drawing}
For any graph $G$, there is an efficient algorithm to determine whether $G$ is planar, and if so, to find a planar drawing of $G$.
\end{theorem}

\paragraph{Irregular Vertices and Edges}
Given any pair $\phi,\psi$ of drawings of a graph $G$, we measure the distance between them in terms of \textit{irregular edges} and \textit{irregular vertices}:

\begin{Definition}
We say that a vertex $x$ of $G$ is irregular iff its degree is greater than $2$, 
and the circular ordering of the edges incident on it, as their images enter $x$, is different in $\phi$ and $\psi$ (ignoring the orientation). Otherwise we say that $v$ is regular.
We denote the set of irregular vertices by $\irreg_V(\phi, \psi)$. (See Figure~\ref{fig: irregular vertices}).
\end{Definition}

\begin{Definition}
For any pair $(x,y)$ of vertices in $G$,
we say that a path $P:x\connect y$ in $G$ is irregular iff $x$ and $y$ have degree at least $3$, all other
vertices on $P$ have degree $2$ in $G$, vertices $x$ and $y$ are regular, but their
orientations differ in $\phi$ and $\psi$. That is, the orderings of the edges adjacent to $x$ and to $y$ are identical in both drawings, but the pairwise orientations are different: for one of the two vertices, the orientations are identical in both drawings (say clock-wise), while for the other vertex, the orientations are opposite (one is clock-wise, and the other is counter-clock-wise). An edge $e$ is an irregular edge iff it is the first or the last edge on an irregular path. In particular, if the irregular path only consists of edge $e$, then $e$ is an irregular edge (see Figure~\ref{fig: irregular edges}). If an edge is not irregular, then we say that it is regular.
We denote the set of irregular edges by $\irreg_E(\phi, \psi)$.
\end{Definition}





\paragraph{Routing along Paths.} One of the central concepts in our proof is  that of routing along paths. Let $G$ be any graph, and $\phi$ any drawing of $G$. Let $e=(u,v)$ be any edge of $G$, and let $P:u\connect v$ be any path connecting $u$ to $v$ in $G\setminus \set{e}$. It is possible that the image of $P$ crosses itself in $\phi$. We will first define a very thin strip $S_P$ around the image of $P$ in $\phi$. We then say that the edge $e$ is routed along the path $P$, iff its drawing follows the drawing of the path $P$ inside the strip $S_P$, possibly crossing $P$.

In order to formally define the strip $S_P$, we first consider the graph $G'$, obtained from $G$, by replacing every edge of $G$ with a $2$-path containing $2\cro_{\phi}(G)$ inner vertices. The drawing $\phi$ of $G$ then induces a drawing $\phi'$ of $G'$, such that, if $P'$ is the path corresponding to $P$  in $G'$, then every edge of $G'$ crosses the image of $P'$ at most once; every edge of $G'\setminus P'$ has at most one endpoint that belongs to $P'$; and if an image of $e\not\in P'$ crosses $P'$, then no endpoint of $e$ belongs to $P'$. Let $E_1$ denote the subset of edges of $G'\setminus P'$ whose images cross the image of $P'$, let $E_2$ denote the subset of edges of $P'$ whose images cross the images of other edges in $P'$, and let $E_3$ denote the set of edges in $G'$ that have exactly one endpoint belonging to $P'$.

We now define a thin strip $S_{P'}$ around the drawing of path $P'$ in $\phi'$, by adding two curves, $\gamma'_L$ and $\gamma'_R$, immediately to the left and to the right of the image of $P'$ respectively, that follow the drawing of $P'$. Each edge in $E_1$ is crossed exactly once by $\gamma'_L$, and once by $\gamma'_R$. Each edge in $E_3$ is crossed exactly once by either $\gamma'_R$ or $\gamma'_L$. For each pair $(e,e')$ of edges in $E_2$ whose images cross, $\gamma'_L$ and $\gamma'_R$ will both cross each one of the edges $e$ and $e'$ exactly once. Curves $\gamma_L'$ and $\gamma_R'$ do not have any other crossings with the edges of $G'$. The region of the plane between the drawings of $\gamma'_L$ and $\gamma'_R$, which contains the drawing of $P'$, defines the strip $S'_P$. We let $S_P$ denote the same strip, only when added to the drawing $\phi$ of $G$. Let $\gamma_L$ and $\gamma_R$ denote the two curves that form the boundary of $S_P$, and let $\gamma\in \set{\gamma_L,\gamma_R}$. Then the crossings between $\gamma$ and the edges of $G$ can be partitioned into four sets, $C_1,C_2,C_3, C_4$ (see Figure~\ref{crossings-c1-c2-c3}), where: (1) There is a $1:1$ mapping between $C_1$ and the crossings between the edges of $P$ and the edges of $G\setminus P$; (2) For each edge $e'\not\in P$ that has exactly one endpoint in $P$, there is at most one crossing between $\gamma$ and $e'$ in $C_2$, and there are no other crossings in $C_2$; (3) For each edge $e'\not\in P$ that has exactly two endpoints in $P$, there are at most two crossings of $\gamma$ and $e'$ in $C_3$, and there are no other crossings in $C_3$; and (4) for each crossing between a pair $e,e'\in P$ of edges, there is one crossing between $\gamma$ and $e$, and one crossing between $\gamma$ and $e'$. Additionally, if $P$ crosses itself $c$ times, then $\gamma$ also crosses itself $c$ times.

\begin{Definition}
We say that the edge $e$ is \emph{routed along the path $P$}, 
iff its drawing follows the drawing of path $P$ inside the strip $S_P$, in parallel to the drawing of $P$, 
except that it is allowed to cross the path $P$. 
\end{Definition}

%
\begin{figure}
\begin{center}
\scalebox{0.8}{\includegraphics{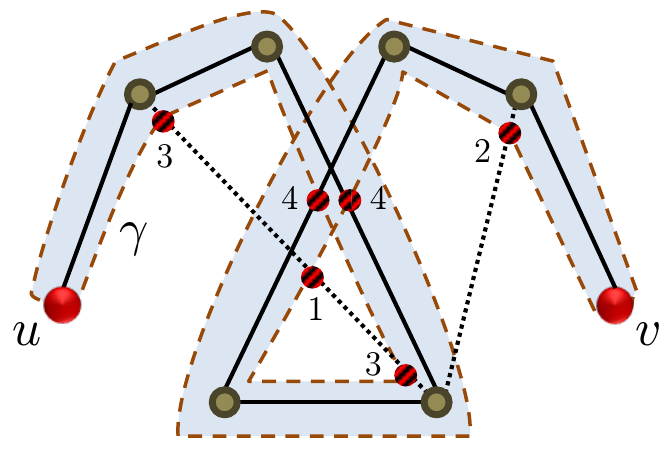}}
\caption{Strip $S_p$ and the four types of crossing between $\gamma$ and edges of $G$. Crossings in $C_1$, $C_2$, $C_3$ and $C_4$ are labeled with ``1'', ``2'', ``3'' and ``4'' respectively. Path $P$ is shown in solid line, dotted lines correspond to other edges of $G$.}
\label{crossings-c1-c2-c3}
\end{center}
\end{figure}


\section{Proof of Theorem~\ref{thm:main}}\label{sec:alg}
The proof consists of two steps. We first assume that the input graph $\G$ is $3$-vertex connected, and prove a slightly stronger version of Theorem~\ref{thm:main} for this case. Next, we show how to reduce the problem on general graphs to the $3$-vertex connected case, while only losing a small additional factor in the number of crossings.

\subsection{Handling $3$-connected Graphs}
In this section we assume that the input graph $\G$ is $3$-vertex connected, and we prove a slightly stronger version of Theorem~\ref{thm:main} for this special case, that is summarized below. 

\begin{theorem}\label{thm:main2}
Let $\G,\H$ and $E^*$ be as in Theorem~\ref{thm:main}, and assume that $\G$ is $3$-connected and has no parallel edges.
 Then we can efficiently find a drawing of $\G$ with at most $O\left(\dmax \cdot k\cdot (\optcro{\G}+k)\right )$ crossings.
\end{theorem}

Notice that we can assume w.l.o.g. that graph $\H$ is connected. Otherwise, we can choose an edge $e\in E^*$ whose endpoints belong to two distinct connected components of $\H$, remove $e$ from $E^*$ and add it to $\H$. It is easy to see that this operation preserves the planarity of $\H$, and we can repeat it until $\H$ becomes connected. We therefore assume from now on that $\H$ is connected.

Recall that $\bphi$ denotes the optimal drawing of $\G$, and $\bphi_{\H}$ is the drawing of $\H$ induced by $\bphi$.
Since the graph $\H$ is planar, we can efficiently find its planar drawing, using Theorem~\ref{thm:planar drawing}. However, since $\H$ is not necessarily $3$-connected, there could be a number of such drawings, and we need to find one that is ``close'' to $\bphi_{\H}$. 
We use the following theorem, whose proof appears in Appendix.
\begin{theorem}\label{thm: good planar drawing}
We can efficiently find a planar drawing $\bpsi$ of $\H$, such that
\begin{align*}
|\irreg_V(\bpsi, \bphi_{\H})| &=  O(\optcro{\G} + k)\\
|\irreg_E(\bpsi, \bphi_{\H})| &= O(\dmax)(\optcro{\G} + k).
\end{align*}
\end{theorem}

We are now ready to describe the algorithm for finding a drawing of $\G$. We start with the planar embedding $\bpsi$ of $\H$, guaranteed by Theorem~\ref{thm: good planar drawing}.
For every edge $e=(u,v)\in E^*$, we add an embedding of $e$ to the drawing $\bpsi$ of $\H$, via a curve $\gamma_e\in \Gamma(\bpsi)$, $\gamma_e:u\connect v$, that crosses the minimum possible number of edges of $\H$. Such a curve can be computed as follows.
Let $\H^{dual}$ be the dual graph of the drawing $\bpsi$ of $\H$. 
Every curve  $\gamma\in \Gamma(\bpsi)$, $\gamma:u\connect v$,  
defines a path in $\H^{dual}$.
The length of the path, measured in the number of edges of $\H^{dual}$ it contains, is exactly the number of edges of $\H$ that $\gamma$ crosses.
Similarly, every path in $\H^{dual}$ corresponds to a curve in $\Gamma(\bpsi)$.
Let $\cal U$ be the set of all faces of $\bpsi$ (equivalently, vertices of $\H^{dual}$) whose boundaries contain $u$, and
let $\cal V$ be the set of all faces whose boundaries contain $v$. 
We find the shortest path $P_{(u,v)}$ between sets $\cal U$ and $\cal V$ in $\H^{dual}$,
and the corresponding curve $\gamma_{(u,v)}:u \connect v$ in $\Gamma(\bpsi)$.
Clearly, the number of crossings between $\gamma_{(u,v)}$ and the edges of $\H$ is 
minimal among all curves connecting $u$ and $v$ in $\Gamma(\bpsi)$. 
%
%
By slightly perturbing the lengths of edges in $\H^{dual}$, we may
assume that for every pair of vertices in $\H^{dual}$, there is exactly one shortest path connecting them. In particular, any pair of such shortest paths may share at most one consecutive segment.
Consequently, for any pair $e,e'\in E^*$ of edges, the drawings $\gamma_e,\gamma_{e'}$ that we have obtained cross at most once.

Let $\bpsi'$ denote the union of $\bpsi$ with the drawings $\gamma_e$ of edges $e\in E^*$ that we have computed. It now only remains to bound the number of crossings in $\bpsi'$. Clearly, $\cro_{\bpsi'}(\G)=\cro_{\bpsi'}(E^*)+\cro_{\bpsi'}(E^*,E(\H))\leq k^2+\sum_{e\in E^*}\cro_{\bpsi'}(\gamma_e,E(\H))$. In order to bound $\cro_{\bpsi'}(\gamma_e,E(\H))$, we use the following theorem, whose proof appears in the next section.

\begin{theorem}\label{thm:main-along-path}
Let $\phi$ and $\psi$ be two drawings of any planar connected graph $H$, whose maximum degree is $\dmax$, where $\psi$ is a planar drawing.
Then for every curve $\gamma\in \Gamma(\phi)$, $\gamma:u \connect v$
there is a  curve $\gamma'\in \Gamma(\psi)$, $\gamma':u \connect v$,
that participates in at most $O(\cro_{\phi}(H) + \cro_{\phi}(\gamma, E(H)) + |\irreg_E(\phi,\psi)| + \dmax|\irreg_V(\phi,\psi)|)$ crossings.
\end{theorem}

In other words, the number of additional crossings incurred by $\gamma'$ is roughly bounded by the total number of crossings in $\phi$, and the difference between the two drawings, that is, the number of irregular vertices and edges.

Since the optimal embedding $\bphi$ of $G$ contains an embedding of every edge $e\in E^*$, Theorem~\ref{thm:main-along-path} guarantees that for every edge $e=(u,v)\in E^*$, there is a curve $\gamma'_e:u\connect v$ in $\Gamma(\bpsi)$, that participates in at most $O(\cro_{\phi}(\H)+\cro_{\phi}(e,E(\H))+ |\irreg_E(\bphi,\bpsi)| + \dmax|\irreg_V(\bphi,\bpsi)|)\leq O(\optcro{\G}+|\irreg_E(\bphi,\bpsi)| + \dmax|\irreg_V(\bphi,\bpsi)|)$ crossings. Combining this with Theorem~\ref{thm: good planar drawing}, the number of crossings between $\gamma'_e$ and $E(H)$ is bounded by $O(\dmax)(\optcro{\G} + k)$. 
Since for each edge $e\in E^*$, our algorithm chooses the optimal curve $\gamma_e$, we are guaranteed that $\gamma_e$ participates in at most $O(\dmax)(\optcro{\G} + k)$ crossings with edges of $H$.
Summing up over all edges $e\in E^*$, we obtain that $\cro_{\bpsi'}(\G)\leq k^2+k\cdot O(\dmax)(\optcro{\G} + k)\leq O(\dmax\cdot k\cdot (\optcro{\G}+k))$, as required. In order to complete the proof of Theorem~\ref{thm:main2}, it now only remains to prove Theorem~\ref{thm:main-along-path}.

\subsection{Proof of Theorem~\ref{thm:main-along-path}: Routing along Paths}\label{sec:routing along paths}
The proof consists of two steps. In the first step, we focus on the drawing $\phi$ of $H$, and we show that for any curve $\gamma:u\connect v$ in $\Gamma(\phi)$, there is a path $P:u\connect v$ in $H$, and another curve $\gamma^*: u\connect v$ in $\Gamma(\phi)$ routed along $P$ in $\phi$, such that the number of crossings in which $\gamma^*$ is involved is small. In the second step, we consider the planar drawing $\psi$ of $H$, and show how to route a curve $\gamma':u\connect v$ along the same path $P$ in $\psi$, so that the number of crossings is suitably bounded. The next proposition handles the first step of the proof.

\begin{proposition}\label{thm:rerouting-along-path} 
Let  $\gamma: u\connect v$ be any curve in $\Gamma(\phi)$, where $\phi$ is a drawing of $H$. Then there is a path $P:u\connect v$ in $H$, and a curve $\gamma^*: u\connect v$ in $\Gamma(\phi)$ routed along $P$, such that 
$\cro_{\phi}(\gamma^*, H)\leq O(\cro_{\phi}(H) + \cro_{\phi}(\gamma, H)).$ 
Moreover, $\gamma^*$ does not cross the images of the edges of $P$. Path $P$ is not necessarily simple, but an edge may appear at most twice on $P$.
\end{proposition}
\begin{proof}
Consider the drawing $\phi$ of $H$, together with the curve $\gamma$.
Let $E_1\sse E(H)$ be the subset of edges whose images cross the images of other edges of $H$, and
let $E_2\sse E(H)\setminus E_1$ be the subset of edges whose images cross $\gamma$ and that are not in $E_1$. 
Let $H'=H\setminus(E_1\cup E_2)$. Note that $\phi_{H'}$ is a planar drawing of $H'$, and $\gamma$ does not cross any edges of $H'$. Therefore, vertices $u$ and $v$ lie on the boundary of one face, denoted by $F$, of $\phi_{H'}$. Without loss of generality, we may assume that 
$F$ is the outer face of $\phi_{H'}$. The boundary of $F$
consists of one or several connected components. 
Let $B_1, \dots, B_r$ be the boundary walks of the face $F$ (where $r \geq 1$ is the number of connected components):
each $B_i$ is the (not necessarily simple) cycle obtained by walking around the boundary of the $i$th connected component,
if the component contains at least 2 vertices; and it is a single vertex otherwise.

Consider two cases. First, assume that $u$ and $v$ are connected in $H'$, and so they both belong to the same component $B_i$.
We then let $P$ be one of the two segments of $B_i$ that connect $u$ and $v$. Notice that while $P$ is not necessarily simple, each edge appears at most twice on it. We let $\gamma^*$ be a curve drawn along the path $P$ inside the face $F$. 
Notice that the only edges that $\gamma^*$ crosses in the drawing $\phi$ of $H$, are the edges of $E_1\cup E_2$ that have at least one endpoint on $P$. Each such edge is crossed at most twice by $\gamma^*$ (once for each endpoint that belongs to $P$). Therefore, $\cro_{\phi}(\gamma^*, H)\leq O(|E_1|+|E_2|)\leq O(\cro_{\phi}(H) + \cro_{\phi}(\gamma, H)).$
%
%
%
%
Assume now that $u$ and $v$ are not connected in $H'$, and assume w.l.o.g. that $u\in B_1$ and $v\in B_2$.
Let $L$ be a minimal set of edges of $E_1\cup E_2$, such that $u$ and $v$ are connected in $H'\cup L$.
Each edge $e\in L$ connects two distinct components $B_i$ and $B_j$ (as otherwise we could remove
$e$ without affecting the connectivity of $H'\cup L$). In particular, the drawings of all edges of $L$ in $\phi$ lie inside 
the face $F$. 
Consider the following graph $H^*$: each vertex of $H^*$ corresponds to a component $B_i$, for $1\leq i\leq r$, and the edges of $H^*$ are the edges of $L$ connecting these components. Since $u$ and $v$ are connected in $H'\cup L$, vertices representing $B_1$ and $B_2$ belong to the same connected component $C$ of $H^*$. Moreover, because of the minimality of $L$, this connected component is a simple path $P'$ connecting the vertices representing $B_1$ and $B_2$ in $H^*$.

\begin{figure}
\begin{center}
\scalebox{0.45}{\includegraphics{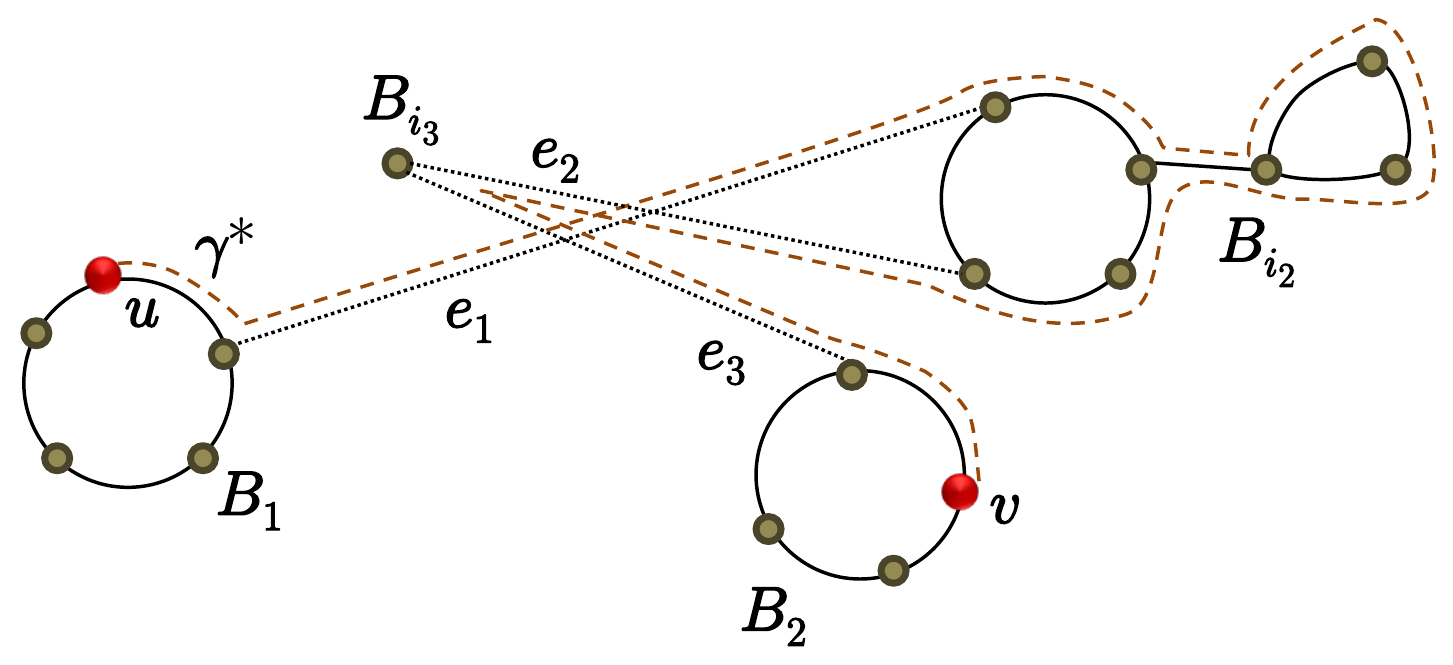}}
\caption{Routing the curve $\gamma^*:u\connect v$.\label{fig:curve}}
\end{center}
\end{figure}

Denote the edges of this path by $e_1,\dots, e_{q-1}$;
denote its vertices by $B_{i_1} \equiv B_1, B_{i_2}, \dots, B_{i_q}\equiv B_2$; each edge $e_j$, for $1\leq j\leq q-1$, corresponds to an edge $e'_j\in L$, that connects a pair $(x_j,y_j)$ of vertices, where $x_j\in B_{i_j}$,  $y_j\in B_{i_{j+1}}$, in $H$. We denote $u=y_0,v=x_q$. Notice that since $P'$ is simple, each component $B_j$, $1\leq j\leq r$ appears at most once on the path.
We now define the path $P$ and the curve $\gamma^*$. The path $P$ is defined as follows: $P=(P_1,e_1',P_2,e_2',\ldots,e_{q-1}',P_q)$, where for each $j: 1\leq j\leq q$, $P_j$ is obtained by traversing the boundary $B_{i_j}$ in the clock-wise direction from $y_{j-1}$ to $x_j$. The curve $\gamma^*: u\connect v$ is simply routed along $P_j$ on the inside of the face $F$. That is, curve $\gamma^*$ never crosses the images of the edges of $H'$ (see Figure~\ref{fig:curve}).

We now bound the number of edges of $H$, whose images in $\phi$ are being crossed by $\gamma^*$. We partition the crossings in which $\gamma^*$ participates into four sets $C_1,C_2,C_3, C_4$, like in the definition of routing along paths in Section~\ref{sec:prelims}. The only edges incident on vertices of path $P$ that $\gamma^*$ crosses are the edges in $E_1\cup E_2$, and each such edge contributes at most two crossings to $C_2\cup C_3$,  while the number of crossings in $C_1\cup C_4$ is bounded by $2\cro_{\phi}(H)$. 
Therefore, $\cro_{\phi}(\gamma^*, H) \leq O(|E_1|+|E_2|+\cro_{\phi}(H))=O(\cro_{\phi}(H) + \cro_{\phi}(\gamma, H))$.
\end{proof}

We now focus on the other embedding, $\psi$ of $H$, and show how to obtain the final curve $\gamma':u\connect v$, $\gamma'\in \Gamma(\psi)$, that participates in a small number of crossings.

\begin{proposition}\label{thm:embedding-along-path}
There is a curve $\gamma':u \connect v$ in $\Gamma(\psi)$, that has no self-crossings, and participates in at most $\cro_{\phi}(\gamma^*, H) + O(|\irreg_E(\phi,\psi)| + \dmax|\irreg_V(\phi,\psi)|)$ crossings with the edges of $H$.
\end{proposition}
\begin{proof}
We will route $\gamma'$ along the path $P$ in $\psi$. Since an edge
may appear at most twice on $P$, path $P$ may visit a vertex at most $\dmax$ times. We will assume however that $P$ visits every irregular vertex
at most once, by changing $P$ as follows: whenever an irregular vertex $x$ appears more than once on $P$, we create a shortcut, by removing the segment of $P$ that lies between the two consecutive appearances of $x$ on $P$.  As a result, in the final path $P$, each edge appears at most twice, and each irregular vertex at most once.

We will route the curve $\gamma'$ along $P$, but we will allow it to cross the image of the path $P$. Therefore, we only need to specify, for each edge $e\in P$, whether $\gamma'$ crosses it, and if not, on which side of $e$ it is routed. Since $\psi$ is planar, the edges of $P$ do not cross each other.

We partition the path $P$ into consecutive segments $\tau_0,\sigma_1,\tau_1,\sigma_2,\ldots,\sigma_t,\tau_t$, where for each $j: 1\leq j\leq t$, $\sigma_j$ contains regular edges only, and all its vertices are regular, except perhaps the first and the last. For each $j: 0\leq j\leq t$, either $\tau_j$ contains one or several consecutive  irregular edges connecting the last vertex of $\sigma_j$ and the first vertex of $\sigma_{j+1}$; or it contains a single irregular vertex, which serves as the last vertex of $\sigma_j$ and the first vertex of $\sigma_{j+1}$. 

Consider some such segment $\sigma_j$, and a thin strip $S=S_{\sigma_j}$ around this segment. Then the parts of the drawings of the edges incident on the vertices of $\sigma_j$, that fall inside $S$ are identical in both $\phi$ and $\psi$ (except possibly for the edges incident on the first and the last vertex of $\sigma_j$). We can therefore route $\gamma'$ along the same side of $\sigma_j$ along which $\gamma^*$ is routed. If necessary, we may need to cross the path $P$ once for each consecutive pair of segments, if the routings are performed on different sides of $P$. Let $\gamma^*_j$ and $\gamma'_j$ denote the segments of $\gamma^*$ and $\gamma'$, respectively, that are routed along $\sigma_j$, and include crossings with all edges incident on $\sigma_j$.
It is easy to see that the difference $\cro_{\psi}(\gamma'_j, H)-\cro_{\phi}(\gamma^*_j,H)$ is bounded by $2\dmax$: we pay at most $\dmax$ for crossing the edges incident on each endpoint of $\sigma_j$, which may be an irregular vertex. We may additionally pay $1$ crossing for each irregular edge on $P$. Since each irregular vertex appears at most once on $P$, and each irregular edge at most twice, $\cro_{\psi}(\gamma',H)-\cro_{\phi}(\gamma^*,H)\leq O(|\irreg_E(\phi,\psi)| + \dmax|\irreg_V(\phi,\psi)|)$. Finally, if $\gamma'$ crosses itself, we can simply short-cut it by removing all resulting loops. 
\end{proof}

Combining Propositions~\ref{thm:rerouting-along-path} and~\ref{thm:embedding-along-path}, we get that $\cro_{\phi}(\gamma^*, H)\leq O(\cro_{\phi}(H) + \cro_{\phi}(\gamma, H))$, and $\cro_{\psi}(\gamma',H)\leq \cro_{\phi}(\gamma^*, H) + O(|\irreg_E(\phi,\psi)| + \dmax|\irreg_V(\phi,\psi)|)\leq O(\cro_{\phi}(H) + \cro_{\phi}(\gamma, H) + |\irreg_E(\phi,\psi)| + \dmax|\irreg_V(\phi,\psi)|)$.

\subsection{Non 3-Connected Graphs}
\vspace{-3mm}

We briefly explain how to reduce the general case to the $3$-connected case. We decompose the graph into a collection of sub-graphs. For each sub-graph, we find a drawing separately, and then combine them together to obtain the final solution. Each one of the sub-graphs is either a $3$-connected graph, for which we can find a drawing using Theorem~\ref{thm:main2}, or it can be decomposed into a planar graph plus one additional edge. In the latter case, we employ the algorithm of Hlineny and Salazar~\cite{HlinenyS06} to find an $O(\dmax)$-approximate drawing. 
\iffull
The detailed proof of this part is presented in Section~\ref{sec:reduce-three-connected} in the Appendix.
\fi

\section{Improved Algorithm for General Graphs}\label{sec:planarization for general graph}
In this section we prove Theorem~\ref{thm:sqrt n} and Corollary~\ref{corollary:result for general graphs}.
We will rely on the Planar Separator Theorem of Lipton and Tarjan~\cite{planar-separator}, and on the approximation algorithm for the Balanced Cut problem of Arora, Rao and Vazirani~\cite{ARV}, that we state below.
\begin{theorem}[Planar Separator Theorem~\cite{planar-separator}]\label{thm: planar separator}
Let $G$ be any $n$-vertex planar graph. Then there is an efficient algorithm to partition the vertices of $G$ into three sets $A,B,C$, such that $|A|,|C|\leq 2n/3$, $|B|\leq O(\sqrt{n})$, and there are no edges in $G$ connecting the vertices of $A$ to the vertices of $C$.
\end{theorem}
\begin{theorem}[Balanced Cut~\cite{ARV}]\label{thm: ARV}
Let $G$ be any $n$-vertex graph, and suppose there is a partition of vertices of $G$ into two sets, $A$ and $C$, with $|A|,|C|\leq 2n/3$, and 
$|E(A,C)|=c$. Then there is an efficient algorithm to find a partition $(A',C')$ of vertices of $G$, such that $|A'|,|C'|\leq \alpha n$ for some constant $\alpha<1$, and $|E(A',C')|\leq O(c\sqrt{\log n})$.
\end{theorem}
Combining the two theorems together, we get the following corollary\iffull:\fi
\begin{corollary}\label{corollary: cut}
Let $G$ be any $n$-vertex graph with maximum degree $\dmax$. Then there is an efficient algorithm to partition the vertices of $G$ into two sets $A',C'$, with $|A'|,|C'|\leq \alpha n$ for some constant $\alpha$, such that $|E(A',C')|\leq O(\sqrt{\log n})(\dmax\sqrt n+\optmp{G})$.
\end{corollary}
\iffull
\begin{proof}
Let $E^*$ be an optimal solution for the \MP problem on $G$, $|E^*|=\optmp{G}$, and let $H=G\setminus E^*$. Since $H$ is a planar graph, there is a partition $(A,B,C)$ of its vertices as in Theorem~\ref{thm: planar separator}. Assume w.l.o.g. that $|A|\leq |C|$, and consider the partition $(A\cup B,C)$. Then $|A\cup B|,|C|\leq 2n/3$, and $|E_G(A\cup B,C)|\leq |E_H(A\cup B,C)|+|E^*|\leq O(\dmax\sqrt n)+\optmp{G}$. We can now apply Theorem~\ref{thm: ARV} to obtain the desired partition of $G$.
\end{proof}
\fi

We are now ready to describe the algorithm from Theorem~\ref{thm:sqrt n}. The algorithm consists of $O(\log n)$ iterations, and in each iteration $i$, we are given a collection $G_1^i,\ldots,G_{k_i}^i$ of disjoint sub-graphs of $G$, with $k_i\leq \optmp{G}$. The number of vertices in each such sub-graph is bounded by $n_i=\alpha^{i-1}n$, where $\alpha<1$ is the constant from Corollary~\ref{corollary: cut}. In the input to the first iteration, $k_1=1$, and $G_1^1=G$.
Iteration $i$, for $i\geq 1$ is performed as follows. Consider some graph $G^i_j$, for $1\leq j\leq k_i$. We apply Corollary~\ref{corollary: cut} to this graph, and denote by $H_j,H'_j$ the two sub-graphs of $G^i_j$ induced by $A'$ and $C'$, respectively. The number of vertices in each one of the subgraphs is at most $\alpha\cdot |V(G^i_j)|\leq \alpha n_i=n_{i+1}$. We denote by $E^i_j$ the corresponding set of edges $E(A',C')$, and let $E^i=\bigcup_{j=1}^{k_i}E^i_j$. Since for all $j$, $|E^i_j|\leq O(\sqrt{\log n})(\dmax\sqrt {n_i}+\optmp{G^i_j})$, and $\sum_{j=1}^{k_i}\optmp{G^i_j}\leq \optmp{G}$, we get that $|E^i|\leq O(\sqrt{\log n})(k_i\dmax\sqrt{n_i}+\optmp{G})\leq O(\dmax\sqrt{\log n}\cdot \sqrt{n_i})\optmp{G}$, as $k_i\leq \optmp{G}$. Finally, consider the collection $\gset_{i+1}=\set{H_1,H_1',\ldots,H_{k_i},H_{k_i}'}$ of the new graphs, and let $\gset'_{i+1}\sse \gset_{i+1}$ contain the non-planar graphs. Then $|\gset'_{i+1}|\leq \optmp{G}$, and the graphs in $\gset'_{i+1}$ become the input to the next iteration. 
Since we can efficiently check whether a graph is planar, the set $\gset'_{i+1}$ can be computed efficiently.

The algorithm stops, when all remaining sub-graphs contain at most $O(\sqrt{\log n})$ edges. We then add the edges of all remaining sub-graphs to set $E^{i^*}$, where $i^*=O(\log n)$ is the last iteration. Our final solution is $E'=\bigcup_{i=1}^{i^*}E^i$, and its cost is bounded by $|E'|\leq \sum_{i=1}^{i^*}|E^i|\leq \sum_{i=1}^{i^*} O(\dmax\sqrt{\log n}\cdot \sqrt{n_i})\optmp{G}\leq O(\dmax\sqrt{n\log n})\optmp{G}$, since the values $\sqrt{n_i}$ form a decreasing geometric series for $i\geq 1$. This finishes the proof of Theorem~\ref{thm:sqrt n}.
We now show how to obtain Corollary~\ref{corollary:result for general graphs}. Combining Theorems~\ref{thm:main} and \ref{thm:sqrt n}, we immediately obtain an efficient algorithm for drawing any graph $G$ with at most $O(n\log n\cdot \dmax^5)\optcrosq{G}$ crossings. In order to get the approximation guarantee of $O(n \cdot\poly(\dmax)\cdot \log^{3/2}n)$, we use an extension of the result of Even et al.~\cite{EvenGS02} to arbitrary graphs, that we formulate in the next theorem, \iffull whose proof appears in Appendix.\fi 

\begin{theorem}[Extension of~\cite{EvenGS02}]\label{thm: Even: extension}
There is an efficient algorithm that, given any $n$-vertex graph $G$ with maximum degree $\dmax$, outputs a drawing of $G$ with $O(\poly(\dmax)\log^2n)(n+\optcro{G})$ crossings.
\end{theorem}

We run our algorithm, and the algorithm given by Theorem~\ref{thm: Even: extension} on the input graph $G$, and output the better of the two solutions. If $\optcro{G}\geq \sqrt{\log n}$, then the algorithm of Even et al. is an $O(n \cdot\poly(\dmax)\cdot \log^{3/2}n)$-approximation. Otherwise, our algorithm gives an $O(n\cdot \poly(\dmax)\cdot \log^{3/2}n)$-approximation.


\bibliography{soda}
\bibliographystyle{siam}

\iffull
\newpage
\fi

\appendix
\section{Block Decompositions}\label{sec: blocks}
In this section we introduce the notion of \textit{blocks}, and present a theorem for computing block decompositions of graphs, that we will
later use to handle graphs that are not $3$-connected. 

\begin{Definition}
Let $G=(V,E)$ be a $2$-connected graph. 
A subgraph $B=(V',E')$ of $G$ is called a \emph{block} iff: 

\begin{itemize}
\item $V\setminus V'\neq\emptyset$ and $|V'|\geq 3$;
\item There are two special vertices $u,v\in V'$, called \emph{block end-points} and denoted by $I(B)=(u,v)$, such that there are no edges connecting vertices in $V\setminus V'$ to $V'\setminus\set{u,v}$ in $G$, that is, $E(V\setminus V',V'\setminus\set{u,v})=\emptyset$. All other vertices of $B$ are called \emph{inner vertices}; 

\item  $B$ is the subgraph of $G$ induced by $V'$, except that it {\bf does not} contain the edge $\{u,v\}$ even if it is present in $G$.
\end{itemize}
\end{Definition}

Notice that every $2$-separator $(u,v)$ of $G$ defines at least two internally disjoint blocks $B',B''$ with $I(B'),I(B'')=(u,v)$.

\begin{Definition}
Let $\fset$ be a laminar family of sub-graphs of $G$, and let $\tset$ be the decomposition tree associated with $\fset$. We say that $\fset$ is a \emph{block decomposition} of $G$, iff:

\begin{itemize}
\item The root of the tree $\tset$ is $G$, and all other vertices of $\tset$ are blocks. For consistency, we will call the root vertex ``block'' as well.

\item For each block $B\in \fset$, let $\tilde{B}$ be the graph obtained by replacing each child $B'$ of $B$
with an artificial edge connecting its endpoints. Let $\tilde{B}'$ be the graph obtained from $\tilde{B}$ by adding an artificial edge connecting the endpoints of $B$ (for the root vertex $G$, $\tilde{G}'=\tilde{G}$). Then $\tilde{B}'$ is $3$-connected.
\item If a block $B\in \fset$ has exactly one child $B'$ then $I(B) \neq I(B')$.
\end{itemize} 
\end{Definition}


\iffull
The next theorem states that we can always find a good block decomposition for any 2-connected graph. 
\fi

\begin{theorem}\label{thm: block decomposition}
Given a $2$-connected graph $G=(V,E)$ with $|V|\geq 3$, we can efficiently find a laminar block decomposition $\fset$ of $G$, such that for
every vertex $v\in V$ that participates in any $2$-separator $(u,v)$ of $G$, one of the following holds:
\iffull
\begin{itemize}
\item\fi Either $v$ is an endpoint of a block $B\in\cal F$;
\iffull
\item\fi or  $v$ has exactly two neighbors in $G$, and there is an edge $(u',v)\in E$, such that $u'$ is an endpoint of a block $B\in\cal F$.
\iffull
\end{itemize}
\fi
\end{theorem}


\iffull


We will use the notion of \textit{SPQR-trees} in the proof. 
Recall that an SQPR tree $\tset'$ for the graph $G$ defines a recursive decomposition
of $G$, as follows. Each node $x$ of the tree $\tset'$ is associated with a graph $G_x=(V_x,E_x)$, where $V_x\sse V$, and $E_x$ consists of edges of $G$ (called  \textit{actual} edges), and some additional edges, called \textit{artificial} or \textit{virtual} edges.
Graph $G_x$ is allowed to contain parallel edges. Additionally, we are given a bijection $f_x$ between the artificial edges of $G_x$, and the edges adjacent to the vertex $x$ in the tree $\tset'$.

Each actual edge of graph $G$ belongs to exactly one of the graphs $G_x$, for $x\in V(\tset')$. The edges of the tree $\tset'$, together with the artificial edges of the graphs $G_x$, show how to compose the graphs $G_x$ together to obtain the original graph $G$. More specifically, if $e=(x,y)$ is an edge in the tree $\tset'$, then there is a unique artificial edge $e_x$ in $G_x$ associated with it, and a unique artificial edge $e_y$ in $G_y$ associated with it (that is, $f_x(e_x)=f_y(e_y)=e$). The endpoints of both these artificial edges are copies of the same two vertices, that is, $e_x=(u,v)=e_y$, for $u,v\in V(G)$. 
The graphs $G_x$ and $G_y$ share no vertices other than $u$ and $v$.
The removal of the edge $(x,y)$ from $\tset'$ decomposes the tree into two connected components, $C_1$ and $C_2$. Let $V_1\sse V$ denote the set of all vertices appearing in the graphs $G_z$, for $z\in C_1$, and similarly, let $V_2\sse V$ denote the vertices appearing in the graphs $G_z$ where $z\in C_2$. Then $V_1\cap V_2=\set{u,v}$, and $V_1\cup V_2=V(G)$. 

The vertices of the tree $\cal T'$ belong to one of the four types: $S$, $P$, $Q$, and $R$.
\begin{itemize}
\item If $x$ is an \emph{$S$-node}, then the graph $G_x$ is a cycle, with no parallel edges.
\item If $x$ is a \emph{$P$-node}, then $G_x$ consists of a pair of vertices connected by at least 
three parallel edges, and at most one of these edges is an actual edge.
\item If $x$ is an \emph{$R$-node}, then $G_x$ is a 3-connected graph, with more than 3 vertices and
no parallel edges.
\item If $x$ is a \emph{$Q$-node}, then $G_x$ is just a single edge. The tree $\tset'$ has
$Q$-nodes only if $G$ itself is an edge, and then $\cal T'$ has no other nodes.
\end{itemize}
No two $P$-nodes and no two $S$-nodes are adjacent in ${\cal T'}$.

For every 2-connected graph $G$, there is a unique SPQR tree. Moreover,
this tree can be found in linear time, as was shown by Gutwenger and Mutzel~\cite{SPQR-trees}.  
The SPQR tree of $G$ describes the set of all 2-separators of $G$, as follows: a pair of 
vertices $u$ and $v$ of $G$ is a 2-separator if and only if either (1) there is an artificial edge $(u,v)$ in some graph $G_x$, for $x\in \tset'$, or (2) $u$ and $v$ are non-adjacent vertices in some graph 
$G_x$, where $x$ is an $S$-node.

We now describe how to construct the laminar block decomposition $\cal F$, and the associated decomposition tree $\tset$, for $G$. Since $|V|\geq 3$ and $G$ is $2$-connected, the tree $\tset'$ does not contain any $Q$-nodes.
We assume first that $G$ is not a cycle, and we treat the case of the cycle graph 
separately. We start by computing the SPQR tree $\cal T'$ of $G$. We then choose an arbitrary 
$P$ or $R$-node $r$ of $\cal T'$ to serve as the root of the tree $\tset'$. For each node $x\in V(\tset')$, we denote the subtree rooted at
$x$ by ${\cal T}_x'$. If $x$ is the parent of $y$ in the tree $\tset'$, we call the artificial edge 
of $e_x$ of $G_x$ that is mapped to $(x,y)$ a \emph{child edge};
we call the edge of $G_y$ that is mapped to $(x,y)$ a \emph{parent edge}. For each $x\in V(\tset')\setminus\set{r}$, $G_x$  has exactly one parent edge. If $(u,v)$ is the parent edge of $G_y$ and graph $G$ contains an actual edge $(u,v)$, we can assume w.l.o.g that $G_y$ does not contain the actual edge $(u,v)$, as we can remove it from $G_y$ and add it to $G_x$. This operation may introduce parallel edges in graphs $G_x$ corresponding to $S$-nodes or $R$-nodes. 

We now proceed in two steps. First, for each node $x\in V(\tset')$, we define a block $B_x$, that is added to $\fset$. This will define a valid laminar block decomposition $\fset$, except that if $x$ is an $S$-node, then $\tilde{B}_x'$ is not necessarily $3$-connected. For each such $S$-node $x$, we then add additional blocks to $\fset$ in the second step, to avoid this problem.

\paragraph{Step 1:}
Consider a node $x$ of $\cal T'$. If $x\neq r$, denote the parent edge of $G_x$ by 
$(u,v)\in E(G_x)$ (if $x=r$, we do not define $u$ and $v$). Let $B_x$ be the union
of all graphs $G_y$ associated with the nodes $y$ in ${\cal T}_x'$, with all artificial 
edges removed. 

Since $B_x$ does not contain any artificial edges, it is a subgraph of $G$.
Clearly, $B_r = G$ since every edge of $G$ appears in some graph $G_z$, for $z\in\tset'$.
We now verify that for every node $x\in V(\tset')\setminus\set{r}$, $G_x$ is indeed a block. 
Let $V_x=V(B_x)$, and let $(u,v)$ be the parent edge of $G_x$. 
Since the graph does not contain any $Q$-nodes, $|V(B_y)|\geq 3$ for all $y\in V(\tset')$, and since for every adjacent pair $(y,z)$ of vertices on $\tset'$, $G_y$ and $G_z$ only share two vertices, $V\setminus V(B_x)\neq \emptyset$.
We now show that $B_x$ is an induced subgraph of $G$ (except that edge $(u,v)$ is not in $B_x$,  even if it is present in $G$).
Let $e$ be the edge of $\tset'$, connecting $G_x$ to its parent. Recall that $e$ decomposes the tree $\tset'$ into two connected components, $C_1$ and $C_2$, where one of the components is $\tset_x'$. Assume it is $C_1$. We have also defined $V_1\sse V$ to be the union of $V(G_z)$ for $z\in C_1$, and similarly $V_2\sse V$ is the union of $V(G_z)$ for $z\in C_2$. Recall that $V_1\cap V_2=\set{u,v}$, and by definition of $B_x$, $V(B_x)=V_1$.
Consider some edge $(w_1,w_2) \neq (u,v)$ of $G$ that connects two vertices of $V_x$. Then since  $V_1\cap V_2=\set{u,v}$, vertices $w_1,w_2$ do not both belong to $V_2$. 
By the definition of SPQR trees, the edge $(w_1, w_2)$ belongs to some graph $G_z$, for $z\in V(\tset')$. Therefore $z\in C_1$ must hold, and $(w_1,w_2)\in E(B_x)$. Also, as we have observed before, there are no edges connecting $V_1\setminus\set{u,v}$ to $V_2\setminus\set{u,v}$ in $G$. Therefore, $B_x$ is indeed a block.
 Note that if $x$ is a $P$-node with only one child then $G_x$ consists of 2 edges: one parent artificial edge, and one child artificial edge (in this case, originally $G_x$ contained a third, actual edge, that we have moved to the graph of its father). Therefore for the child node $y$ of $x$, $B_y$ contains all vertices and edges that $B_x$ contains. We add all blocks $B_x$, for all $x\in V(\tset')$, to $\cal F$ except if $x$ is a $P$-node with only one child. Notice that under this definition of $\fset$, for each block $B_x\in \fset$, $\tilde{B}_x'=G_x$. Clearly, if $x$ is not an $S$-node, then $\tilde{B}'_x$ is $3$-connected.

\paragraph{Step 2} In this step we take care of the $S$-nodes. Consider an $S$-node $x$, and assume that $G_x$ is a cycle
$(a_1,\dots, a_s)$, where $(a_1,a_s)$ is the parent artificial edge of $G_x$. We define a nested set $B^1_x\supset B^2_x\supset\cdots\supset B^{s-3}_x$ of blocks, where $B^1_x\subset B_x$, that will be added to $\fset$, in addition to $B_x$. In order to define the blocks $B^i_x$, we define a collection $P_0,P_1,\ldots,P_{s-3}$ paths, as follows. Path $P_0$ is obtained from $G_x$ by removing the artificial parent edge $(a_s,a_1)$, so $P_0=(a_1,a_2,\ldots,a_s)$. Path $P_i$, for $i> 0$, is obtained from $P_{i-1}$ as follows: if $i$ is even, remove the first edge of $P_{i-1}$, and if it is odd, remove the last edge of $P_{i-1}$. Therefore, $P_i$ is the portion of $P_0$ between $a_{1 + \lfloor i/2 \rfloor}$ and $a_{s-\lceil i/2 \rceil})$. Let $Y_i$ be the set of all child nodes of $x$ in $\tset'$, corresponding to the artificial edges of $P_i$.
We are now ready to define $B^i_x$, for $1\leq i\leq s-3$: it contains all actual edges of $P_i$ and
all blocks $B_y$ for $y\in Y_i$. 
Graph $B_x^i$ has two types of vertices: inner vertices of blocks $B_y$, for $y\in Y_i$, and
the vertices of $P_i$. 
We now show that $B_x^i$ is a block. Since each $B_y$, for $y\in Y_i$, is a block, no edge connects the interior of $B_y$
to $G \setminus B_y$ (and therefore to $G \setminus B^i_x$). So in order to prove that
$B^i_x$ is a block with endpoints $a_{1 + \lfloor i/2 \rfloor}$ and $a_{s-\lceil i/2 \rceil}$,
it remains to show that there is no edge connecting an internal vertex $a_j$ of $P_i$ and 
a vertex in $V \setminus B_x^i$. But this is clearly true since we have already proved that $B_x$ is a block.
We add all blocks $B_x^i$, for $1\leq i\leq s$ to $\cal F$.
For convenience, we denote $B_x^0 = B_x$ and $B_x^{s-2} = \text{``\textsc{empty graph}''}$.

This completes the description of the family of blocks $\cal F$.
It is clear from the construction, that $\cal F$ is a laminar family of blocks.
Now consider a block $B\in \cal F$, and the corresponding graph $\tilde{B}'$. We need to prove that $\tilde{B}'$ is $3$-connected. First,
if $B=B_x$, where $x$ is an $R$ or $P$-node then $\tilde{B}' = G_x$, and therefore it is 3-connected. 
If $B=B_x^i$, for $i\in \{0,\dots, s-3\}$, for an $S$-node $x$, where $G_x$ is a cycle on $s$ vertices, then $\tilde{B}'$ is obtained from $P_i$ by replacing $P_{i+1}$ with an artificial edge, and connecting the endpoints of $P_i$ with another artificial edge. Therefore, $\tilde{B}'$ is the triangle graph, which is $3$-connected.

We now prove that every vertex $w$ that belongs to some 2-separator of $G$ 
is an endpoint of a block in $\fset$, or it is a degree-$2$ vertex, and it has a neighbor that serves as an endpoint of a block in $\fset$. By the properties of the SPQR trees, every such vertex 
$w$ is either an endpoint of some artificial edge $e$, lying in some graph $G_x$, for $x\in V(\tset')$, or it belongs to some graph $G_y$
associated with an $S$-node $y$. In the former case, let  $G_x$ be the graph for which the artificial edge $e$, containing $w$, is the parent edge. Then $B_x$ belongs to $\fset$, and $w$ is one of its endpoints. The only exception is when $x$ is a $P$-node, with a unique child $x'$. But then $x'$ cannot be a $P$-node, and $w$ is one of its endpoints. In the latter case, we consider the graph
$G_y$, containing $w$, where $y$ is an $S$-node. As before, we denote the vertices of $G_y$ by $a_1, \dots, a_s$.
Assume that $w=a_i$. If $i\neq \lceil(s-1)/2\rceil$, then $a_i$ is an endpoint of some path
$P_j$, $1\leq j\leq s-3$, and thus, it is an endpoint of the block $B^j_x$. Assume now that $i=\lceil(s-1)/2\rceil$. Then the vertex 
$w=a_{\lceil(s-1)/2\rceil}$ is connected to $w'=a_{\lceil(s-1)/2\rceil + 1}$ and $w''=a_{\lceil(s-1)/2\rceil - 1}$ by edges $e'$ and $e''$, respectively, in $G_x$. If either of these edges is an artificial edge, then there is a block $B\in \fset$, such that $I(B)=(w,w')$, or $I(B)=(w,w'')$. Otherwise, if both edges are actual edges, then the degree of $w$ is $2$, and $w'$ is a neighbor of $w$ that serves as an endpoint of a block in $\fset$. 

Finally, we show that if a block $B$ has only one child $B'$ then $I(B)\neq I(B')$. Observe that if $B=B_x$ for a $P$ or $R$-node $x$ then $I(B)$ is the set of endpoints of the parent artificial edge, and $I(B')$ is the set of endpoints of the only child artificial edge. Therefore, if $I(B)=I(B')$, then $G_x$ has two parallel aritificial edges, and thus $x$ must be a $P$-node. However, we add a block $B_x$ associated with a $P$-node to $\fset$ only if it has more than one child. Now let $B$ be either $B_x$ or $B_x^i$ for some $S$-node $x$. Since all paths $P_i$ (defined for $x$) have distinct pairs of endpoints, and they differ from the endpoints of edges, it is straightforward that $I(B)\neq I(B')$.

In the proof, we did not consider the case where $G$ is the cycle graph. We now briefly address this case. 
Denote the vertices of the cycle by $a_1,\dots, a_s$.  
We create blocks $B_1, \dots, B_{s-4}$ defined by  $B_i = \{a_{1 + \lfloor i/2 \rfloor}, \dots, a_{s - 1 -\lceil i/2 \rceil}\}$.
Additionally, we create a block $\hat B = \{a_{s-1}, a_s\}$ if $s > 3$.
It is straightforward to verify that this family of blocks together with $G$ satisfies the conditions 
of the theorem.






\fi

\section{Proof of Theorem~\ref{thm: good planar drawing}}\label{sec:finding a planar drawing}


We subdivide the sets of irregular vertices and edges into several subsets, that are then bounded separately. 
We start by defining the following sets of vertices and edges.
\begin{align*}
S_1 & = \{u\in V(\H) : u \mbox{ is a $1$-separator in } \H\}\\
E_1 & = \{e\in E(\H) : e \mbox{ is incident on some } u\in S_1\}
\end{align*}
Let $\cset$ be the set of all $2$-connected components of $\H$.
For every 2-connected component $X\in \cset$, we define
\iffull
\begin{align*}
S_2 (X)& = \{u\in V(X)\setminus S_1: \exists v\in V(X) \mbox{ s.t. $(u,v)$ is a $2$-separator in } X\}\\
E_2 (X)& = \{e\in E(X) : e \mbox{ has both end-points in } S_2(X)\}
\end{align*}
\fi
Let $S_2 = \cup_{X\in \cset} S_2(X)$ and $E_2 = \cup_{X\in \cset} E_2(X)$. 
We start by showing that the number of vertices and edges in sets $S_1$ and $E_1$, respectively, is small, in the next lemma, whose proof appears in Section~\ref{subsec:lemma1}.

\begin{lemma}[Irregular 1-separators]\label{lem:irregular1} We can bound the sizes of sets $S_1$  and $E_1$ as follows:
$|S_1|=O(|E^*|)$ and $|E_1| = O(\dmax \cdot |E^*|)$.
Moreover, $\sum_{C\in \cset} |S_1 \cap V(C)| \leq 9|E^*|$.
\end{lemma}

Next, we show that for {\bf any} planar drawing $\psi$ of $\H$, the number of irregular vertices and edges that do not belong to sets $S_1\cup S_2$, and $E_1\cup E_2$, respectively is small, in the next lemma, whose proof appears in Section~\ref{subsec:lemma2}.
Given any drawing $\phi$ of any graph $H$,
we denote by $\pcro_{\phi}(H)$ the number of pairs of crossing edges in the drawing $\phi$ of $H$. Clearly, $\pcro_{\phi}(H)\leq\cro_{\phi}(H)$ for any drawing $\phi$ of $H$.

\begin{lemma}\label{lem:irregular3}
Let $H$ be any planar graph, and let the sets $S_1,S_2$ of vertices and the sets $E_1,E_2$ of edges be defined as above for $H$. Let $\phi$ be an arbitrary drawing of $H$ and $\psi$ be a planar drawing of $H$.
Then
\iffull
$$
|\irreg_V(\psi, \phi)\setminus (S_1\cup S_2)| + |\irreg_E(\psi, \phi)\setminus (E_1\cup E_2)| = O(\pcro_{\phi}(H))=O(\cro_{\phi}(H)).
$$
\fi
\end{lemma}
Finally, we need to bound the number of irregular vertices in $S_2$ and irregular edges in $E_2$. The bound does  not necessarily hold for {\bf every} drawing $\psi$. However, we show how to efficiently find a planar drawing, for which we can bound this number, in the next lemma.

\begin{lemma}[Irregular 2-separators]\label{lem:irregular2}
Let $\G$, $\H$, $E^*$ and $\bphi$ be as in Theorem \ref{thm: good planar drawing}.
Given $\G$, $\H$ and $E^*$ (but not $\bphi$), 
we can efficiently compute a planar drawing $\bpsi$ of $\H$, such that
\[
|\irreg_V(\bpsi, \bphi_{\H})\cap S_2| =O(\optcro{\G} + |E^*|).\]
and
\[|\irreg_E(\bpsi, \bphi_{\H})\cap E_2|=O(\dmax)(\optcro{\G} + |E^*|)
\]
\end{lemma}

Theorem \ref{thm: good planar drawing} then immediately follows from Lemmas \ref{lem:irregular1}, \ref{lem:irregular3}, and \ref{lem:irregular2}, where we apply Lemma~\ref{lem:irregular3} to the drawings $\bpsi$ and $\bphi_{\H}$ of the graph $\H$.
In the following subsections, we present the proofs of Lemmas~\ref{lem:irregular1}, \ref{lem:irregular3} and \ref{lem:irregular2}.

\subsection{Proof of Lemma~\ref{lem:irregular1}}\label{subsec:lemma1}
Consider the following tree $\tset$: the vertices of $\tset$ are $S_1\cup\set{v_C\mid C\in \cset}$, and there is an edge between $v_C$ and $u\in S_1$ iff $u\in V(C)$.
We partition the set $\set{v_C:C\in \cset}$ into three subsets: set $D_1$ contains the leaf vertices of $\tset$, set $D_2$ contains vertices whose degree in $\tset$ is $2$, and set $D_3$ contains all remaining vertices.
Since $\G$ is 3-connected, for every component $C$ with $v_C\in D_1\cup D_2$, there is an edge $e\in E^*$ with 
one end-point in $C$. We \textit{charge} edge $e$ for $C$. 
Clearly, we charge each edge at most twice (at most once for each of its endpoints), and
therefore, $|D_1|+|D_2| \leq 2|E^*|$. Since the number of vertices of degree greater than $2$ is bounded by the number of leaves in any tree, we get that $|D_3|\leq |D_1|\leq 2|E^*|$, and so $|\cset|\leq |D_1|+|D_2|+|D_3|\leq 4|E^*|$.  Since the parent of every vertex $u\in S_1$ in the tree is a vertex of the form $v_C$ for $C\in \cset$, this implies that $|S_1|\leq |\cset|+1\leq 4|E^*|+1$, and $|E_1|\leq \dmax|S_1|\leq O(\dmax)|E^*|$.

We now bound the sum $\sum_{C\in \cset} |S_1 \cap V(C)|$. The sum equals the number of pairs $(C,u)$, where $C\in\cset$ and $u\in S_1\cap V(C)$. The number of such pairs in the tree $\tset$ is bounded by the number of edges in the tree, which in turn is bounded by the number of vertices, $|S_1|+|\cset|\leq 8|E^*|+1$.
This finishes the proof of Lemma~\ref{lem:irregular1}.

\subsection{Proof of Lemma~\ref{lem:irregular3}}\label{subsec:lemma2}
\label{sec:irregular-vertices-edges}

In this section we bound on the number of irregular vertices and irregular edges that do not belong to 
$S_1\cup S_2$ and $E_1\cup E_2$, respectively. Lemma~\ref{lem:irregular-vertices} bounds the number of irregular vertices and Lemma~\ref{lem:irregular-edges} the number of irregular edges.

\begin{lemma}
\label{lem:irregular-vertices}
Let $\phi$ be an arbitrary drawing of $H$ and 
let $\psi$ be a planar drawing of $H$. Let $S = S_1\cup S_2$. Then
\begin{equation}
|\irreg_V(\psi, \phi) \setminus S| \leq 12 \pcro_{\phi}(H) \leq 12 \cro_{\phi}(H). \label{eq:bound-conflicts}
\end{equation}
\end{lemma}
\begin{proof}
Note first that we may assume that no two adjacent edges cross each other
in the drawing $\varphi$. Indeed, if the images of two edges incident to a vertex $u$ cross, we can uncross their drawings, possibly
changing the cyclic order of edges adjacent to $u$, and preserving the cyclic order for all other
vertices. The right-hand side of (\ref{eq:bound-conflicts})  will then decrease by
12, and the left hand side by at most 1, so we only strengthen the inequality.
We can also assume w.l.o.g. that the graph $H$ is $2$-connected: otherwise, if $\cset$ is the set of all $2$-connected components of $H$, then,  
since $\pcro_{\phi} (H) \geq \sum_{C\in \cset} \pcro_{\phi}(C)$, it is enough to prove the inequality~(\ref{eq:bound-conflicts}) for each component $C\in\cset$ separately. So we assume below that $H$ is 2-connected.

Consider some vertex $u \in \irreg_V(\psi, \phi) \setminus S$. Let $F$ be the face of $H \setminus\set{u}$ 
that contains the image of $u$ in the drawing $\psi_{H\setminus\set{u}}$. Note that graph  $H \setminus\set{u}$ is 2-connected: otherwise,
if $v$ is a vertex separator of $H\setminus\set{u}$ then $\{u,v\}$ is a $2$-separator for $H$, contradicting
the fact that $u\notin S$. Therefore, the boundary of $F$ is a simple cycle, that we denote by $\gamma$. Let $v_1, \dots, v_{\kappa}$ be the neighbors of $u$ in the order induced by $\gamma$. Vertices $v_i$ partition $\gamma$ into $\kappa$ paths $P_1, \dots, P_{\kappa}$, where path $P_i$ connects vertices $v_i$ and $v_{i+1}$ (we identify indices $\kappa+1$ and $1$). Let $F_i$ be the face of the planar drawing $\psi$, that is
bounded by $(u,v_i)$, $P_i$ and $(v_{i+1},u)$. Note that since for all $i\neq j$, the two paths $P_i$ and $P_j$ do not share
any internal vertices, the total number of vertices that the boundaries of $F_i$ and $F_j$ for $i\neq j$ share is at most $3$, with the only possibilities being $u$, $v_i$ and $v_{i+1}$ (the endpoints of $P_i$).

Consider the graph $W$ formed by $\gamma$, $u$, and edges $(u, v_i)$, for $1\leq i\leq \kappa$. This graph is homeomorphic to the wheel 
graph on $\kappa$ vertices. In any planar embedding of $W$, the ordering of the vertices $\set{v_i}_{i=1}^{\kappa}$ is $(v_1,\dots, v_{\kappa})$. So if the drawing 
$\varphi_W$ of $W$ is planar, then the circular ordering of the edges adjacent to $u$ in $\phi$ is
$((u,v_1),\dots, (u,v_{\kappa}))$ -- the same as in $\psi$, up to orientation. Therefore, if
$u\in \irreg_V(\psi, \phi)$ then either there is a pair $P_i,P_j$ of paths, with $i\neq j$, whose images cross in $\phi$, or an image of an edge $(u, v_i)$ crosses a path $P_j$ (recall that we have assumed that no two 
edges $(u, v_i)$ and $(u, v_j)$ cross each other; all self-intersections of paths $P_i$ can be removed without changing the rest of the embedding).
We say that this crossing point pays for $u$.
Thus every irregular vertex is paid for by a crossing in the drawing $\varphi$.
It only remains to show that every pair of crossing edges pays for at most 12 vertices.

Suppose that $u$ is paid for by a crossing of edges $e_1$ and $e_2$. 
For each edge $e\in\{e_1,e_2\}$, there is a face $F^e$ (in the embedding of $\psi$)
such that $e$ and $u$ lie on the boundary of $F^e$: if $e$ lies on path 
$P_i$ then $F^e = F_i$; if $e = (u, v_j)$ then $F^e$ is either 
$F_{j-1}$ or $F_j$. Since in the latter case we have two choices
for $F^e$, we can choose distinct faces $F^{e_1}$ and $F^{e_2}$.
Therefore, if a crossing of edges $e_1$ and $e_2$ pays for a vertex $u\in\irreg_V(\psi, \phi) \setminus S$,
then there are two distinct faces $F^{e_1}$ and $F^{e_2}$ in $\psi$, incident to $e_1$ and 
$e_2$ respectively, such that $u$ lies on the intersection of the boundaries of $F^{e_1}$ and
$F^{e_2}$.
We say that the pair of faces $F^{e_1}$ and $F^{e_2}$ is the witness 
for the irregular vertex $u$. Since the boundaries of $F^{e_1}$ and $F^{e_2}$ may share at most $3$ vertices that do not belong to $S_2$, the pair $(F^{e_1}, F^{e_2})$ is a witness
for at most 3 irregular vertices. Since each edge $e_i$ is incident to at most two faces in $\psi$,
there are at most $4$ ways to choose $F^{e_1}$ and $F^{e_2}$, and for each such choice $(F^{e_1}, F^{e_2})$ is a witness for at most 3 irregular vertices. We conclude that each pair of edges that cross in $\phi$ pays for at most $12$ irregular vertices.
\end{proof}

\begin{lemma}\label{lem:irregular-edges}
Let $\varphi$ be an arbitrary drawing of $H$ and $\psi$ be its planar drawing. 
Let $E_S = E_1\cup E_2$. Then
$$
|\irreg_E(\psi, \phi) \setminus E_S| \leq 8 \pcro_{\varphi}(H) \leq 8 \cro_{\varphi}(H).
$$
\end{lemma}
\begin{proof}
We can assume w.l.o.g. that there are no vertices of degree $2$ in $H$, by iteratively removing such vertices $u$, and replacing the two edges incident on $u$ with a single edge. This operation may decrease the number of irregular edges by at most factor $2$, and can only decrease the number of pairs of crossing edges.
Similarly to the proof of Lemma~\ref{lem:irregular-vertices}, we assume that the graph $H$ is $2$-connected:
otherwise, we can apply the argument below separately to each 2-connected component.


We say that the orientation of a regular vertex $u$ is positive,
if the ordering of the edges incident to $u$ is the same  in $\phi$ and $\psi$, including the flip.
If the flips in $\phi$ and $\psi$ are opposite, we say that the orientation is negative.
For every irregular edge $e$, the orientation of one of its endpoints is positive, and of the other is negative.

Consider an irregular edge $e = (u,v) \in \irreg_E(\psi, \phi)\setminus E_S$, and assume w.l.o.g. that the orientation of $u$ is positive and the orientation of $v$ is negative.
Let $F_1$ and $F_2$ be the two faces incident to $e$ in the embedding $\psi$. 
Since $H$ is 2-connected, the boundaries of $F_1$ and $F_2$ are simple cycles. Denote them by $C_1$ and $C_2$. 
Let $P_i = C_i \setminus\set{e}$ be the sub-path of $C_i$ that connects $u$ to $v$. We now prove that $P_1$ and $P_2$ do not share any
vertices except for $u$ and $v$. 
Indeed, assume for contradiction that a vertex $w\notin \{u, v\}$ lies
on both $P_1$ and $P_2$. Since $e\notin E_S$, either $u$ or $v$ (or both) are not in $S$. Assume w.l.o.g. that $u\notin S$.
 We draw two curves, connecting $w$ to the middle
of the edge $e$ inside the planar drawing $\psi$ of $H$; one of the two curves lies inside $F_1$ and the other lies inside $F_2$.
The union of the two curves defines a cycle that separates $H\setminus\set{w}$ into two pieces, with $u$ belonging to one piece and 
$v$ to the other. Denote these pieces by $B_u$ and $B_v$, respectively. (We assume that $w\in B_u$, $w\in B_v$). We will now show that $(u,w)$ is a $2$-separator for $H$, leading to a contradiction. Observe first that  since the degrees of $u$ and $v$ are at least $3$, and the separating cycle only crosses one edge of $H$ (the edge $e$), both $B_u$ and $B_v$ contain at least $3$ vertices each. Since every path from 
$B_u$ to $B_v$ must cross the separating cycle,  each such path either contains
the vertex $w$ or the edge $e$. Therefore, $(u,w)$ is a $2$-separator for $H$, contradicting our assumption that $u\notin S$.

We say that the pair of faces $(F_1, F_2)$ is the witness for the irregular edge $e$. From the above discussion,
each pair of faces is a witness for at most one irregular edge.

\begin{figure}
\begin{center}
\scalebox{0.65}{\includegraphics{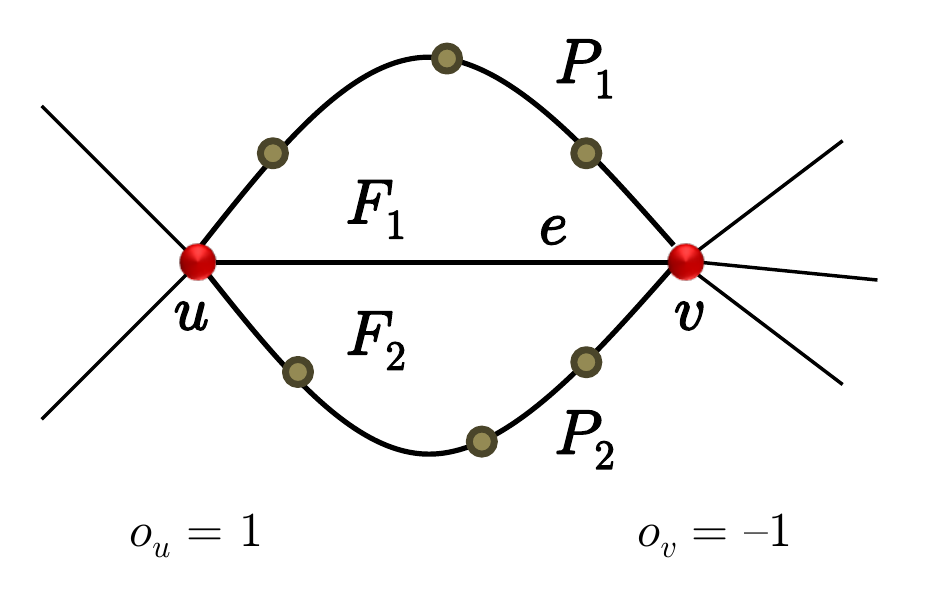}}
\caption{Graph $H$, irregular edge $e$, paths $P_1$ and $P_2$.\label{fig:irregular-edge}}
\end{center}
\end{figure}
Let $o_u^\mu$ be the orientation --- either clockwise or counterclockwise ---
in which paths paths $P_1$, $e$, and $P_2$ leave $u$ in the embedding $\mu$
(where $\mu$ is either $\phi$ or $\psi$). If the orientation is clockwise
$o_u^\mu = 1$; otherwise $o_u^\mu = -1$. Similarly, we define $o_v^\mu$.
Note that in any embedding $\mu'$ in which paths $P_1$, $e$ and $P_2$
do not cross each other, $o^{\mu'}_u = - o^{\mu'}_v$. 
In particular, since $\psi$ is a planar embedding, $o_u^\psi=-o_v^\psi$. But since the orientation of $u$ is positive, and the orientation of $v$ is negative, 
$o_u^\phi =o_v^\phi$.
Therefore, there is a pair $(e_1,e_2)$ of crossing edges in $\phi$, where either $e_1 \in P_1$, $e_2 \in P_2$;
or $e_1 \in P_1$, $e_2 = e$; or $e_1 = e$ and $e_2 \in P_2$. 
We say that the crossing of $e_1$ and $e_2$ pays for the irregular edge $e$.
The edges $e_1$ and $e_2$ lie on the boundaries of $F_1$ and $F_2$ respectively.
Similarly to the previous lemma, given two crossing edges $e_1$ and $e_2$, there are at most $4$ ways to choose the 
faces $(F_1, F_2)$ incident to them, and each such pair of faces is a witness for at most one edge. 
Therefore, each pair of crossing edges pays
for at most 4 irregular edges. We conclude that the number of irregular edges is bounded by $4\cro_{\phi}(H)$. Replacing the edges back by the original $2$-paths increases the number of irregular edges by at most factor $2$, as each irregular $2$-path contains two irregular edges.
\end{proof}


\subsection{Proof of Lemma \ref{lem:irregular2}} 
We start with a high level overview of the proof. Assume first that the graph $\H$ is 2-connected. We can then use Theorem~\ref{thm: block decomposition} to find a laminar block decomposition $\fset$ of $\H$. Moreover, each vertex $v\in S_2$ is either an endpoint of a block in $\fset$, or it is a neighbor of an endpoint of a block in $\fset$. Therefore, $|S_2|$ is roughly bounded by $O(|\fset|\cdot \dmax)$. On the other hand, since the graph $\G$ is 3-connected, each block $B\in \fset$ must contain an endpoint of an edge from $E^*$ as an inner vertex, that can be charged for the block $B$, for its endpoints, and for the neighbors of its endpoints. This approach would work if we could show that every edge $e\in E^*$ is only charged for a small number of blocks. This unfortunately is not necessarily true, and an edge $e\in E^*$ may be charged for many blocks in $\fset$. However, this may only happen if there is a large number of nested blocks, all of which contain the same endpoint of the edge $e$. We call such set of blocks a ``tunnel''. We then proceed in two steps. First, we bound the number of blocks of $\fset$ that do not participate in such tunnels, by charging them to the edges of $E^*$, as above. Next, we perform some local changes in the embeddings of the tunnels (by suitably flipping the embedding of each block of the tunnel), so that we can charge the number of irregular vertices that serve as endpoints of blocks participating in the tunnels to the crossings in $\bphi$.

We now proceed with the formal proof.
We start with an arbitrary planar drawing $\bpsi_{init}$ of $\H$. Let $\cset$ be the set of all $2$-connected components of $\H$. We consider each component $X\in \cset$ separately.
For a component $X\in \cset$,
let $\cro_{\bphi}(\G,X)$ 
denote the number of crossings in $\bphi$ in which edges of $X$ participate, and let $E^*(X)$ denote the subset of edges of $E^*$ that have at least one endpoint in $X$.
We will modify $\bpsi_{init}$ locally on each 2-connected component $X\in \cset$ and obtain a planar
drawing $\bpsi$ of $\H$ such that
%
\iffull
\begin{align*}
|\irreg_V(\bphi_{\H},\bpsi)\cap S_2(X)| &=  O(\cro_{\bphi}(\G,X)+|E^*(X)| + |S_1 \cap X|)\\
|\irreg_E(\bphi_{\H},\bpsi)\cap E_2(X)| &=  O(\dmax (\cro_{\bphi}(\G,X)+|E^*(X)| + |S_1 \cap X|)).
\end{align*}
\fi%
Summing up over all $X\in \cset$, and using Lemma~\ref{lem:irregular1} gives the desired bound.
Since we  guarantee that the modifications of $\bpsi_{init}$ are restricted to $X$,
we can modify the 2-connected components $X\in \cset$ independently to obtain 
the final desired drawing.

Fix a 2-connected component $X\in \cset$. If $X$ is 3-connected then 
$S_2(X) = E_2(X) = \emptyset$ and there is nothing to prove. So we assume below that $X$ is not
3-connected. We compute the laminar block decomposition $\fset(X)$ and the corresponding decomposition tree $\tset(X)$ for $X$, given by Theorem \ref{thm: block decomposition}. For convenience, we use $\fset'(X)=\fset(X)\setminus \set{X}$ to denote the set of all blocks in $\fset(X)$, excluding the whole component $X$.
We now proceed in three steps. Our first step is to explore some structural properties of the blocks $B\in \fset(X)$. We will use these properties, on the one hand, to bound the number of blocks that do not participate in tunnels, and on the other hand, to find the layout of the tunnels.
In the second step, we define the subsets of blocks that we can charge to the edges in $E^*$. We then charge some of the vertices in $S_2$ and edges in $E_2$ to these blocks. In the last step, we define tunnels, to which all remaining blocks belong, and we show how to take care of them.

\subsubsection*{Step 1: Structural properties of blocks}
Consider some block $B \in \fset'(X)$, with endpoints $u$ and $v$. Since $X$ is 2-connected, 
there is a path $P_{out}^B:u\connect v$ in $(X\setminus B)\cup\set{u,v}$.
Moreover, if $B'$ is the parent of $B$ in $\tset(X)$, whose endpoints are $u'$ and $v'$, we can ensure that $P_{out}^{B'}\sse P_{out}^B$, as follows. Consider the graph $B^*$ obtained from $B'$ after we remove all inner vertices of $B$ from it. Since $X$ is $2$-connected, so is $B'$. Therefore, there are $2$ vertex disjoint paths in $B^*$, connecting the vertices in $\set{u',v'}$ to the vertices in $\set{u,v}$. We assume w.l.o.g. that these paths are $P_1: u\connect u'$ and $P_2:v\connect v'$. We can then set $P_{out}^B=(P_1,P_{out}^{B'},P_2)$ (see Figure~\ref{fig: paths pout}). Therefore, from now on we assume that if $B'$ is the parent of $B$, then $P_{out}^{B'}\sse P_{out}^B$.

\begin{figure}[h]
\begin{center}
\scalebox{0.3}{\rotatebox{0}{\includegraphics{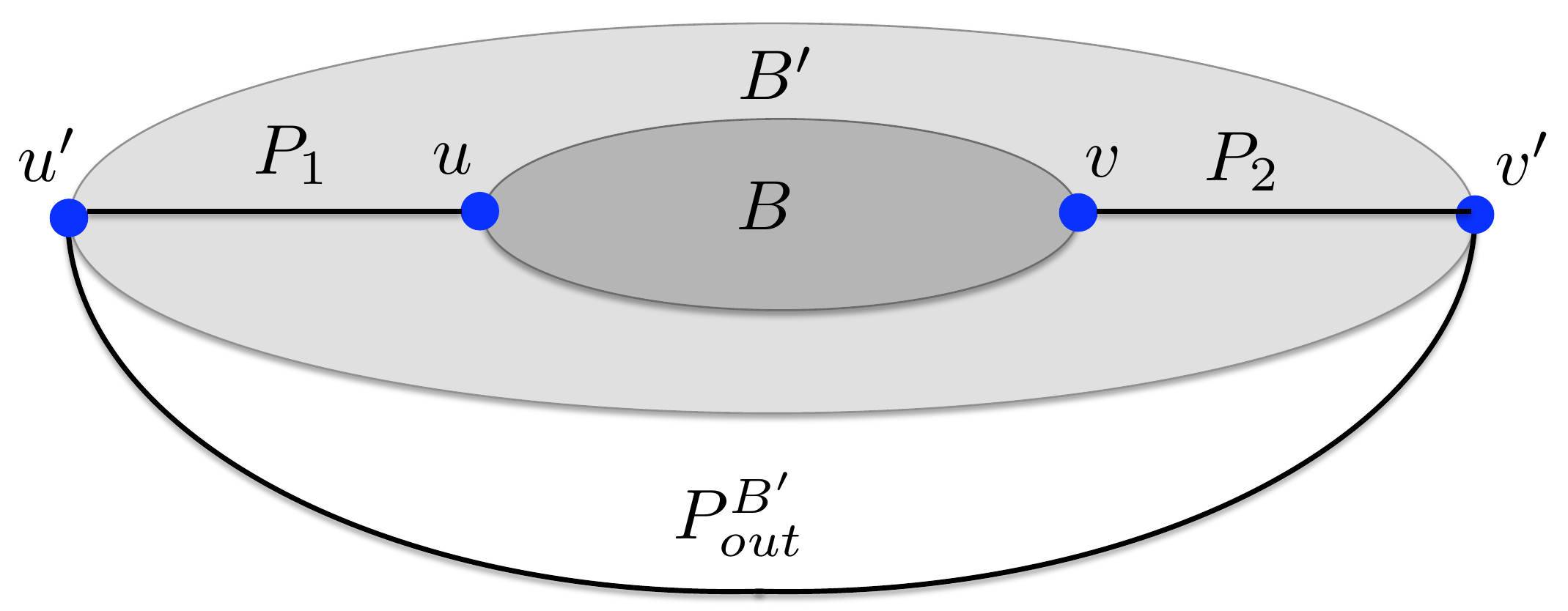}}} \caption{Paths $P^B_{out}$, $P^{B'}_{out}$.} \label{fig: paths pout}
\end{center}
\end{figure}

Since we have assumed that $\G$ is 3-vertex connected, for every block $B\in \fset'(X)$, there is also a path $Q$ in $\G\setminus\set{u,v}$, connecting
an {\bf inner} vertex of the block $B$, with an {\bf inner} vertex of the path $P_{out}^B$. 
Let $x_B$ be the last vertex on $Q$ that belongs to $B$ and $y_B$ 
be the first vertex on $Q$ that belongs to $P^B_{out}$ (notice that $y_B\neq u,v$, since $Q$ does not contain $u$ or $v$). We denote the segment of $Q$ 
between $x_B$ and $y_B$ by $P_0^B$, and we call the vertex $x_B$ \textit{the connector vertex} for 
the block $B$. 

Note that if $B''$ is a child block of $B$ and $x_B$ is an inner vertex of $B''$ as well, then since $P_{out}^{B}\sse P_{out}^{B''}$, we can choose
$x_B$ to be the connector vertex of $B''$ as well, and use $P_0^{B''}=P_0^B$.  So we assume that each connector 
vertex $x$ appears contiguously in the tree $\tset$. That is, if $B$ is a descendant
of $B_1$ and an ancestor of $B_2$ and $x_{B_1} = x_{B_2}$, then $x_B = x_{B_1} = x_{B_2}$. We also assume that in this case $P_0^{B_1}=P_0^B=P_0^{B_2}$.
We denote the segment of $P_{out}^B$ between $u$ and $y_B$ by $P_{1,out}^B$ and the segment
between $y_B$ and $v$ by $P_{2, out}^B$. 

Since $X$ is 2-connected, there are 
two vertex disjoint paths between $x_B$ and $y_B$ in $X$. One of them must pass through $u$
and the other through $v$. We denote the segment between $u$ and $x_B$ of the former path
by $P^B_{1,in}$ and the segment between $x_B$ and $v$ of the latter path by $P^B_{2,in}$.
Let $P^B_{in}$ be the concatenation of $P_{1,in}^B$, $P_{2,in}^B$.
Note that the paths $P_0^B$, $P_{1,in}^B$, $P_{2,in}^B$, $P_{1,out}^B$ and $P_{2,out}^B$
do not intersect, except at endpoints (see Figure~\ref{fig: paths pout2}). We emphasize that $x_B$ is an {\bf inner} vertex of $B$, and $y_B$ is an {\bf inner} vertex on path $P_{out}^B$ --- a fact that we use later.

\begin{figure}[h]
\begin{center}
\scalebox{0.3}{\rotatebox{0}{\includegraphics{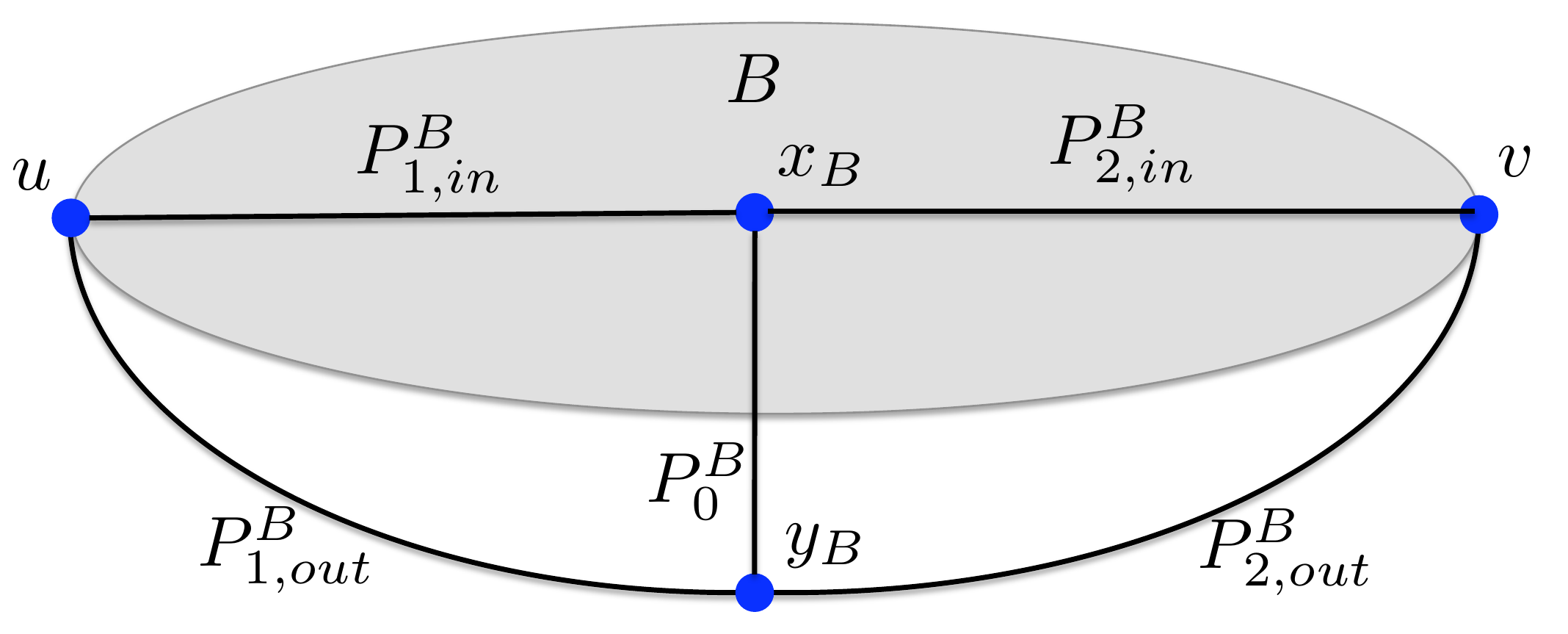}}} \caption{Paths $P_0^B$, $P_{1,in}^B$, $P_{2,in}^B$, $P_{1,out}^B$ and $P_{2,out}^B$. Vertex $x_B$ is an inner vertex of $B$, and vertex $y_B$ is an inner vertex of $P_{out}^B$. All five paths are non-empty and completely disjoint except for their endpoints.} \label{fig: paths pout2}
\end{center}
\end{figure}

For each component $X\in\cset$, let $\sset_X$ be the union of 
(i) the set $S_1\cap X$ and (ii) the set of vertices of $X$ incident to edges of $E^*$. 
 Using Lemma~\ref{lem:irregular1},

\begin{equation}\label{eq: bound SX}
\sum_{X\in \cset} |{\cal S}_X| \leq \sum_{X\in \cset} (|E^*(X)| + |S_1 \cap X|) \leq O(|E^*|).
\end{equation}

We now show that for each block $B\in \fset'(X)$, the connector vertex $x_B\in \sset_X$.
Indeed, consider the first edge $(x_B, z)$ of the path
$P_0^B$. If $z\in X$, then $(x_B, z) \in E^*(X)$, as by the definition
of the block, no edges of $X$ connect inner vertices of $B$ to $X\setminus B$. Otherwise, if $z\notin X$, then $x_B$ must be a $1$-separator, so $x_B\in S_1$. 


Finally, we study structural properties of chains of nested blocks. We also introduce a notion of a \textit{simple} block, and show that all non-simple blocks contain a certain useful structure.

\begin{Definition}
Let $B_1\in \fset'(X)$ be any block, whose endpoints are denoted by $u_1$ and $v_1$. We say that $B_1$ is a \emph{simple block} iff it contains exactly three vertices, $u_1,v_1$, and $u_2$, and has exactly one child in $\tset(X)$, denoted by $B_2$ (assume w.l.o.g. that the endpoints of $B_2$ are $(u_2,v_1)$). Moreover, $B_1$ is obtained by adding exactly one edge, $(u_1,u_2)$, to $B_2$ (see Figure~\ref{fig:irregular-block}). If $B_1\in \fset'(X)$ has exactly one child in $\tset(X)$, but it is not a simple block, then we say that it is \emph{complex}.
\end{Definition}

\begin{figure}
\begin{center}
\scalebox{0.6}{\includegraphics{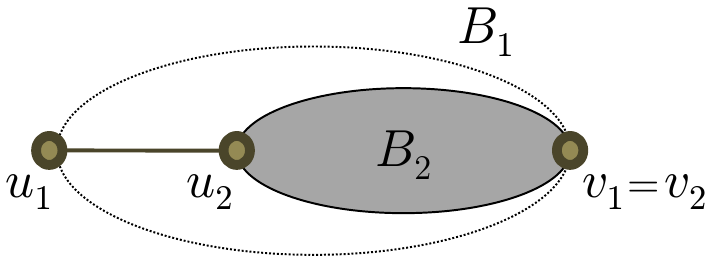}}
\caption{A simple block $B_1$.
\label{fig:irregular-block}}
\end{center}
\end{figure}

\begin{figure}
\begin{center}
\scalebox{0.6}{\includegraphics{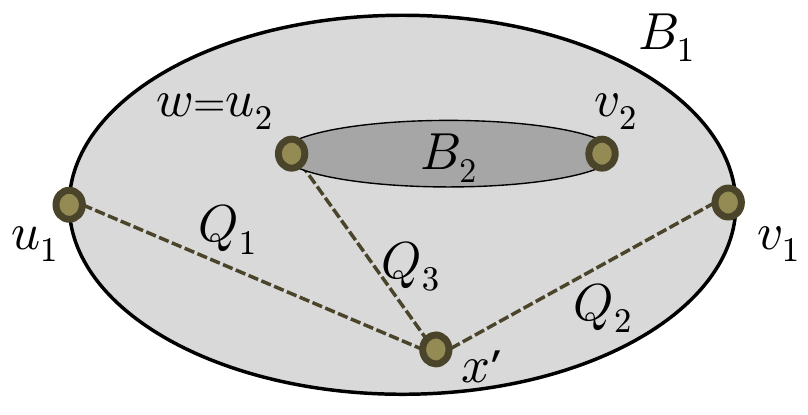}}
\caption{A complex block. Paths $Q_1,Q_2,Q_3$ are pairwise vertex disjoint, except for containing $x'$ as a common endpoint.}
\label{regular-block}
\end{center}
\end{figure}

We need the following two claims.

\begin{claim}
\label{cor:regular-block-in-a-chain}
Consider a chain of 5 nested blocks: $B_1$, $B_2$, $B_3$, $B_4$ and $B_5$, where $B_{i+1}$ is the only child of $B_i$ (for $i\in\{1,\dots, 4\}$).
Assume that no vertices in $V(B_1)\setminus V(B_5)$ have degree 2 in $X$. Then one of the blocks $B_1$,$B_2$,$B_3$, or $B_4$ is complex.
\end{claim}
\begin{proof}
Notice that from the definition of simple blocks, if all blocks $B_1,B_2,B_3,B_4$ are simple, at least one vertex $z\in\set{u_2,v_2,u_3,v_3}\setminus V(B_5)$ must have degree $2$ in $X$ (where $u_i$ and $v_i$ are endpoints of $B_i$), contradicting the fact that $V(B_1)\setminus V(B_5)$ cannot contain such vertices.
\end{proof}

\begin{claim}\label{claim: irregular}
Suppose that a non-simple block $B_1\in \fset'(X)$ has exactly one child $B_2$ in $\tset(X)$. Denote the endpoints of $B_1$ by $u_1$ and $v_1$, and the endpoints of $B_2$ by $u_2$ and $v_2$. 
Then for every vertex $x'\in V(\tilde B_1)\setminus \{u_1, v_1\}$, 
there are three paths $Q_1:x'\connect u_1$, $Q_2:x'\connect v_1$, and $Q_3:x'\connect w$, with $w\in\set{u_{2},v_{2}}$, and all three paths are contained in $\tilde B_1 \setminus\set{(u_{2}, v_{2})}$. Moreover, $Q_1,Q_2$ and $Q_3$ do not share any vertices, except for the vertex $x'$ that serves as their endpoint. (See Figure~\ref{regular-block} for an illustration.)
\end{claim}

\begin{proof}
Since $B_1$ has only one child, $u_2\notin\{u_1, v_1\}$ or $v_{2}\notin \{u_1, v_1\}$ 
(or both). Let us assume w.l.o.g. that $u_2\notin \{u_1, v_1\}$.
 In particular, $B_1$ contains at least $3$ vertices.
 
We consider two cases. Assume first that $\tilde{B}_1$ contains exactly $3$ vertices. Then these vertices must be $u_1,v_1$ and $u_{2}$. The only valid choice for the vertex $x'$ is $x'=u_2$. From the definition of blocks, $B_1$ cannot contain the edge $(u_1,v_1)$. But since it is connected, it must contain the edge $(u_1,u_{2})$. Therefore, the only way for $B_i$ not to be simple (since we have assumed that $\G$ contains no parallel edges) is if $B_1$ contains the edge $(u_{2},v_1)$. But in this case, we get the following three paths: $Q_1=(u_1,u_2), Q_2=(u_2,v_1)$, and $Q_3=\emptyset$.

Assume now that $\tilde{B}_1$ contains at least $4$ vertices.  From Theorem~\ref{thm: block decomposition}, the graph $\tilde{B}_1'$ is 3-connected.  Let $x'$ be an arbitrary inner vertex of $B_1$. Assume first that $x'\notin\{u_{2}, v_{2}\}$.
Recall that the Fan Lemma states that for every $r$-connected graph
$A$, a vertex $a$ in $A$ and a set of $r$
vertices $B\subset V(A)\setminus\{a\}$, there exist $r$
paths that connect $a$ to vertices of $B$ that have no
common vertices other than $a$.
We apply the Fan Lemma in graph $\tilde{B}_1'$ to $x'$ and $\{u_1,v_1, u_{2}\}$.
Let $Q_1$ be the resulting path between $x'$ and $u_1$,
$Q_2$ the path between $x'$ and $v_1$, and
$Q_3'$ the path between $x'$ and $u_{2}$. 
Note that paths $Q_1$ and $Q_2$ do not contain 
the artificial edge $(u_{2},v_{2})$, as otherwise they would
contain $u_{2}$. 
Notice also that none of the three paths contains the artificial edge $(u_1,v_1)$, as this would violate their disjointness. 
Finally, let $Q_3$ be equal to either $Q_3'$, if $Q_3'$ does not visit $v_{2}$,
or the segment of $Q_3'$ between $x'$ and $v_{2}$, if it does (the latter can only happen if $v_2\not\in\set{u_1,v_1}$).
We have thus constructed the required paths $Q_1$, $Q_2$ and $Q_3$.
Assume now that $x'\in \{u_{2}, v_{2}\}$.
Since $\tilde B_1'$ is 3-vertex connected (and $\tilde B_1' \neq K_3$),
the graph $\tilde B_1'\setminus \set{(u_{2},v_{2})}$ is 2-vertex connected. We again apply the Fan Lemma
to $w$ and $\{u_{1}, v_{1}\}$ in this graph and find the desired paths $Q_1$ and $Q_2$.
We let $Q_3$ to be the trivial path of length $0$.
\end{proof}

\subsubsection*{Step 2: Blocks we can pay for}
Fix a $2$-connected component $X\in \cset$.
In this step, we define three subsets $\rset_1(X),\rset_2(X),\rset_3(X)$ of $\fset(X)$, and bound the number of blocks contained in them. We also define a subset $\tilde{S}_2\sse S_2$ of vertices and a subset $\tilde{E}_2\sse E_2$ of edges, that can be charged to these blocks. The remaining blocks of $\fset(X)$ will be partitioned into structures called tunnels, and we take care of them in the next step.

\paragraph{Set $\mathbf{\boldsymbol{\rset}_1(X)}$:} Let $\rset_1(X)$ denote the set of blocks $B\in \fset(X)$, 
such that $B$ is either the root of $\tset(X)$, or it is one of its leaves, 
or it has a degree greater than $2$ in $\tset(X)$,
or it contains a vertex from ${\cal S}_X$ that does not belong to any of its child blocks.
We also add five immediate ancestors of every such block to $\rset_1(X)$.

\begin{claim} $\sum_{X\in\cset}|\rset_1(X)|=O(|E^*|)$.
\end{claim}
\begin{proof}
Denote the number of leaves in $\tset(X)$ by $L_X$.
For each leaf block $B$, we charge the connector vertex $x_B\in {\cal S}_X$ for $B$.
For each non-leaf block $B$, such that $B$ contains a vertex $x\in \sset_X$ that does not belong to any of its children, we charge $x$ for $B$ (even if $x_B \neq x$).
Since $\fset(X)$ is a laminar family, it is easy to see that each vertex $x\in \sset_X$ is charged at most once.
The number of vertices of degree at least $3$ in ${\cal T}(X)$ is at most $L_X-1$. 
By adding five ancestors of each block, we increase the size of $\rset_1(X)$
by at most a factor of $5$. 
Therefore, $\sum_{X\in\cset} |\rset_1(X)| \leq \sum_{X\in \cset} O(|{\cal S}_X|) = O(|E^*|)$.
\end{proof}

\paragraph{Set $\mathbf{\boldsymbol{\rset}_2(X)}$:} 
Consider a vertex $x\in {\cal S}_X$.
Notice that the set of blocks $B\in\fset(X)$ with $x_B = x$ must be a nested set.
We add the smallest such block and its five immediate ancestors to $\rset_2(X)$.

\begin{claim} $\sum_{X\in\cset}|\rset_2(X)|=O(|E^*|)$.
\end{claim}
\begin{proof}
For each block $B\in\rset_2(X)$, we charge the connector vertex $x_B$ for $B$. 
By the definition of $\rset_2(X)$, each connector vertex pays for at most $6$ blocks.
Therefore, $\sum_X |\rset_2(X)| \leq \sum_X O(|{\cal C}_X|) = O(|E^*|)$.
\end{proof}

\paragraph{Set $\mathbf{\boldsymbol{\rset}_3(X)}$:} 
Note that the blocks of $\fset(X)$ that do not belong to $\rset_1(X) \cup \rset_2(X)$ 
all have degree exactly 2 in $\tset(X)$, and therefore the sub-graph of $\tset(X)$ induced by such blocks is simply a collection of disjoint paths. 
Consider some block $B\in \fset(X)\setminus (\rset_1(X)\cup \rset_2(X))$.
It has exactly one child in $\tset(X)$, that we denote by $B'$.
Let $u$ and $v$ be the endpoints of $B$, and let $u'$ and $v'$ be the endpoints of $B'$.
Consider the graph $\tilde B'$ obtained from $B$ by first replacing $B'$ with 
an artificial edge $(u',v')$ and then by adding a new artificial edge $(u,v)$.
By Theorem \ref{thm: block decomposition}, the graph $\tilde B'$ is $3$-vertex connected.
Therefore, it has a unique planar drawing $\pi_{\tilde{B}'}$.
We add $B$ to $\rset_3(X)$ iff the four vertices $u,v,u',v'$ {\bf do not lie} 
on the boundary of the same face in this drawing.

\begin{lemma}\label{lemma:r3}
 $\sum_{X\in \cset}\left(|\rset_3(X)|\right) = O(\cro_{\bphi}(\G)).$
\end{lemma}
\begin{proof}
Consider some block $B\in \rset_3(X)$. Denote $B_0 = B$, and for $i=1,\dots,5$,
let $B_{i}$ be the child of $B_{i-1}$ in $\tset(X)$. For each $i: 1\leq i\leq 5$, let $(u_i,v_i)$ denote the endpoints of the block $B_i$.
Since when we added a block to $\rset_1(X)$ or $\rset_2(X)$, we also added 
five its immediate ancestors to $\rset_1(X)$ or $\rset_2(X)$, respectively, each of the blocks $B_i$, for $0\leq i\leq 5$,
has a unique child, and moreover, for $i=1,\ldots,5$, $x_{B_i} = x_{B}$ and $P_0^{B_i}=P_0^B$.
Let $\hat{E}_B$ denote the edges of $B$ that do not belong to $B_5$, that is, $\hat{E}_B=E(B)\setminus E(B_5)$. We will show that for each $B\in\rset_3$, there is at least one crossing in $\bphi$, in which the edges of $\hat{E}_B$ participate. Since every edge may belong to at most $5$ such sets $\hat{E}_B$, it will follow that $|\rset_3(X)|\leq O(\cro_{\bphi}(X,\G))$, and $\sum_{X\in \cset}\left(|\rset_3(X)|\right) = O(\cro_{\bphi}(\G))$.
Therefore, it now only remains to show that for each block $B\in \rset_3(X)$, the edges of $\hat{E}_B$ participate in at least one crossing in $\bphi$. Assume for contradiction that this is not true, and let $B$ be the violating block. We will show that we can find a planar drawing of $\tilde{B}'$, in which the vertices $(u,v,u_1,v_1)$ all lie on the boundary of the same face, contradicting the fact that $B\in\rset_3(X)$.

We denote by $B^*$ the graph obtained from $B$ after we remove all inner vertices of $B_1$ and their adjacent edges from it. Notice that all edges of $B^*$ belong to $\hat{E}_B$. We also denote $x_B=x, y_B=y$ and $P_0^B=P_0$. Recall that for all $1\leq i\leq 5$, $x_i=x,y_i=y$ and $P_0^{B_i}=P_0$.
Recall that by definition, $x$ is an {\bf inner} vertex on $P_{in}^{B_i}$ for all $1\leq i\leq 5$, and $y$ is an {\bf inner} vertex on $P_{out}^{B}$.

We start with a high-level intuition for the proof. Let $P_{in}=P^{B_1}_{in}\sse B_1$, and assume for now that $P_{in}$ only contains the edges of $\hat{E}_B$ (this is not necessarily true in general). Observe that $P_{in}$ contains no edges of $B\setminus B_1$. Therefore, the sets $E(B^*),E(P_{in}),E(P_0)$ and $E(P_{out}^B)$ of edges are completely disjoint. 
Consider the drawing $\bphi$ of $\G$, and erase from it all edges and vertices, except those participating in $B^*$, $P_{in},P_0$ and $P_{out}^B$. Let $\phi'$ be the resulting drawing. For convenience, we call the edges of $\hat{E}_B$ \emph{blue edges}, and the remaining edges \emph{red edges}. By our assumption, the blue edges do not participate in any crossings. Since we have assumed that $P_{in}$ only consists of blue edges, all crossings in $\phi'$ are between the edges of $P_0$, $P_{1,out}^B$ and $P_{2,out}^B$. All these three paths share a common endpoint, $y$, and they are completely disjoint otherwise. Therefore, we can uncross their drawings in $\phi'$, and obtain a planar drawing $\phi''$ of $B^*\cup P_{in}\cup P^B_{out}\cup P_0$. Erase the drawing of $P_0$ from $\phi''$, and replace the drawings of paths $P_{out}^B$ and $P_{in}$ by drawings of edges $e: u\connect v$, $e': u_1\connect v_1$, respectively, to obtain a planar drawing $\pi'$ of $\tilde{B}'$. 
Note that in $\pi'$, the drawings of edges $(u,v)$ and $(u_1, v_1)$ (and therefore
their endpoints) lie on the boundary of one face, since the drawing of the path $P_0$ in $\phi''$ connects {\bf internal} points of edges
$(u_1,v_1)$ and $(u,v)$ and does not cross the images of any edges.  Therefore, we have found a planar drawing of $\tilde{B}'$, in which the vertices $u,v,u',v'$ lie on the boundary of the same face, contradicting the fact that $B\in \rset_3(X)$. The only problem with this approach is that $P_{in}$ does not necessarily only consist of edges of $\hat{E}_B\setminus E(B^*)$. We overcome this by finding a new path $P'_{in}:v\connect u$ that only contains edges of $\hat{E}_B$ but no edges of $B^*$, and another path $P'_0$ connecting an inner vertex $x'$ of $P'_{in}$ to the vertex $y$. 
If we ensure that (1) $P'_{in}:v\connect u$ only contains edges of $\hat{E}_B$ but no edges of $B^*$; (2) path $P'_0:x'\connect y$ connects an inner vertex $x'$ of $P'_{in}$ to $y$ and contains no edges of $B^*$; and (3) The paths $P'_{in},P'_0$ and $P_{out}^B$ are completely disjoint, except for possibly sharing endpoints, then we can again apply the above argument, while replacing the path $P'_{in}$ with $P_{in}$, and path $P_0$ with $P'_0$. We now provide the formal proof. 
We first note that at least one of the four blocks $B_1,B_2,B_3,B_4$ is complex. 
Indeed, by Claim~\ref{cor:regular-block-in-a-chain} it suffices to show that $V(B_1)\setminus V(B_5)$ does not contain a vertex $w$ whose degree is $2$ in $X$. Note that if $w\in V(B_1)\setminus V(B_5)$ and the degree of $w$ in $X$ is $2$, then $w\in\sset_X$. This is since $\G$ is $3$-connected, and so all degree-$2$ vertices in $X$ must either be incident on an edge of $E^*$, or belong to $S_1$. Therefore, one of the blocks $B_1,\ldots,B_4$ must have been added to $\rset_1(X)$, together with its five immediate ancestors. 

We finally show that since one of the blocks $B_i$, for $1\leq i\leq 4$, is complex, we can find the planar drawing of $\tilde{B}'$ in which $u,v,u_1,v_1$ lie on the same face, thus leading to contradiction.

\begin{claim}
If at least one of the blocks $B_i$, for $1\leq i\leq 4$ is complex, then there is a planar drawing of $\tilde{B}'$, in which $u,v,u_1,v_1$ all lie on the boundary of the same face.
\end{claim}

\begin{proof} 
Let $B_i$ be the first complex block among $B_1$, $B_2$, $B_3$ and $B_4$. Notice that since $B_i$ has only one child in $\tset(X)$, it must contain at least one inner vertex.
Choose an arbitrary inner vertex $x'$ of $\tilde B_i$. Since $B_i$ is complex, 
there are three paths $Q_1:x'\connect u_i$, $Q_2:x'\connect v_i$,
and $Q_3:x'\connect w$, as in Claim~\ref{claim: irregular}. We assume w.l.o.g, that $w = u_{i+1}$.
We extend paths $Q_1$ and $Q_2$ to paths $Q_1'$ and $Q_2'$, connecting $x'$ to vertices $u_1$ and $v_1$, as follows.
Since $X$ is 2-connected, there are two vertex disjoint paths connecting $\set{u_i,v_i}$ to $\set{u_1,v_1}$ in $B_1$. We assume w.l.o.g. that these paths are $\Delta_1: u_i\connect u_1$ and $\Delta_2: v_i\connect v_1$.
We append these paths to $Q_1$ and $Q_2$, obtaining the desired paths $Q_1': x'\connect u_1$ and $Q_2': x'\connect v_1$. 
Finally, we define paths $P'_{in}$ and $P'_0$, as follows. 
Let $P'_{in}:u_1\connect v_1$ be the union of paths $Q_1':x'\connect u_1$ and $Q_2':x'\connect v_1$.
Let $P'_0:x'\connect y_B$ be the union of paths $Q_3:x'\to u_{i+1}$,
$P^{B_i}_{1, in}:u_{i+1} \connect x$ and $P_0^B:x \connect y$ (see Figure~\ref{fig:b3-paths}).
Observe that $x'$ is indeed an inner vertex of $P_{in}'$, so $P_0'$ connects an inner vertex of $P_{in}'$ to an inner vertex of $P_{out}^B$, as required.

\begin{figure}
\begin{center}
\scalebox{0.7}{\includegraphics{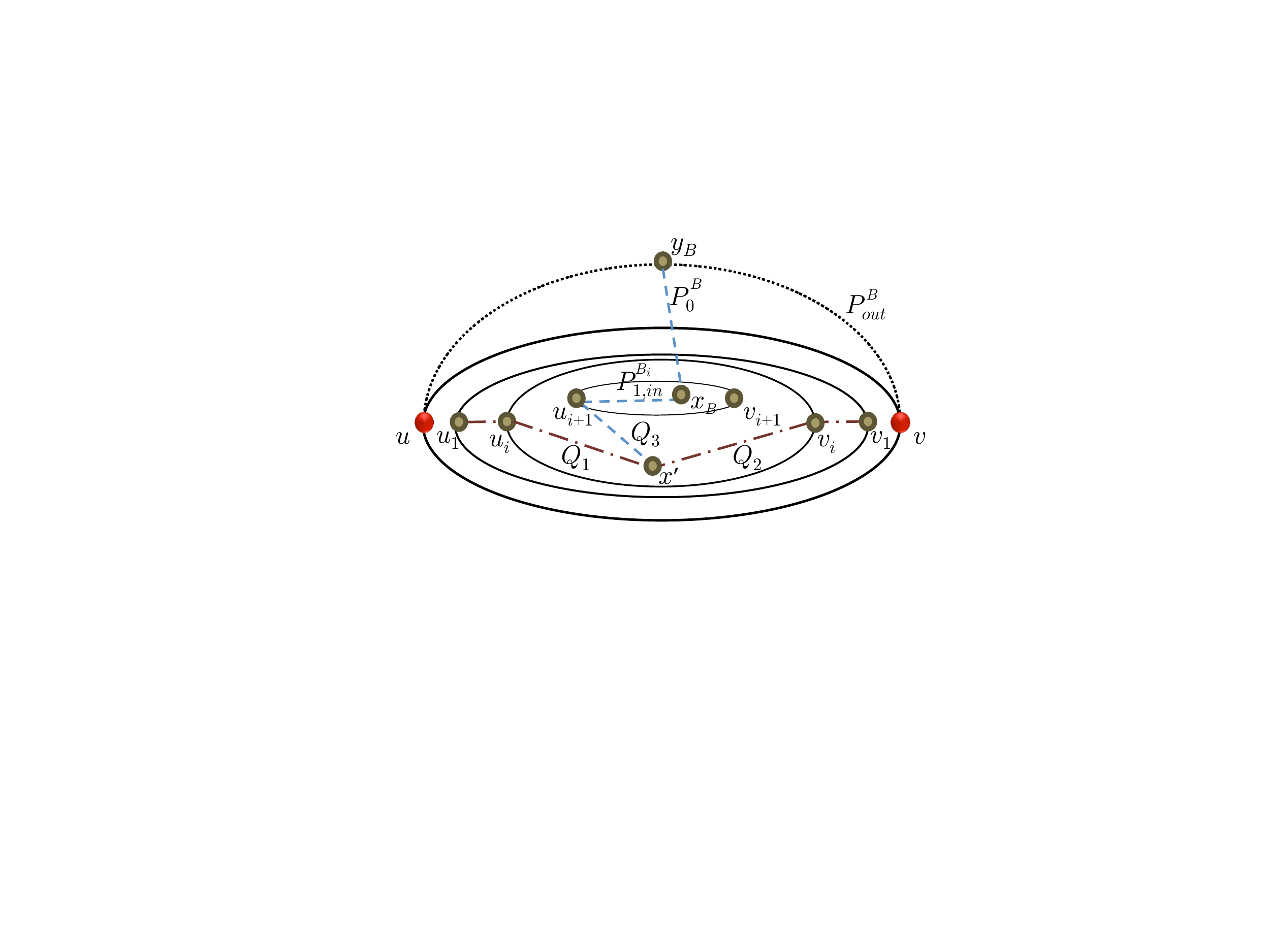}}
\caption{Paths $Q_1$, $Q_2$ (and their extensions $Q_1'$ and $Q_2'$), and $Q_3$. Recall that path $P_{in}'=(Q_1',Q_2')$, and path $P'_0=(Q_3,P_{1,in}^{B_i},P_0)$.
\label{fig:b3-paths}}
\end{center}
\end{figure}

We now verify that paths $P'_{in}$ and $P'_0$ satisfy other required conditions.
First, $P'_{in}$ only contains edges of $\hat{E}_B$ but no edges of $B^*$, since all paths $Q_1$, $Q_1'$, $Q_2$, $Q_2'$
lie in $B_1$ but do not contain edges of $B_{i+1} \supseteq B_5$. 
Next, path  $P_0':x'\connect y_B$ does not contain edges of $B^*$, since it
is the concatenation of the path $Q_3\sse B_i \subseteq B_1$,
the path $P^{B_{i}}_{1,in}\sse B_i \subseteq B_1$ and the path $P_0$, that does not contain edges of $B$. 
It is straightforward to verify that paths $P'_{in}$, $P_0'$, and $P^B_{out}$ 
share no vertices except for $y$ and $x'$. Therefore, the sets $E(B^*),E(P_0'),E(P_{in}')$ and $E(P_{out}^B)$ of edges are completely disjoint, as required. 

We now consider the drawing $\phi'$ obtained from $\bphi$, after we remove all edges and vertices, except those participating in $B^*$, $P_{out}^B,P'_{in}$ and $P'_0$. We call the edges of $\hat{E}_B$ blue, and the remaining edges red. Then $P'_{in}$ only consists of blue edges, but it does not contain edges of $B^*$. Since in the resulting drawing, $\phi'$, no blue edges participate in crossings, the only crossings involve paths $P_{1,out}^B,P_{2,out}^B$ and $P'_0$. As before, we can uncross them and obtain a planar drawing $\phi''$, which gives a planar drawing $\pi'$ of $\tilde{B}'$, in which the vertices $u,v,u_1,v_1$ all lie on the same face.
\end{proof}

\end{proof}

Let $\rset(X)=\rset_1(X)\cup \rset_2(X)\cup \rset_2(X)$, and let
$\rset'(X)$ be the set of all blocks $B\in \fset(X)$, whose parent belongs to $\rset(X)$.
Since all leaves of tree $\tset(X)$ belong to $\rset_1(X)$, it is easy to see that $|\rset'(X)|\leq |\rset_1(X)|$.
Therefore, we get the following corollary:

\begin{corollary}\label{corollary: sizes of R's}
$$\sum_{X\in\cset}( |\rset(X)|+|\rset'(X)|)\leq O(\cro_{\bphi}(\G) + |E^*|).$$
\end{corollary}

By Theorem~\ref{thm: block decomposition}, every vertex in $S_2(X)$ is an endpoint of 
a block in $\fset(X)$, or it has degree 2 in $X$. 
Let $\tilde{S}_2(X)\sse S_2(X)$ denote the set of vertices of $S_2(X)$ that either have degree $2$ in $X$, or serve as endpoints of blocks in $\rset(X)\cup \rset'(X)$, and let $S'_2(X)=S_2(X)\setminus \tilde{S}_2(X)$. Additionally, let $\tilde{S}_2=\bigcup_{X\in\cset}\tilde{S}_2(X)$, and $S'_2=S_2\setminus\tilde{S}_2$. Since, as we already observed, vertices that have degree $2$ in $X$ belong to $\sset_X$, we have that: 
\iffull
$$
|\tilde{S}_2| \leq \sum_{X\in \cset}(2|\rset(X)|+2|\rset'(X)|+|\sset_X|) \leq O(\cro_{\bphi}(\G) + |E^*|).
$$
\fi%
We let $\tilde{E}_2(X)\sse E_2(X)$ denote the edges of $E_2(X)$ that have at least one endpoint in $\tilde{S}_2(X)$, and $E'_2(X)=E_2(X)\setminus \tilde{E}_2(X)$. Additionally, let $\tilde{E}_2=\bigcup_{X\in\cset}\tilde{E}_2(X)$, and $E'_2=E_2\setminus \tilde{E}'_2$. Clearly,
\[|\tilde{E}_2|\leq \dmax|\tilde{S}_2|\leq O(\dmax)(\cro_{\bphi}(\G) + |E^*|).\]

It now only remains to bound the number of irregular vertices in set $S'_2$, and the number of irregular edges in set $E'_2$. From our definitions,  for each $X\in \cset$, for each $v\in S'_2(X)$, there is a block $B\in \fset(X)\setminus (\rset(X)\cup \rset'(X))$, such that $v$ is an endpoint of $B$. Moreover, for each $e\in E'_2(X)$, both endpoints of $e$ belong to $S_2'(X)$.

\subsubsection*{Step 3: Taking care of tunnels} 
We now consider blocks of $\fset(X)\setminus\rset(X)$.
The degree of each such block in $\tset(X)$ is $2$.
A \emph{tunnel} $Z$ is a maximal path in $\tset(X)$ containing blocks in 
$\fset(X)\setminus \rset(X)$.
Let $\zset(X)$ denote the set of all such tunnels in $\tset(X)$, and let $\zset=\bigcup_{X\in\cset}\zset(X)$. Notice that each pair of tunnels is completely 
disjoint in the tree $\tset(X)$ (but their blocks may share vertices: if the first block of one of the tunnels is a descendant of the last block of another in $\tset(X)$, then the blocks are nested; also, the first blocks of two tunnels can share endpoints).

The parent of the first block (closest to the root of ${\cal T}(X)$) in a tunnel 
belongs to $\rset(X)$. 
Therefore, by Corollary~\ref{corollary: sizes of R's}, the total number of tunnels is at most
\begin{equation}
|\zset|\leq\sum_{X\in \cset} |\rset'(X)| = O(\cro_{\bphi}(\G) + |E^*|). \label{eq:num of tunnels}
\end{equation}

Consider some tunnel $Z=B_1\supset\cdots\supset B_{\kappa}$. Denote the endpoints of the block $B_i$ by $(u_i,v_i)$, for $1\leq i\leq 
\kappa$.
Let $B'\sse B_{\kappa}$ be the unique child of block $B_{\kappa}$ in $\tset(X)$, and denote its endpoints by $(u',v')$. Since a tunnel consists of consecutive blocks in ${\cal T}(X)$, none of which
are in $\rset_2(X)$, all blocks in the tunnel have the same connector vertex.
Denote $x = x_{B_1}$, $y=y_{B_1}$, $P_0 = P_0^{B_1}$, and recall that for all $1\leq i\leq \kappa$, $x_{B_i}=x$, $y_{B_i}=y$, and $P_0^{B_i}=P_0$. Let
$P_{in}=P_{in}^{B'}$ and $P_{out} = P_{out}^{B_1}$.
Note that $x$ is an inner vertex of $P_{in}$, and $y$ is an inner vertex of $P_{out}$.
All three paths $P_0:x\connect y$, $P_{in}:u'\connect v'$ and $P_{out}:u\connect v$ share no vertices except for $x$ and $y$.

We define two auxiliary graphs corresponding to the tunnel $Z$. First, we remove all inner vertices of $B'$ from $B_1$, to obtain the graph $\H_Z$. 
We then add paths $P_0$, $P_{out}$, $P_{in}$ to $\H_Z$, contracting all degree-$2$ vertices
in the subgraph $P_0 \cup P_{out} \cup P_{in}$, to obtain the graph $J_Z$.
Therefore, the paths $P_0$, $P_{out}$ and $P_{in}$ are represented by $5$ edges in $J_Z$ (see Figure \ref{fig:tunnel}).
We call these edges \textit{artificial edges}. 


Observe that $\psi_{init}$ induces a planar drawing $\psi_Z$ of the graph $\H_Z\cup P_{in}\cup P_{out}$. However, in this drawing, we are not guaranteed that the vertices $(v_1,v_2,\ldots,v_{\kappa},v',u',u_{\kappa},\ldots,u_1)$ all lie on the boundary of the same face. Our next goal is to change the drawing $\psi_Z$ to ensure that all these vertices lie on the boundary of the same face. We can then extend this drawing to obtain a planar drawing of $J_Z$. Combining the final drawings $\bpsi_Z$ for all tunnels $Z$ will give the final drawing $\bpsi$ of the whole graph.

We start with the drawing $\psi_Z$ of $\H_Z\cup P_{in}\cup P_{out}$, induced by $\psi_{init}$. We then perform $\kappa$ iterations. In iteration $i: 1\leq i\leq \kappa$, we ensure that  all vertices in $(v_1,v_2,\ldots,v_{i+1},u_{i+1},\dots,u_1)$ lie on the boundary of the same face. We refer to this face as the outer face. For convenience, we denote $v'$ and $u'$ by $v_{\kappa+1}$ and $u_{\kappa+1}$, respectively.

Consider some iteration $i: 1\leq i\leq\kappa$, and assume that we are given a current drawing $\psi_Z$ of $H_Z\cup P_{in}\cup P_{out}$, in which the vertices in 
$(v_1,v_2,\ldots,v_{i},u_{i},\dots,u_1)$ lie on the boundary $\gamma$ of the outer face $F_{out}$ of the drawing. Let $\psi_i$ be the drawing, induced by $\psi_Z$, of the graph $B_i\cup \gamma$. Let $\psi_{i}'$ be the drawing obtained from $\psi_i$ after we replace $B_{i+1}$ with a single edge. Notice that $(u_i,v_i)$ both lie on $\gamma$, so we can view $\gamma$ as the drawing of the path $P_{out}^{B_i}$. Recall that in the unique planar drawing $\pi_{\tilde{B}_i'}$ of $\tilde{B}_i'$, the four vertices $u_i,v_i,u_{i+1},v_{i+1}$ all lie on the boundary of the same face. In particular, there is a cycle $C_i\sse B_i$, such that $u_i,v_i,u_{i+1},v_{i+1}\in C_i$, and if $\gamma_i$ denotes the drawing of $C_i$ given by $\pi_{\tilde{B}_i'}$, then all edges and vertices of $B_i\setminus C_i$ are drawn inside $\gamma_i$. Let $C_i'$, $C_i''$ be the two segments connecting $u_i$ to $v_i$ in $C_i$. Notice that both $u_{i+1}$ and $v_{i+1}$ must belong to the same segment, since otherwise, the ordering of the four vertices along $C_i$ is  either $(v_i,v_{i+1},u_{i},u_{i+1})$, or $(u_i,v_{i+1},v_{i},u_{i+1})$, and the images of the artificial edges $(u_i,v_i)$ and $(u_{i+1},v_{i+1})$ would cross in $\pi_{\tilde B'_i}$. Assume w.l.o.g. that $u_{i+1},v_{i+1}\in C_i'$
 We have three possibilities. The first possibility is that the vertices $u_{i+1},v_{i+1}$ belong to $\gamma$ -- in this case we do nothing. The second possibility is that the segment $C_i''\sse \gamma$. In this case we can ``flip'' the drawing of $B_i$, so that now $C_i'$ lies on the boundary of the outer face of the drawing of $\H_Z$, thus ensuring that all vertices  $(v_1,v_2,\ldots,v_{i+1},u_{i+1},\dots,u_1)$ lie on the boundary of the outer face.
The third possibility is that there is an edge $e=(u_i,v_i)$ that belongs to $\gamma$. In this case, we ``flip'' the image of the edge $e$ (possibly together with the image of $B_i$), so that $C'_i$ becomes the part of the boundary of the outer face (see Figure~\ref{fig: flipping one round}).

\begin{figure}
\begin{center}
\iffull
\scalebox{0.5}{\includegraphics{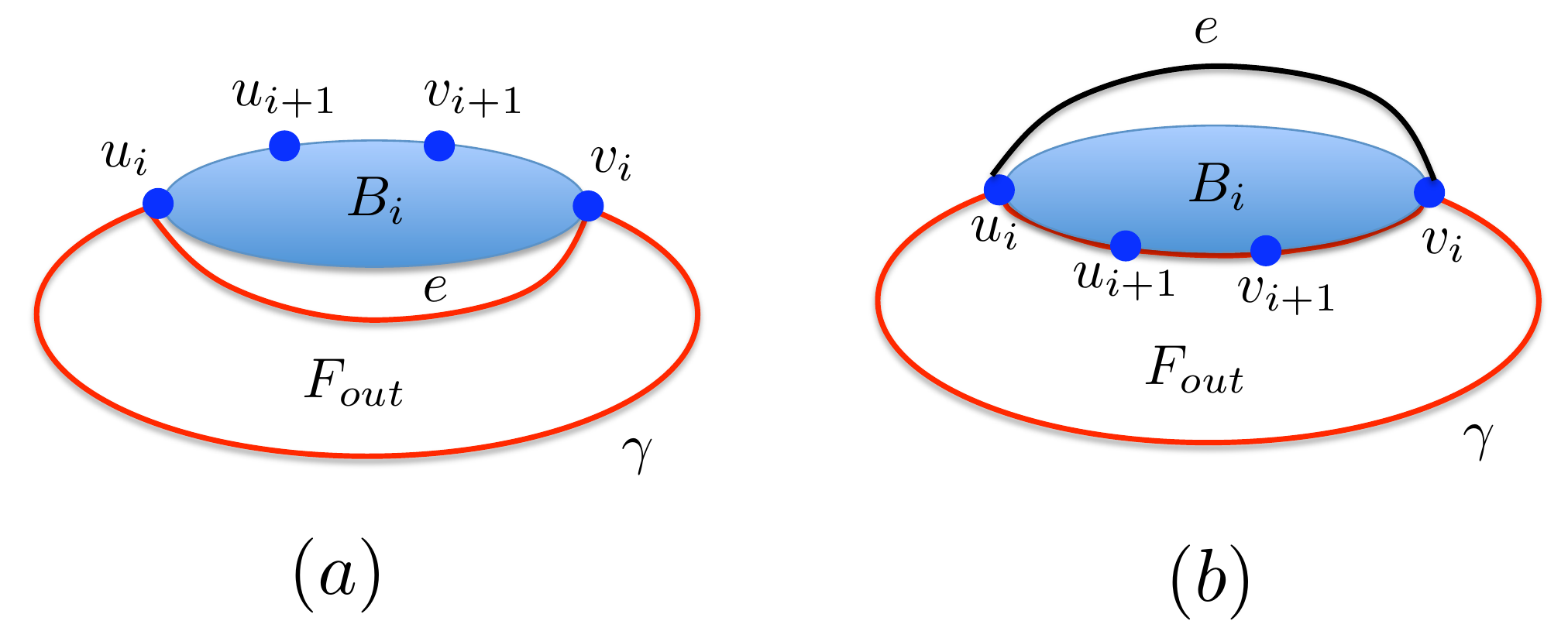}}
\fi
\caption{Iteration $i$.\label{fig: flipping one round}}
\end{center}
\end{figure}

Let $\bpsi_Z$ be this new embedding of the graph $\H_Z$. 
Since different tunnels are completely disjoint (except that it is possible that the last block of one tunnel contains the first block of another), we can perform this operation independently for each tunnel $Z\in \zset(X)$, for all $X\in \cset$ and the resulting planar 
embedding $\bpsi$ is our final planar embedding of $\H$. 
Notice that for every tunnel $Z$, we can naturally extend $\bpsi_Z$ to a planar embedding $\bpsi(J_Z)$ of $J_Z$, by adding a planar drawing of the $5$ artificial edges of $J_Z$ inside the face on whose boundary the vertices $u_1,u_2,\ldots,u_{\kappa},v_{\kappa},\ldots,v_1$ lie.

\begin{figure}
\begin{center}
\iffull
\scalebox{0.6}{\includegraphics{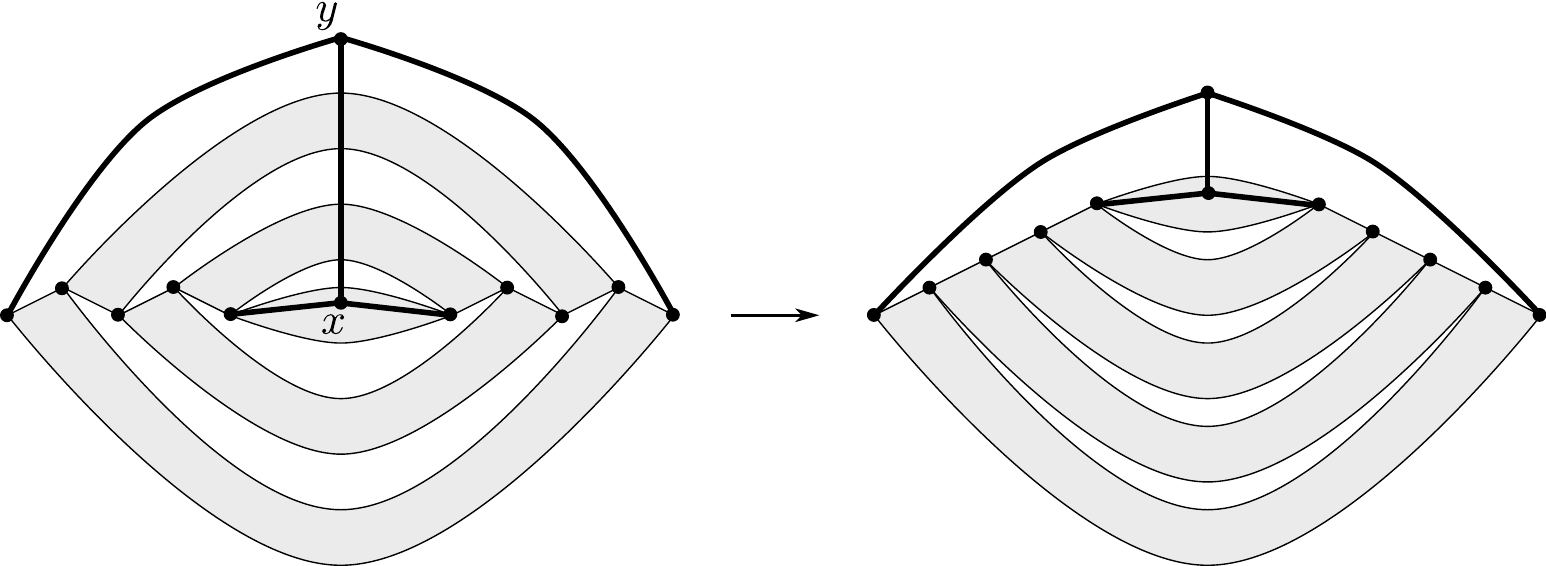}}
\fi
\caption{Graph $J_Z$. Bold lines are the artificial edges, representing the paths $P_0,P_{in}$ and $P_{out}$. The second figure shows the outcome of the flipping procedure, where all vertices $u_1,u_2,\ldots,u_{\kappa},v_{\kappa},\ldots,v_1$ lie on the boundary of one face.\label{fig:tunnel}}
\end{center}
\end{figure}

It now only remains to bound the number of irregular vertices in $\irreg_V(\bphi,\bpsi)\cap S'_2$, and the number of irregular edges in $\irreg_E(\bphi,\bpsi)\cap E'_2$.

For every tunnel $Z\in\zset$, let $\hat{S}_2(Z)=\{u\in V(J_Z):\exists v\in V(J_Z)\text{ s.t. } (u,v) \text{ is a 2-separator for } J_Z\}$.
We need the following lemma, whose proof \iffull appears in the end of this section\fi .

\begin{lemma} \label{lem:JZ-connected}
For every tunnel $Z\in \zset$, $|\hat{S}_2(Z)|\leq 8$.
\end{lemma}

We now show how to complete the proof of
Lemma~\ref{lem:irregular2}, using Lemma~\ref{lem:JZ-connected}.

Recall that $\bphi$ is the optimal embedding of $\G$. For each tunnel $Z\in \zset$, we define the following drawing $\bphi(J_Z)$: first, erase from $\bphi$ all edges and vertices, except those participating in $Z$, $P_0$, $P_{in}$ and $P_{out}$ (that have been defined for $Z$). Next, route the five artificial edges of $J_Z$ along the images of the paths $P_0$, $P_{in}$ and $P_{out}$. Finally, if any pair of artificial edges crosses more than once in the resulting embedding, perform uncrossing, that eliminates such multiple crossings, without increasing the number of other crossings in the drawing. Let $\cro_{\bphi(J_Z)}$ denote the number of crossings in the resulting drawing. Since the five artificial edges may have at most $25$ crossings with each other, we have that:
\[\cro_{\bphi(J_Z)}\leq \cro_{\bphi}(\H_Z,\G)+25\]
and
\[\sum_{Z\in\zset}\cro_{\bphi(J_Z)}\leq O(\cro_{\bphi}(\G))+O(|\zset|)\leq O(\optcro{\G}+|E^*|).\]

Fix some tunnel $Z\in \zset$. Since the drawing $\bpsi(J_Z)$ is planar, we can apply
Lemma~\ref{lem:irregular3} to the drawings $\bpsi(J_Z),\bphi(J_Z)$ of $J_Z$, and get that:
 \[|\irreg_V(\bpsi(J_Z),\bphi(J_Z))|\leq O(\cro_{\bphi(J_Z)}(J_Z)+|\hat{S}_2(Z)|)\]
 and
 \[|\irreg_E(\bpsi(J_Z),\bphi(J_Z))|\leq O(\dmax)(\cro_{\bphi(J_Z)}(J_Z)+|\hat{S}_2(Z)|).\]
Summing up over all tunnels $Z\in\zset$, we get that:
\iffull
$$
\sum_{Z\in\zset}|\irreg_V(\bpsi(J_Z),\bphi(J_Z))| \leq O(\optcro{G}+|E^*|)+O(\zset) = O(\optcro{G}+|E^*|)
$$
\fi
and
\[\sum_{Z\in\zset}|\irreg_E(\bpsi(J_Z),\bphi(J_Z))|\leq O(\dmax)(\optcro{G}+|E^*|)\]

Finally, we observe that since the tunnels are disjoint, if $v\in S_2'$, $v\in V(Z)$, and $v\in \irreg_V(\bphi,\bpsi)$, then either $v\in \irreg_V(\bpsi(J_Z),\bphi(J_Z))$, or $v$ is an endpoint of the first block of the tunnel $Z$. Therefore,
\iffull
$$
|\irreg_V(\bphi,\bpsi)\cap S'_2| \leq \sum_{Z\in\zset}(|\irreg_V(\bpsi(J_Z),\bphi(J_Z))|+2)
\leq O(\cro_{\bphi(J_Z)}(J_Z)+|E^*|).
$$
\fi
Each edge in $E_2'$ has both endpoints in $S'_2$, and therefore must be either completely contained in some tunnel, or be adjacent to an endpoint of the first block of a tunnel. So if $e\in E_2'$, and $e\in \irreg_E(\bphi,\bpsi)$, then either $e\in \irreg_E(\bpsi(J_Z),\bphi(J_Z))$ for some tunnel $Z$, or it is adjacent to an endpoint of the first block of some tunnel $Z$. Therefore,
\iffull
$$
|\irreg_E (\bphi,\bpsi) \cap E'_2| \leq \sum_{Z\in\zset}(|\irreg_V(\bpsi(J_Z),\bphi(J_Z))|+2\dmax) \leq O(\dmax)(\cro_{\bphi(J_Z)}(J_Z)+|E^*|).
$$
\fi

\iffull
It now only remains to prove Lemma~\ref{lem:JZ-connected}.
\fi

\iffull
\subsubsection*{Proof of Lemma~\ref{lem:JZ-connected}}

Consider a tunnel $Z=(B_1,\ldots,B_{\kappa})$, $Z\sse X$ for some $X\in\cset$.
We will show that if we remove any pair of vertices, except for, possibly, pairs in the set
$\pset=\set{(x,u_{\kappa}),(x,v_{\kappa}),(y,u_2),(y,v_2)}$, the graph $J_Z$ remains connected.
For convenience, we denote $u_{\kappa+1} = u'$ and $v_{\kappa+1} =v'$, the endpoints of the unique child $B'$ of $B_{\kappa}$.

Observe that every vertex $a \notin \{u_1, v_1, u_{\kappa+1}, v_{\kappa+1}\}$ of $J_Z$ has degree
at least $3$ in $J_Z$: otherwise we would have a vertex $a\in V(Z)\setminus\set{u_1,v_1,u_{\kappa+1},v_{\kappa+1}}$ with $\deg_{X} a \leq 2$. Then either $a$ is incident on an edge in $E^*(X)$, or it belongs to $S_1$, and therefore
and $a\in \sset_X$. Then the smallest $B_i$ of $Z$ that contains $a$
would belong to $\rset_1(X)$. 

Consider now a vertex $w\notin \{u_i, v_i, x, y: 1\leq i\leq \kappa+1\}$.  
Let $B_j$ be the smallest block that contains $w$. Since $V(B_j) \setminus V(B_{j+1})$ contains an inner vertex $w$, $B_j$ must be a complex block. Therefore,  from Claim~\ref{claim: irregular}, $w$ is connected to $u_j$, $v_j$ and either $u_{j+1}$
or $v_{j+1}$ by three vertex disjoint paths. It is obvious that each of the 
vertices $x$ and $y$ is connected to $u_1$, $v_1$, and
either $u_{\kappa+1}$ and $v_{\kappa+1}$ by three vertex disjoint paths. 

Let $L=\set{u_i,v_i: 1\leq i\leq \kappa+1}$.
Observe that it is enough to show that for any pair $(p,q)\not\in \pset$ of vertices, all vertices in $L\setminus\set{p,q}$ remain connected in the graph $J_Z\setminus\set{p,q}$. Indeed, assume that all vertices in set $L\setminus\set{p,q}$ remain connected in graph $J_Z\setminus\set{p,q}$. Let $w,w'$ be any pair of vertices of $J_Z\setminus\set{p,q}$. We show that $w,w'$ remain connected as well in the resulting graph. Indeed, if both $w,w'\not\in L$, each one of these vertices has three paths disjoint paths connecting them to vertices of $L$, and at least one of the three paths must survive even after the removal of $p,q$ from $J_Z$. Similarly, if one of vertices $w,w'$ belongs to $L$, they remain connected as long as all vertices of $L$ remain connected.

We now show that for any pair $(p,q)$ of vertices, $(p,q)\not\in \pset$, the vertices in $L\setminus\set{p,q}$ all remain connected in graph $J_Z\setminus\set{p,q}$.

Consider the unique planar embedding of $J_Z$. If we remove the edge $(x,y)$ from this embedding, we obtain a cycle $C$ containing the vertices of $L$: this cycle is simply the boundary of the face that contained the edge $(x,y)$.
Denote the ordering of vertices on the cycle by $(v_1,v_2,\ldots,v_{\kappa+1},x,u_{\kappa+1},\ldots,u_1,y)$, where $(u_i,v_i)$ are endpoints of block $B_i$ (but observe that some consecutive vertices in this ordering may coincide, e.g. it is possible that $u_i=u_{i+1}$).
Let $L_u=\set{u_i: 1\leq i\leq \kappa+1}$, and $L_v=\set{v_i: 1\leq i\leq \kappa+1}$.

Consider a vertex $u_i$ with $\deg_X u_i \geq 3$. Denote its neighbors on the cycle
by $u_{i_1}$ and $u_{i_2}$ (since some vertices $u_j$ might coincide, the vertex
$u_{i_1}$ is not necessarily equal to $u_{i-1}$).
Let $w$ be a neighbor of $u_i$ other than $u_{i_1}$ and $u_{i_2}$. Clearly,
$w\notin L_u$, because of definition of block. 

Assume first that $w\notin L_v$ either.
Consider the smallest block $B_h$ that contains $w$.
Again, since $w\not\in L$, block $B_h$ has to be complex. Therefore, there is a path from $w$ to $v_h$ that 
does not visit any vertex in $L_u$. By concatenating this path with the edge $(u_i,w)$, 
we get a path $P(u_i)$ from $u_i$ to $v_h$  that does not visit any other vertex in $L_u$. If $w\in L_v$, then the path $P(u_i)$ is simply the edge $(u_i,w)$.
Similarly, there is a path $P(v_i)$ from every vertex $v_i$ of degree at least $3$ 
to some $u_h\in L_u$, that does not visit any other vertex in $L_v$. 

Assume for contradiction that there is some pair $(p,q)\not\in \pset$, such that in the graph $J_Z\setminus\set{p,q}$, the set $L\setminus\set{p,q}$ is not connected. Since all vertices of $L$ lie on the cycle $C$, it is clear that $p$ and $q$ must belong to $C$.
Moreover, both of them must lie on the same of the two arcs connecting $x$ and $y$ in $C$.
Let us say that $(p,q)$ belong to the arc on which the vertices $L_u$ lie.
Then the other arc of $C$ remains connected. But then for each vertex $u_j\not\in\set{p,q}$ of degree at least $3$, the path $P(u_j)$ will connect it to the vertices in $L_v$. Similarly, if $(p,q)$ belong to the arc containing vertices of $L_v$, each vertex $v_j\not\in\set{p,q}$ of degree at least $3$, remains connected to the vertices of $L_u$ via the path $P(v_j)$.
 
 Now if $u_i$ has degree $2$, then it must be either $u_1$ or $u_{\kappa+1}$.
It is clear that we can disconnect $u_1$ from $u_2$ and $y$
only by removing both $u_2$ and $y$ (since $B_i$ is $2$-vertex connected). Similarly, we can disconnect
$u_{\kappa+1}$ from $u_{\kappa}$ and $x$ only by removing both $u_{\kappa}$ and $x$. If vertex $v_i$ has degree $2$, then it must be either $v_1$, or $v_{\kappa+1}$, and so the only pairs of vertices that can disconnect it are $(x,v_{\kappa})$ and $(y,v_2)$.
\fi

\iffull
\section{Handling Non 3-Connected Graphs} \label{sec:reduce-three-connected}

In this section we prove Theorem~\ref{thm:main}, by describing a reduction from the general case --- when
the graph $G$ is not necessarily 3-connected --- to the 3-connected case. 
Our algorithm consists of two parts. In the first part, we decompose the original graph $\G$ into a number of sub-graphs, and find a drawing for each one of the sub-graphs separately. In the second part, we combine these drawings together to obtain the final drawing.

\subsection{Part 1: Decomposition}
We first note that we can assume w.l.o.g. that the input graph $\G$ is $2$-connected:
Otherwise, we can separately embed the 2-connected components of $\G$ and then combine their embeddings.
We also assume that the graph $\H=\G\setminus E^*$ is connected, since otherwise we can start removing edges from $E^*$ and adding them to $\H$, until it becomes connected, as in Section~\ref{sec:alg}. Finally, we can assume w.l.o.g. that the input graph contains no parallel edges: otherwise, if there is a collection $(e_1,\ldots,e_{\kappa})$ of parallel edges, we can subdivide each edge $e_i$ by adding a vertex $v_i$ to it, and add edges connecting every consecutive pair $(v_i,v_{i+1})$ of vertices, for $1\leq i\leq \kappa$ (we identify $v_{\kappa+1}$ and $v_1$). It is easy to verify that this transformation does not increase the maximum vertex degree in $\G$, and does not increase the cost of the optimal solution.
We will use the following theorem of Hlin\v{e}n\'{y} and Salazar~\cite{HlinenyS06}.

\begin{theorem}[\cite{HlinenyS06}]\label{thm planar and edge}
Let $G$ be any graph of maximum degree $\dmax$, $e\in E(G)$, such that $G\setminus\set{e}$ is planar. Then we can efficiently find a drawing $\psi$ of $G$ with at most $O(\dmax\cdot \optcro{G})$ crossings.
\end{theorem}

Our high-level idea is to decompose the graph $\G$ into blocks (see Section~\ref{sec: blocks} for the definition). For each such block $B$, we will add an artificial edge connecting its endpoints, that will ``simulate'' the rest of the graph, $\G\setminus B$. Similarly, we will add an artificial edge connecting the endpoints of $B$ to the remaining graph, $\G\setminus B$, that will simulate $B$. In the course of such recursive decomposition, a block may end up containing a number of such artificial edges. We will then try to find drawings of each such augmented block separately. We need to argue that the total optimal solution cost in these new sub-problems does not increase by much. This is done as follows. We will have two types of blocks. The first type is blocks that have at most two artificial edges, and when one of these edges is removed from the block, we obtain a planar graph.  For such blocks, we will argue that their total solution cost is bounded by $O(\optcro{\G})$, and then use Theorem~\ref{thm planar and edge} to find their drawings. For the remaining blocks, we will show that we can augment the set $E^*$ of edges, to set $\hat{E}^*$, with $|\hat{E}^*|\leq O(|E^*|)$, such that for each such block $B$, $B\setminus \hat{E}^*$ is planar. (Observe that now $\hat{E}^*$ may have to contain artificial edges). We then use Theorem~\ref{thm:main2} to find the drawing of each such block separately. We now describe the decomposition procedure in more detail.

We say that a path $P$ in graph $\G$ is \emph{nice} if it does not contain any edge of $E^*$. We will be repeatedly using the following two easy observations.

\begin{observation}\label{observation: remains planar}
Let $G$ be any graph, $E^*\sse E(G)$ a subset of edges, such that $G\setminus E^*$ is planar. Let $B$ be a block of $G$, whose endpoints are $I(B)=(u,v)$, and assume that $B$ contains a nice path connecting $u$ to $v$. Let $G'$ be the graph obtained by removing $B$ from $G$, and adding an artificial edge $e=(u,v)$ to it. Then $G'\setminus E^*$ is also planar.
\end{observation}

\begin{observation}\label{observation: path}
Let $G$ be any $2$-connected graph, and let $B$ be any block of $G$ with endpoints $I(B)=(u,v)$. Then there is a path $P:u\connect v$ contained in $(G\setminus B)\cup\set{u,v}$.
\end{observation}

We say that a block $B$ of $G$ is nice, iff it does not contain any edge of $E^*$. We start by iteratively removing {\bf maximal} (w.r.t. inclusion) nice blocks from $G$, and adding each one of them to the set $\aset$ of nice blocks we construct. 
Consider one such block $B$ with endpoints $I(B)=(u,v)$. We remove block $B$ from $G$, and add it to the set $\aset$ of nice blocks. We then add an artificial edge $(u,v)$ both to the remaining graph $G$, and to the block $B$. Let $\aset$ be the resulting set of these augmented nice blocks, and let $G'$ be the resulting remaining graph. From Observation~\ref{observation: remains planar}, graph $G'\setminus E^*$ is planar. It is also easy to see that $G'$ does not contain any nice blocks, and it is $2$-connected.
Consider now some block $B\in \aset$. Since the algorithm was repeatedly choosing maximal nice blocks, $B$ contains a unique artificial edge, $e=(u,v)$, connecting its endpoints $I(B)=(u,v)$. Moreover, from Observation~\ref{observation: path}, the graph $(\G\setminus B)\cup\set{u,v}$ contains a path $P:u\connect v$. Therefore, $\optcro{B}\leq \cro_{\bphi}(B, \G)$, and $\sum_{B\in \aset} \optcro{B}\leq O(\optcro{\G})$. The artificial edges that have been added to graph $G'$ at this step are called type-1 artificial edges, and all artificial edges that will be added throughout the rest of the algorithm are called type-2 artificial edges.
As our next step, we use Theorem~\ref{thm: block decomposition} to find a laminar block decomposition $\fset$ for graph $G'$.
Recall that for each block $B\in \fset$, we denote by $\tilde{B}$ the graph obtained from $B$ by replacing its children with artificial edges, and $\tilde{B}'$ is obtained from $\tilde{B}$ by adding an artificial edge connecting the endpoints of $B$. For each $B\in \fset$, we are guaranteed that $\tilde{B}'$ is $3$-connected. 
Intuitively, we would now like to solve each one of the blocks $\tilde{B}'$, for $B\in\fset$, separately. However, since we have added artificial edges to such blocks, the graph $\tilde{B}'\setminus E^*$ may not be planar anymore. Of course, if we add all type-$2$ artificial edges to the set $E^*$, this problem will be resolved, but then the resulting set $E^*$ may become too large. We show below how to avoid this problem. We start by defining a set $\bset$ of blocks, whose size will be bounded by $O(|E^*|)$. We show that the type-2 artificial edges belonging to all such blocks can be added to the set $E^*$ without increasing its size by too much. We then show how to take care of remaining blocks.

\begin{itemize}
\item Let $\bset_1\sse \fset$ be the set of blocks $\tilde{B}'$, for $B\in\fset$, such that $\tilde{B}'$ contains an edge of $E^*$. Clearly, $|\bset_1|\leq O(|E^*|)$, and all the leaves of the tree $\tset$ belong to $\bset_1$ (since otherwise such a leaf would be a nice block).

\item Let $\bset_2$ be the set of blocks $\tilde{B}'$, such that $B$ has at least two children in the tree $\tset$. Since all leaves of the tree $\tset$ belong to $\bset_1$, it is easy to see that $|\bset_2|\leq |\bset_1|\leq |E^*|$.
\end{itemize}

Consider the decomposition tree $\tset$, and remove all the vertices corresponding to the blocks in $\bset_1$ and $\bset_2$ from it. The resulting sub-graph of $\tset$ is simply a collection $\pset$ of disjoint paths. Moreover, $|\pset|\leq |\bset_1|+|\bset_2|\leq 2|E^*|$. Consider some such path $P\in \pset$, and assume that $P=(B_1,B_2,\ldots,B_k)$, where $B_k\subset B_{k-1}\subset\cdots\subset B_1$. Let $i^*$ be the largest index $i: 1\leq i\leq k$, such that $B_{i}$ contains a nice path connecting its endpoints. We add 
$\tilde{B}'_{i^*},\tilde{B}'_{i^*+1},\ldots,\tilde{B}'_k$ to $\bset_3$. We show in the next claim that $i^*\geq k-5$, so $|\bset_3|\leq 6|\pset|\leq 12|E^*|$. 

\begin{claim}
$i^*\geq k-5$.
\end{claim}
\begin{proof}
Assume otherwise. Then the laminar family $\fset$ contains four blocks $B^1,B^2,B^3,B^4$, such that for $1\leq j<4$, $B^j$ is the father of $B^{j+1}$ in $\tset$, and  $B^j\setminus B^{j+1}$ does not contain edges of $E^*$. Moreover, each one of the four blocks has exactly one child, and none of these blocks contains a nice path connecting its endpoints.

Denote $I(B^1)=(u,v)$ and $I(B^4)=(u',v')$. Since $B^1$ does not contain a nice path connecting its endpoints, and $B^1\setminus B^2,B^2\setminus B^3, B^3\setminus B^4$ do not contain edges of $E^*$, all paths $P':u\connect v$ in $B^1$ must contain $u'$ and $v'$. Therefore, if we remove the vertices of $B^4\setminus\set{u',v'}$ from $B^1$, we will obtain two connected components, $R$ and $R'$, where $u\in R$ and $v\in R'$. Each one of the two components contains exactly one vertex from $\set{u',v'}$, and we assume w.l.o.g. that $u'\in R$, $v'\in R'$.

We now claim that $R$ does not contain any vertices outside of $u$ and $u'$, that is, $R$ is just a collection of parallel edges (or just a single edge): otherwise, $R$ is a nice block with end-points at $u$ and $u'$, and we have assumed that $G'$ does not contain any nice blocks. Similarly, $R'$ does not contain any vertices outside of $v$ and $v'$. Therefore, the set of vertices of block $B^1$ is $V(B^4)\cup \set{v,u}$. But then it is impossible that there are two additional blocks, $B^2,B^3$, such that $B^4\subset B^3\subset B^2\subset B^1$, as every pair of distinct blocks must differ in their vertices. 
\end{proof}

Finally, we let $C_P=(B_1\setminus B_{i^*})\cup I(B_{i^*})$. We add to $C_P$ two artificial edges: edge $e$ connecting the endpoints of $B_1$, and edge $e'$ connecting the endpoints of $B_{i^*}$. For simplicity, we denote the new graph by $C'$, and the old graph by $C=C_P$. We then add $C'$ to a new set $\cset$ of sub-graphs of $G'$. We need the following claim:

\begin{claim}
Graph $C'\setminus \set{e}$ is planar. Moreover, $\optcro{C'}\leq \cro_{\bphi}(C,\G)+1$.
\end{claim}
\begin{proof}
 Since block $B_{i^*}$ contained a nice path, denoted by $P'$, connecting $v_{i^*}$ to $u_{i^*}$, and since $B_1\setminus B_{i^*}$ did not contain edges of $E^*$, from Observation~\ref{observation: remains planar}, $C\cup\set{e'}=C'\setminus\set{e}$ is planar.

For the second part, from Observation~\ref{observation: path}, there is a path $P''$, connecting the endpoints of $B_1$ in graph $(\G\setminus B_1)\cup I(B_1)$. 
Consider now the optimal embedding $\bphi$ of $\G$, and remove from it all edges and vertices except for those in $C, P',P''$. This drawing gives a drawing $\phi$ of graph $C'$, where edge $e$ is drawn along the image of $P'$, and edge $e'$ along the image of $P''$. The number of crossings in this drawing is at most $\cro_{\bphi}(C,\G)+\cro_{\bphi}(P',P'')$. Finally, if the images of edges $e,e'$ cross multiple times in $\phi$, we can un-cross them, without increasing the number of crossings between any other pair of edges. This will result in a drawing of $C'$ with at most $ \cro_{\bphi}(C,\G)+1$ crossings.
\end{proof}

Since $|\cset|\leq |\pset|\leq 2|E^*|$, we have that $\sum_{C'\in \cset}\optcro{C'}\leq \sum_{C'\in \cset}(\cro_{\bphi}(C,\G)+1)\leq O(\optcro{\G}+|E^*|)$. 

Finally, let $\bset=\bset_1\cup\bset_2\cup \bset_3$. From the above discussion, $|\bset|\leq O(|E^*|)$. Moreover, the total number of children of blocks in $\bset$ in the tree $\tset$ is also bounded by $O(|\bset|)\leq O(|E^*|)$. 
Let $\tilde{B}'\in \bset$ be any such block, and let $B_1,\ldots,B_{\kappa}$ be its children. Recall that $\tilde{B}'$ is obtained from $B$ by replacing each child $B_i$ by an artificial edge $e_i$, connecting the endpoints $(v_i,u_i)$ of $B_i$. For each such child $B_i$, there is also a path $P_i:v_i\connect u_i$, $P_i\sse B_i$. Additionally, we have added an edge $e$ connecting the endpoints $(u,v)$ of block $B$. We associate this edge with a path $P: u\connect v$, $P\sse(\G\setminus B)\cup I(B)$, that is guaranteed by Observation~\ref{observation: path}. We now add the edges 
$e,e_1,\ldots,e_{\kappa}$ to $E^*$, and we denote by $\hat{E}^*$ be the resulting set of edges. We then have $|\hat{E}^*|\leq O(|E^*|)$.

Clearly, for each $\tilde{B}'\in \bset$, $\tilde{B}'\setminus\hat{E}^*$ is planar. This is since $B\setminus E^*$ was planar, and all the edges that have been added to $B$ in order to obtain $\tilde{B}'$, were also added to $\hat{E}^*$. Moreover, $\optcro{\tilde{B}'}\leq \optcro{\G}$, since we have a collection $P,P_1,\ldots,P_{\kappa}$ of edge disjoint paths associated with the edges $e,e_1,\ldots,e_{\kappa}$, that are contained in $\G\setminus \tilde{B}'$, connecting the endpoints of their corresponding edges.

Let $\Gamma=\aset\cup \bset\cup\cset$. Notice that since we have added artificial edges, it is possible that graphs $A\in \Gamma$ now contain parallel edges. For each such graph $A\in \Gamma$, we let $A'$ denote the corresponding graph with no parallel edges, that is, we replace every set of parallel edges with a single edge. Let $\aset',\bset',\cset'$ and $\Gamma'$ be the collections of these modified graphs, corresponding to the collections $\aset,\bset,\cset$ and $\Gamma$, respectively.

As we have already observed, each graph $A'\in \aset'\cup \cset'$ can be decomposed into a planar graph plus one additional edge, and
 $\sum_{A\in \aset\cup \cset} \optcro{A}\leq O(\optcro{\G}+|E^*|)$. Using Theorem~\ref{thm planar and edge}, we can efficiently find drawings $\psi_{A'}$ of graphs $A'\in \aset'\cup \cset'$, with at most $O(\dmax\cdot(\optcro{\G}+|E^*|))$ crossings in total.
 We can also use Theorem~\ref{thm:main2} to find a drawing $\psi_{A'}$ for each graph $A'\in \bset'$, having at most $O(\dmax\cdot |E^*|\cdot(\optcro{\G}+|E^*|))$ crossings in total. Overall, from the above discussion, for each graph $A'\in \Gamma'$, we can efficiently find a drawing $\psi_{A'}$ of $A'$, such that the total number of crossings in these drawings is bounded by $O(\dmax\cdot |E^*|\cdot(\optcro{\G}+|E^*|))$.

\subsection{Part 2: Composition of Drawings}
In this section we show how to compose the drawings of the graphs in $\Gamma'$, to obtain the final drawing of $\G$.
We build a binary decomposition tree $\cal T'$ corresponding to the
collection $\Gamma$ of sub-graphs of $\G$, as follows.
The graph at the root of the tree is $\G$. The graphs at the leaves of $\cal T'$
are the graphs in $\Gamma$. For every non-leaf node, the corresponding graph $G_0$
is the composition of its two child subgraphs $G_1$
and $G_2$ along the unique artificial edge that belongs to both $G_1$ and $G_2$.
Notice that our original decomposition tree $\tset$ can be turned into a binary tree whose leaves are graphs in $\bset\cup \cset$, and we can add graphs in $\aset$ to this tree one-by-one, as we merge them with the root of the tree, to obtain the final binary tree $\tset'$.

\begin{theorem}
Suppose that we are given the decomposition tree $\cal T'$,
and drawings $\psi_{A'}$ of graphs $A'\in \Gamma'$.
Then we can efficiently find a drawing of $\G$ with at most
$\dmax^2 \sum_{A'\in {\Gamma'}} \cro_{\psi_{A'}} (A')$ crossings.
\end{theorem}
\begin{proof}
We start by assigning weights to the edges of the graphs in the decomposition tree $\tset'$. Once the weights are assigned, for each graph $G_0$ in the tree, the weighted degree of a vertex $x\in V(G_0)$, denoted by $\deg_{G_0}^w x$, is the sum of 
the weights of the edges incident to $x$ in $G_0$. 
We assign the weights to the edges of the graphs from the top to the bottom of the tree $\tset'$. For the root graph $\G$, the weights of all its edges (which are non-artificial edges), are $1$. Let $G_0$ be the current graph, with its two children $G_1$ and $G_2$,
that share an artificial edge $e=(u,v)$.  The weights of all edges, other than the edge $e$ remain in graphs $G_1$ and $G_2$ the same as in graph $G_0$. The weight of the edge $e$ is set in both graphs $G_1$ and $G_2$ to be:
 
$$weight(e) = \min\set{\deg_{G_1\setminus e}^w(u), \deg_{G_1\setminus e}^w(v), \deg_{G_2\setminus e}^w(u), \deg_{G_2\setminus e}^w (v)}.$$

 It is easy to see that if, for all vertices $x\in V(G_0)$, $\deg_{G_0}^w(x)\leq \dmax$, then for all vertices $y\in V(G_i)$, for $i\in\set{1,2}$, $\deg_{G_i}^w(y)\leq \dmax$ as well. Therefore, the weighted degrees of all vertices in all graphs in the tree $\tset'$ are bounded by $\dmax$.
Finally, we assign weights to edges of graphs $A'\in \Gamma'$, as follows. Let $A\in \Gamma$ be the graph corresponding to $A'$. For each set $e_1,\ldots,e_{\kappa}$ of parallel edges in $A$, the weight of the corresponding edge in $A'$ is the sum of the weights of the edges $e_1,\ldots,e_{\kappa}$ in $A$. The weights of all other edges are identical in $A$ and $A'$.

We now define the weighted cost of a drawing $\psi$ of any edge-weighted graph $H$, $\cro_\psi^w(H)$, as follows.
The weighted cost of a crossing of two edges of weights $w_1$ and $w_2$ is $w_1 w_2$. 
The cost of the drawing is the sum of weighted costs of all crossings. 

Notice that for each graph $A'\in \Gamma'$, the drawing $\psi_{A'}$ of $A'$ induces a drawing $\psi_A$ of the corresponding graph $A\in \Gamma$, such that the weighted cost of $\psi_A$ is bounded by that of $\psi_{A'}$. 
Since the weighted degrees of vertices in all graphs in $\Gamma'$ are bounded by $\dmax$, we have
$$\sum_{A\in {\Gamma}} \cro_{\psi_A}^w (A) \leq\sum_{A'\in {\Gamma'}} \cro_{\psi_{A'}}^w (A') \leq \dmax^2 \sum_{A'\in {\Gamma'}} \cro_{\psi_{A'}} (A').$$
We now combine all drawings of graphs in $\Gamma$ as follows. We proceed from the bottom 
to the top of the tree $\tset'$. At each node $G_0$, we combine the
two drawings $\psi_1$ and $\psi_2$ of its two children $G_1$ and $G_2$
into a drawing $\psi_0$ of $G_0$ so that 
$$\cro^w_{\psi_0} (G_0) \leq \cro^w_{\psi_1} (G_1) + \cro^w_{\psi_2} (G_2).$$
Finally, we obtain a drawing $\psi$ of $\G$ with
$$\cro_{\psi}(\G) \leq \cro^w_{\psi} (\G) \leq \sum_{A'\in {\Gamma'}} \cro_{\psi_{A'}}^w (A')
\leq \dmax^2 \sum_{A'\in {\Gamma}} \cro_{\psi_{A'}} (A').$$

We now show how to combine the drawings $\psi_1$ and $\psi_2$ of graphs $G_1$ and $G_2$. Let $e=(u,v)$ be the unique artificial edge shared by $G_1$ and $G_2$.
Without loss of generality, we assume that $weight(e) = \deg^w_{G_2 \setminus e}(u)$.

We note that we can assume that the following properties hold (for each $i=1,2$):
\begin{itemize} 
\item the vertex $v$ lies on the external boundary of the drawing $\psi_i$;
\item there is a point $t_i$ on the drawing of the edge 
$e$ in $\psi_i$, such that the segment of the drawing of $e$ 
between $t_i$ and $\psi_i(v)$ lies on the external boundary of the drawing $\psi_i$.
\end{itemize}

If these properties do not hold, we transform each drawing $\psi_i$ as follows.
For convenience, assume that the drawing $\psi_i$ is on the 2-sphere.
We take a point $t_i$ on the curve corresponding to the edge $e$ in the drawing $\psi_i$, so that there
are no crossing points on the segment of $e$
between $t_i$ and $\psi_i(v)$.
Then we take a point $t_i'$ that lies on the same face of $\psi_i$  as $v$ and $t_i$.
Finally, we perform a stereographic projection from $t_i'$
and obtain the desired drawing $\hat\psi_i$. 
Since $v$ and $t_i$ lie on the face bounding $t_i'$ in $\psi_i$, it follows that they both lie on the outer face in $\hat\psi_i$.

We superimpose drawings $\psi_1$ and $\psi_2$ so that drawings of $G_1$ and $G_2$
do not overlap and points $\psi_1(v)$, $\psi_2(v)$, $t_1$, and $t_2$
lie on the external boundary of the drawing. We then connect points
$t_1$ and $t_2$ with a curve $\gamma_t$ and points $\psi_1(v)$ and $\psi_2(v)$
with a curve $\gamma_v$ so that curves $\gamma_t$ and $\gamma_v$
do not cross each other and do not cross the drawings of $G_1$ and $G_2$ (see Figure~\ref{fig:composition}).
Now, we erase the drawings of segments of $\psi_i(e)$
between points $t_i$ and $v$. Let $\gamma_u$ be the concatenation
of remaining pieces of $\psi_1(e)$ and $\psi_2(e)$
and $\gamma_t$. The curve $\gamma_u$ connects $\psi_1(u)$
and $\psi_2(u)$. Finally, we ``contract'' curves $\gamma_u$ and $\gamma_v$:
we move points $\psi_2(u)$ and $\psi_2(v)$ along the curves $\gamma_u$ and $\gamma_v$,
until they reach $\psi_1(u)$ and $\psi_1(v)$. We route each edge $e$ incident to $u$
(respectively $v$) in $G_2 \setminus e$: first along the curve $\gamma_u$ (respectively $\gamma_v$)
and then along the original drawing $\psi_2$ of $e$. (If edges parallel to $e$ belong to $G_2$, we re-route them in the same way: first along $\gamma_v$, then along their original drawing in $\psi_2$, and finally along $\gamma_u$).
We obtain an embedding $\psi_0$ of $G_0$ (curves $\gamma_u$, $\gamma_v$ and
the embeddings of the edge $e$ are not parts of $\psi_0$).
Figure \ref{fig:composition} depicts an example of the above composition step.

\begin{figure}
\begin{center}
\scalebox{0.80}{\includegraphics{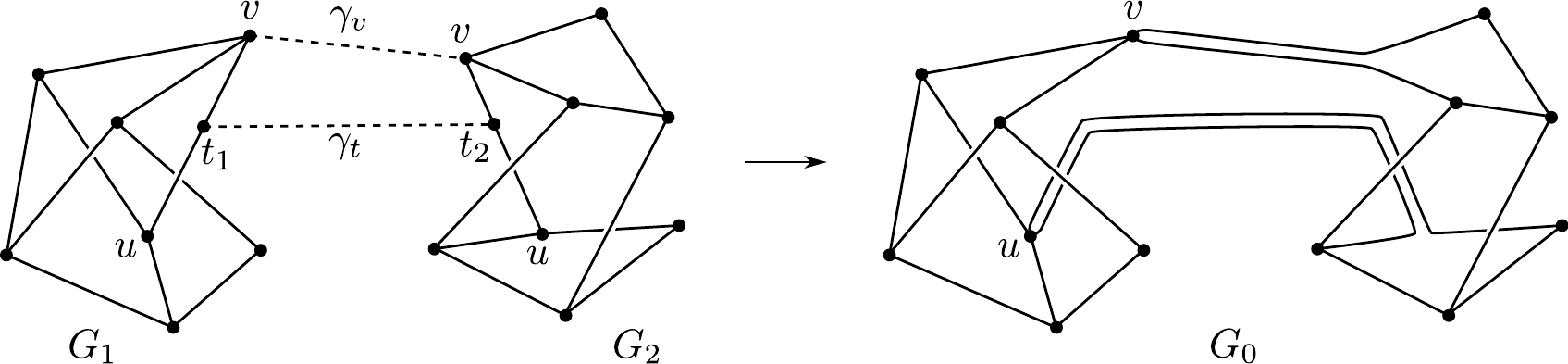}}
\caption{Obtaining a drawing for $G_0$ by composing the drawings for $G_0$ and $G_1$.\label{fig:composition}}
\end{center}
\end{figure}

Let us compute the cost of drawing $\psi_0$.
Since $\gamma_v$ does not cross the drawings $\psi_1$ and $\psi_2$,
we do not introduce any new crossings when we contract $\gamma_v$.
For every crossing of an edge $e'\in E(G_1\cup G_2)$ with $\psi_1(e)$ or $\psi_2(e)$,
we introduce crossings between all edges incident to $u$ in $G_2\setminus e$ and  $e'$.
The total weighted cost of these crossings is $\deg^w_{G_2\setminus e}(u) \cdot weight(e')$.
It is equal to $weight(e) \cdot weight(e')$,
the cost of the crossing between $e$ and $e'$. Therefore, the total weighted 
cost of the drawing does not increase, that is,
$$\cro_{\psi_0}^w(G_0) \leq \cro_{\psi_1}^w(G_1) + \cro_{\psi_2}^w(G_2).$$
\end{proof}

\fi
\iffull
\section{Algorithms for bounded-genus graphs}\label{sec:genus}
In this section we present the proof of Theorem \ref{thm: bounded genus}.

\begin{remark}
Note that for any \emph{fixed} $\gamma\geq 0$, given a graph $G$ of genus $\gamma$, we can find an embedding into a surface of genus $\gamma$ in linear time \cite{Mohar99,KawarabayashiMR08}.
However, in our algorithm we do not assume that the input graph is embedded into a surface of \emph{minimum} genus.
Therefore, if the input graph $G$ has genus $\gamma$, but we are only given an embedding into a surface of genus $g>\gamma$, the approximation guarantee of the drawing produced by our algorithm will depend on $g$.
This is in particular interesting when the genus of the graph $\gamma$ is super-constant, in which case computing an embedding of minimum genus becomes NP-hard \cite{genus-np-complete}.
\end{remark}

For any graph $H$, and an embedding $\tau$ of $H$ into a surface of genus at least 1, the \emph{nonseparating dual edge-width} of $(H, \tau)$, denoted by $\ndew(H,\tau)$, is the length of the shortest surface-nonseparating cycle in the dual of $H$, w.r.t. the embedding $\tau$.
We also write $\ndew(H)$ when $\tau$ is clear form the context.
We will use the following algorithmic result by Hlin\v{e}n\'{y} and Chimani~\cite{crossing_genus}.

\begin{theorem}[Hlin\v{e}n\'{y} \& Chimani \cite{crossing_genus}]\label{thm:genus_soda}
Let $G$ be a graph embedded in an orientable surface of genus $g\geq 1$, with $\ndew(G) \geq 2^{2g+2}\cdot \dmax$.
Then there is an efficient algorithm that computes a drawing of $G$ in the plane with at most $3\cdot 2^{3g+2} \cdot \dmax^2 \cdot \optcro{G}$ crossings.
\end{theorem}

Let $G$ be a graph, and let $\sigma$ be an embedding of $G$ into an orientable surface $S$ of genus $g$.
We begin by computing an integer ${\kappa}\in \{0,\ldots,g\}$, and a sequence of graphs $G_0,\ldots,G_{\kappa}$.
For each graph $G_i$ we also compute a drawing $\sigma_i$ of $G_i$ into a surface of genus $g-i$.
Initially, we set $G_0 = G$, and $\sigma_0 = \sigma$.
For each $i:0\leq i<g$, if $\ndew(G_i,\sigma_i)\geq 2^{2g+2}\cdot \dmax$, then we set ${\kappa}=i$, and we terminate the sequence $G_0,\ldots,G_{\kappa}$.
Otherwise, if $\ndew(G_i,\sigma_i) < 2^{2g+2} \cdot \dmax$, then 
we first compute a shortest surface-nonseparating cycle $C_i^*$ in the dual of $G_i$ w.r.t. the embedding $\sigma_i$. Such a cycle can be found in time $O(g^3 n \log n)$ using the
algorithm of Cabello and Chambers \cite{CabelloC07}.
We construct $G_{i+1}$ by removing from $G_i$ all edges whose duals are in $E(C_i^*)$.
We also construct an embedding of $G_{i+1}$ by cutting the surface into which $G_i$ is embedded along the cycle $C_i^*$.
This gives us an embedding $\sigma_{i+1}$ of $G_{i+1}$.
As observed in \cite{crossing_genus}, the graph $G_{i+1}$ is a spanning subgraph of $G_i$, and the embedding $\sigma_{i+1}$ is into a surface of genus $g-i-1$.

Let us define
\[
E^*_1 = E(G)\setminus E(G_{\kappa}).
\]
We have
\begin{align}
|E^*_1| = \sum_{i=0}^{{\kappa}-1} |E(C_i^*)| = \sum_{i=0}^{{\kappa}-1} \ndew(G_i, \sigma_i) < {\kappa} \cdot 2^{2g+2} \cdot \dmax \leq g\cdot 2^{2g+2} \cdot \dmax. \label{eq:genus1}
\end{align}

If ${\kappa}=g$, then the graph $G_{\kappa}$ is drawn into a surface of genus $0$, and therefore $G_{\kappa}$ is planar.
Otherwise, if ${\kappa}<g$, then we have a drawing $\sigma_{\kappa}$ of the graph $G_{\kappa}$ into a surface of genus $g-{\kappa}\geq 1$, and with
\[
\ndew(G_{\kappa}, \sigma_{\kappa})\geq 2^{2g+2}\cdot \dmax > 2^{2(g-{\kappa})+2}\cdot \dmax.
\]
This means in particular that we can run the algorithm from Theorem \ref{thm:genus_soda} to obtain a drawing $\phi$ of $G_{\kappa}$ into the plane with at most $3\cdot 2^{3(g-{\kappa})+2} \cdot \dmax^2 \cdot \optcro{G_{\kappa}}$ crossings.
Define the set $E^*_2\subseteq E(G_{\kappa})$ of edges to contain all edges participating in crossings in $\phi$.
We have
\begin{align}
|E^*_2| \leq |\cro_\phi(G_{\kappa})| \leq 3\cdot 2^{3(g-{\kappa})+2} \cdot \dmax^2 \cdot \optcro{G_{\kappa}} \leq 3\cdot 2^{3g+2} \cdot \dmax^2 \cdot \optcro{G}. \label{eq:genus2}
\end{align}

Let $E^*=E_1^*\cup E_2^*$.
Observe that the graph $G\setminus E^* = G_{\kappa}\setminus E_2^*$ is planar.
The assertion of Theorem~\ref{thm: bounded genus} follows trivially if the graph $G$ is planar, so we may assume that $\optcro{G}\geq 1$.
By \eqref{eq:genus1} and \eqref{eq:genus2} we therefore have
\begin{align}
|E^*| = |E^*_1| + |E^*_2| = 2^{O(g)} \cdot \dmax^2 \cdot \optcro{G}. \nonumber
\end{align}
Running the algorithm from Theorem \ref{thm:main} with the planarizing set $E^*$, we obtain a drawing of $G$ into the plane with at most $\dmax^{O(1)} \cdot |E^*| \cdot (\optcro{G} + |E^*|) = \dmax^{O(1)} \cdot 2^{O(g)} \cdot \optcrosq{G}$ crossings.

To obtain an $\tilde{O}\left(2^{O(g)} \cdot \sqrt{n}\right)$-approximation for bounded-degree graphs, 
run the above algorithm, and the algorithm of Even et al.~\cite{EvenGS02}, and output the drawing with fewer crossings.
\hfill \ensuremath{\Box}

\fi

\iffull
\section{Proof of Theorem~\ref{thm: Even: extension}}
Recall that the algorithm of Even et al.~\cite{EvenGS02} finds a drawing of a bounded degree graph with at most $O(\log^3 n) \cdot (n+\optcro{G})$
crossings. We first show that this algorithm can be extended to arbitrary graphs
to produce drawings with at most $\poly(\dmax) \log^3 n \cdot (n+\optcro{G})$ crossings, where $\dmax$ is the maximum vertex degree in $G$.
We then note that by using the approximation algorithm of Arora et al.~\cite{ARV} for {\sf Balanced Separator}
instead of the algorithm of Leighton and Rao~\cite{LR}, this guarantee can be improved to $\poly(\dmax) \log^2 n \cdot (n+\optcro{G})$.

\begin{lemma}  
Suppose that there is a polynomial time algorithm $\cal A$, that, for any $n$-vertex graph $G=(V,E)$ with vertex degrees at most 3, finds a drawing of $G$ with at most $\alpha (n + \optcro{G})$ crossings. 
Then there is a polynomial time algorithm that finds a drawing of any graph $G$ with at most $O(\dmax^4 \cdot \alpha\cdot  (n + \optcro{G}))$ crossings, where $\dmax$ is the maximum vertex degree in $G$.
\end{lemma}
\begin{proof}
The algorithm first constructs an auxiliary graph $\tilde G$ with maximum vertex degree $3$. Informally, $\tilde G$ is the graph obtained from $G$ by replacing every vertex $v$  with a path $P_v$ of length $\deg v$ (e.g., if $G$ is a $d$-regular graph then $\tilde G$ is the replacement product of $G$ and the path of length $d$). Formally, the vertices of $\tilde G$ are pairs $(u, e)$, where $u\in V$, $e\in E$ and $e$ is incident on $u$. The edges of $\tilde G$ consist of two subsets. First, for every edge $e=(u,v)\in E$, we connect the vertices $(u,e)$ and $(v,e)$ of $\tilde G$ with an edge $\tilde e$. We call such edges ``type 1 edges''. Additionally, for each $u\in V(G)$, if  $e_1, \dots, e_{\deg u}$ is the list of all edges incident
on $u$ (in an arbitrary order), then we connect every consecutive pair $(u, e_i),(u, e_{i+1})$ of vertices, for $1\leq i< \deg u$, with a type-2 edge. Let $P_u$ denote the resulting path formed by these edges. This completes the description of $\tilde{G}$. Note that $\tilde G$ has at most $\dmax \cdot n$ vertices, and every vertex of $\tilde G$ has degree at most 3. Also note that if we contract every path $P_u$  in $\tilde G$, for $u\in V(G)$, into a vertex,  we obtain the graph $G$.

We now bound the crossing number of $\tilde G$. Observe that we can obtain a drawing $\phi_{\tilde G}$ of $\tilde G$ from any drawing $\phi_G$ of $G$ as follows. We put each vertex $(u,e)$ on the drawing of the edge $e$ very close to the drawing of $u$. We draw each type-1 edge $\tilde e = ((u,e), (v,e))$ of $\tilde G$  along the segment of the drawing of $e$ in $\phi_G$, connecting the images of $(u,e)$ and $(v,e)$. We draw type-2 edges on the line segments connecting their endpoints. We now bound the number of crossings in this drawing. Notice that there are no crossings between the type-1 and the type-2 edges. The number of crossings between pairs of type-1 edges is bounded by the total number of crossings in $\phi_{G}$. Finally, in order to bound the number of crossings between pairs of type-2 edges, we notice that if $u\neq v$, then the edges of $P_u$ and $P_v$ do not cross. Any pair of edges on path $P_u$ may cross at most once, since any pair of line segments crosses at most once. Therefore, 
there are at most $\binom{\deg u}{2}$ crossings among the edges of the path $P_u$ for every vertex $u$. Overall, 
$$\cro_{\phi_{\tilde G}}(\tilde G) \leq \cro_{\phi_{G}}(G) + \sum_{u\in V} \binom{\deg u}{2}\leq  \cro_{\phi_{G}}(G)+n\dmax(\dmax-1),$$

and 

$$\optcro{\tilde G} \leq \optcro{G} + n\dmax(\dmax-1).$$

Our algorithm runs $\cal A$ on $\tilde G$ and finds a drawing $\phi'_{\tilde G}$ of $\tilde G$ with at most 
$$\alpha (\dmax n + \optcro{\tilde G})\leq \alpha (\dmax^2 n + \optcro{G}) $$ crossings.
We now show how to transform the resulting drawing $\phi'_{\tilde G}$ of $\tilde G$ into a drawing $\phi'_{G}$ of $G$. Informally, this is done by
contracting the drawing of every path $P_u$ into a point.
More precisely, we draw every vertex $u\in V(G)$ at the point $\phi'_{\tilde G}((u, e_1))$ (where $e_1$ is the first edge in the incidence list for $u$).
For each path $P_u$, for $u\in V(G)$, we construct $(\deg u)$ auxiliary curves $\gamma_{u, e_1}$, \dots, $\gamma_{u, e_{\deg u}}$, where for each $i: 1\leq i\leq \deg u$, curve $\gamma_{u,e_i}$ connects the images of $(u,e_1)$ and $(u,e_i)$ in $\phi'_{\tilde G}$,
and no pair of such curves cross (though curves that correspond to different vertices of $V$ are allowed to cross).
This is done as follows. First, we draw each curve $\gamma_{u,e_i}$  along the image of the path $P_u$ in $\phi'_{\tilde{G}}$,
following the segment that connects the images of $(u,e_1)$ and $(u,e_i)$; in a neighborhood of each vertex $(u,e_j)$ of $P_u$, we draw the curve  $\gamma_{u,e_i}$ on the side of $P_u$ opposite to the side where the edge $\tilde e_j$ enters $(u, e_j)$ (thus $\gamma_{u,e_i}$ does not cross the drawing of $e_j$ near the point 
$\phi'_{\tilde{G}}((u, e_j))$). We make sure that all curves $\gamma_{u,e_i}$ are drawn in general position. 
Next, if any of the resulting curves cross themselves, or cross each other, we perform uncrossing. Since all these curves start at the same point -- the image of $(u,e_1)$ -- we can uncross them so that the final curves do not cross each other, and do not cross themselves.

We are now ready to describe the drawing of every edge $e=(u,v)\in E(G)$. The drawing of $e$ is a concatenation of three curves: 
$\gamma_{u,e}$, $\phi'_{\tilde G}(\tilde e)$, and $\gamma_{v,e}$. The second segment of this drawing is called a type-1 segment, while the first and the third segments are called type-2 segments.
Note that it is possible that some pairs of edges have more than one crossing, adjacent edges cross each other and some edges have self crossings in this drawing; we will fix that later.
We now bound the number of crossings $\cro_{\phi'_G}(G)$.
\begin{itemize}
\item The number of crossings between all pairs of type-1 segments is bounded by $\cro_{\phi'_{\tilde G}}(\tilde G)$.
\item The number of crossings between any pair of type-2 segments $\gamma_{u,e}$ and $\gamma_{v,e'}$ (where $u\neq v$), is bounded by the number of crossings between
paths $P_u$ and $P_v$. Since every crossing between the paths $P_u$ and $P_v$ may pay for crossings of at most $\dmax^2$ such pairs of curves,
the total number of such crossings is at most $\dmax^2 \cro_{\phi'_{\tilde G}}(\tilde G)$.
\item Similarly, the number of crossings between a type-2 curve $\gamma_{u,e}$ and a type-1 curve $\phi'_{\tilde G}(\tilde e')$ (for an arbitrary edge $e'$ of $G$)
is at most the number of crossings between $P_u$ and $\tilde e'$ in $\phi'_{\tilde G}$. So the total number of such crossings is 
at most $\dmax \cdot \cro_{\phi'_{\tilde G}}(\tilde G)$.
\end{itemize}
We conclude that the number of crossings is
$O(\dmax^2 \cro_{\phi'_{\tilde G}}(\tilde G)) \leq O(\alpha \dmax^4(n + \optcro{G}))$.
Finally, the algorithm uncrosses drawings of edges that cross more than once, crossing pairs of adjacent edges, and edges that cross themselves.
During this step the number of crossings can only go down.
\end{proof}

The algorithm of Even et al.~\cite{EvenGS02} uses an algorithm for \textsf{Balanced Separator} as a subroutine.
We need a few definitions. Suppose we are given a graph $G=(V,E)$ with non-negative vertex weights $w$. For each subset $A\sse V$ of vertices, let $w(A)$ denote the total weight of vertices in $A$. We say that a cut $(S,\nots)$ is $b$-balanced w.r.t. the weights $w$, iff $w(S),w(\nots)\geq bw(V)$.
The cost of the cut is $|E(S,\nots)|$.
Even et al. prove the following theorem.

\begin{theorem}\label{lem:ARV-corollary}
Suppose that there is some function $\beta:{\mathbb Z}\rightarrow {\mathbb R}^+$, and an efficient algorithm for Balanced Separator, with the following property. Given any $n$-vertex, vertex-weighted graph $G = (V,E)$ and values $0\leq b\leq 1/2$, $C_b$, such that every sub-graph of $G$ has a $b$-balanced separator of size at most $C_b$, the algorithm returns a $1/3$-balanced cut of $G$, whose cost is $O(\beta(n) C_b)$ (the constant in the $O$-notation may depend on $b$).   Then there is an efficient algorithm to find a drawing of any bounded degree graph $G$ with $O(\beta^2(n) \log n) \cdot (n+\optcro{G})$ crossings.
\end{theorem}
Even et al. use the algorithm of Leighton and Rao~\cite{LR} that gives an algorithm for balanced separators with approximation factor $\beta(n) = O(\log n)$. We note that the results of Arora, Rao and Vazirani~\cite{ARV} gives an improved algorithm, with $\beta(n) = O(\sqrt{\log n})$, and thus we can efficiently find a drawing of a bounded degree graph with at most $O(\log^2 n) \cdot (n+\optcro{G})$ crossings. 
\begin{theorem}[Arora et al. \cite{ARV}]\label{thm:ARV-alg}
For every constant $0<b<1/2$ and some $0 < b' < b$ (that depends on $b$), there is a bi-criteria approximation algorithm for the Balanced Cut Problem with the following approximation guarantee. Given a graph $G=(V,E)$ and a set of vertex weights $w$, the algorithm finds a $b'$-balanced cut w.r.t. $w$ of cost at most $O(\sqrt{\log n} \cdot C_b)$, where $C_b$ is the cost of the optimal $b$-balanced cut w.r.t. $w$. (The constant in the $O$-notation depends on $b$.)
\end{theorem}
We point out that the algorithm in Theorem~\ref{thm:ARV-alg} does not directly satisfy the requirements of Theorem~\ref{lem:ARV-corollary}, since it finds a $b'$-balanced cut only for \textit{some} $b'\in (0,b)$ that depends on $b$, while we are required to produce a $1/3$-balanced cut.
For completeness, we show that the algorithm of~\cite{ARV} can still be used to obtain an algorithm for balanced separator as required in the statement of Theorem~\ref{lem:ARV-corollary}, for $\beta(n)=O(\sqrt{\log n})$.

We iteratively apply the algorithm of Arora et al. We start with $S = V$ and $\nots=\emptyset$. We first find a $b'$-balanced cut in $G[S] = G$ 
(where $b'$ is the constant guaranteed by Theorem~~\ref{thm:ARV-alg}). If the larger side of the cut contains at most $2/3$ of the total weight,
then the cut is $1/3$-balanced and we are done. Otherwise, we let $S$ be the larger side of the cut (w.r.t. weights $w$),
and we add the smaller side of the cut to $\nots$. Then we iteratively apply the ARV-algorithm to $G[S]$, update sets $S$ and $\nots$, and repeat.
We stop when $S$ contains at most $2/3$ of the total weight of $G$. Note that after each iteration the weight of $S$ decreases 
by at least a factor $(1-b')$, since the algorithm of \cite{ARV} finds a $b'$-balanced cut. Therefore, the algorithm stops in at most 
$\lceil \log_{1-b'} 2/3\rceil$ steps. Observe that after each iteration, $w(S) \geq w(V)/3$, since before each iteration
$w(S) \geq 2 w(V)/3$ and we let $S$ to be the larger of the two sides of the cut. Hence, $w(V)/3 \leq w(S) \leq 2w(V)/3$
when the algorithm terminates. That is, the cut $(S, \nots)$ is $1/3$-balanced. 

In every iteration, we cut at most $O(\sqrt{\log n} \cdot C_b)$ edges, and the total number of iterations is at most $\lceil\log_{1-b'} 2/3\rceil$.
Thus the cost of the cut is $O(\sqrt{\log n} \cdot C_b)$.

\fi
\end{document}